\documentclass[oribibl]{llncs}

\usepackage{amsmath}
\usepackage[dvips]{graphicx}
\usepackage{dsfont}
\usepackage{amssymb}
\usepackage{amstext}
\DeclareGraphicsExtensions{.eps}
\usepackage{pdfsync}
\usepackage{subfigure}

\input xy
\xyoption{all}

\usepackage{proof}

 \newtheorem{thm}{Theorem}

 \newtheorem{prop}{Proposition}
 \newtheorem{defn}{Definition}
 
 \newtheorem{exam}{Example}
 
 \newtheorem{alg}{Algorithm}

 \newcommand{\A}{\mathcal{A}}
 \newcommand{\B}{\mathcal{B}}
 \newcommand{\C}{\mathcal{C}}
 
 \newcommand{\E}{\mathcal{E}}
 
 \newcommand{\G}{\mathcal{G}}
 \newcommand{\J}{\mathcal{J}}

 \newcommand{\U}{\mathcal{U}}

\begin{document}

\title{Categorical semiotics}

\author{ Carlos Leandro}

\institute{
   \email{miguel.melro.leandro@gmail.com}\\ Departamento de Matem\'{a}tica, Instituto Superior de Engenharia de Lisboa,  Portugal. 
 }



\begin{abstract}
The integration of knowledge extracted from different models described by domain experts or from models generated by machine learning algorithms is strongly conditioned by the lack of an appropriated framework to specify and integrate structures, learning processes, data transformations and data models or data rules. In this work we extended algebraic specification methods to be used in this  type of framework. This methodology uses graphic structures similar to Ehresmann's sketches \cite{Adamek94}  interpreted on a fuzzy set universe. This approach takes advantages of the sketches ability to integrate data deterministic and nondeterministic structures. Selecting this strategy we also try to take advantage on how the graphic languages, used in Category theory in general and used for sketch definition in particular, are suited to reasoning about problems, to structural description and to task specification and task decomposition.
\end{abstract}

\maketitle

\section*{Introduction}

A model is a system of sets with relations providing
constrains upon the set system. A class of models, all similarly
structures, together with the structure preserving maps between them
is a category. For instance a relational database schema can be
viewed as a specification of a class of systems of sets having the
category of models defined by all the database states and
transformations between states. The relational database schema give
constrains on the state of the database.

The core of an information
system is a set of databases and sets of data transformations,
usually taking the form of workflows. In the modern view database
presents an internal model of real world fragment, and the
transformations offer different ways for construct views of this
fragment and integrate the different aspects of it. A
crucial step for the proper information system design is to specify
the universe and its views in abstract and formalized terms suitable
for semantic refinement, such as be used for low-level system
specification, and able to be used on the specification improvement
though the introduction of new knowledge about the data stored on
the system during its life time. This type of data structure
specification is called semantic modeling. It has to compress
information and process description into a comprehensible way
suitable for communication between database tools or designs, such
as between data mining processes and data analysts. The usual choice
is to use graphic languages, and indeed a great effort has been put
in the development of graphic denotational systems. The history of
graphic notations invented in various scientific and engineering
disciplines is rich. In last years one can observe a great diversity
of visual modeling languages and methods: ERdiagrams and a lot of
their dialects, OOA\&D-schemas in a million of versions and UML which itself comprises a host of various
notations. Our goal is to clarify the basic semantic foundations of
graph languages and present an integrated framework where different
languages and its semantics can be approached consistently and
integrated.

A good graphic language should be formalizable, sufficiently
expressive to capture all the pecualities of the real word, and must
be suitable for semantic refinement. We are particularly interested
in use the same language to model both deterministic and nondeterministic involved structures; for instance data structures and
models generated using machine learning algorithms. In our opinion the best
approach in terms of expressiveness and formalization  to
deterministic graphic specification is Category theory.

Category theory generalized the use of graphic language to specify
structures and properties through diagrams. These categorical
techniques provide powerful tools for formal specification,
structuring, model construction, and formal verification for a wide
range of systems, presented on a grate variety of papers. The data
specification requires finite, effective and comprehensive
presentation of complete structures, this type of methodology was
explored on Category theory for algebraic specification by
Ehresmann. He developed sketches as a specification methodology of
mathematical structures and presented it as an alternative to the
string-based specification employed in mathematical logic. The
functional semantic of sketches is sound in the informal sense that
it preserves by definition the structure given in the sketch. The
analogy to the semantics of traditional model theory is close enough
that sketches and their models fit the definition of "institution"
(see \cite{goguen83}), which is an abstract notion of a logic system
having syntactic and semantic components. The soundness of semantics
appears trivial contrasting with the inductive proof of soundness
that occur in string-based logic because the semantics functor is
not defined recursively. Sketch specification enjoy a unique
combination of rigor, expressiveness and comprehensibility. They can
be used for data modeling, process modeling and metadata modeling as
well thus providing a unified specification framework for system
modeling. However sketch structure forces us to take a global
perspective of the system. It makes impossible decomposing a
specification problem in subproblems. This makes difficult to
specify a large system as the interaction between subsystems or
components, which is a common practice, in engineering. We can give
a better view to this problem by means of a typical application, the
specification of workflows (for more details see
\cite{Baldanall05}). A workflow describes a business process in
terms of tasks and shared resources. Such descriptions are needed,
for example, when interoperability of the workflows of different
organizations is an issue, for example, when applications of
different enterprises are to be integrated over the Internet
\cite{Baldanall05}. A workflow is a net \cite{Aalst98}, some times a
Petri net, satisfying some structural constraints and the
corresponding soundness conditions. Usually the methodology used to
specify workflows have associated a library of components. An
interorganizational workflow is modeled as a set of such workflow
nets \cite{Aalst99} connected through additional places for
asynchronous communication and synchronisation requirements on
transitions. The difficulty of applying sketch to this type of task
results of the way composition is defined on the the category of
sets. We want to interpret a workflow as a arrow, where the gluing
of different workflows must be also interpreted as an arrow and must
be always defined. The net structure defined by a workflow is called
an oriented multi-graph since its components are relations linking
together families defined by sets of entities or data structures.

For our goal we extend the syntax of sketch to multi-graphs and
define its models on a class of fuzzy relations (see \cite{Johnstone02}). Where multi-graphs nodes are interpreted as relations. To extend the
syntax of a sketch we began by formalize a library as the syntactic
structure of admissible configurations using components on that
library. This approach is based on the notion of component, used to
define relationships between two families of entities, and which can
be organized following an hierarchy of complexity. A
component if not atomic is defined by the plugging of other
components. We see the set of admissible configurations as the
graphic language defined by the library structure. A model of a
library is a map from the library multi-graph to the structure associated to the class of relations defined using a multi-valued logic. An interpretation for an
admissible configuration, is defined for each library model. It is
the limit of the multi-graph, in the category of $\Omega$-sets, resulting of applying a library model to a
admissible multi-graph. It seems to be the adequate framework for the definition and the study of graphic-based logics, if we interpret one of its node as the set of truth values.

We used libraries as a way to define the lexicon of the language used on the description of a domain. The category defined by library models and natural transformation aren't a accessible category (see \cite{Adamek94}) it can't be axiomatized by a basic theory in first-order logic. By this we mean what classic Ehresmann sketches not have sufficient expressive power to specify the category of library models defined in a multi-valued logic with more than three truth values.

To be able to formalize linguistic, Chomsky in \cite{chomsky57}
propose a language as a set of grammatically correct sentences
possible in the language. The goal of defining a language is then to
characterize the set of grammatical sentences explicitly, by means
of a formal grammar. The two main categories of grammar are that of
\emph{generative grammar}, which are sets of rules for how elements
of a language can be generated, and that of \emph{analytic
grammars}, which are sets of rules for how a structure can be
analyzed to determine whether it is a member of the language. In
this sense, our approach to the definition of a graphic language
based on libraries of components uses a multi-graph as the language
analytic grammar.

A generative grammar does not in any way correspond to the algorithm
used to parse the generated language. Analytic grammar corresponds
more directly to the structure and the semantic of a parser for the
language. Examples of analytic grammars formalisms include top-down parsing language (TDPL)\cite{Birman}, link
grammars \cite{Sleator91} and parsing expression grammars \cite{Ford04}.

Our extension to the syntax of sketch is based on Link grammars. A
theory of syntax proposed by Davy Temperley and Daniel Slator in
\cite{Sleator91} which builds relations between pairs of words,
rather than constructing constituents in a tree-like hierarchy. They
are similar to dependency grammars developed by Lucien Tesni\`{e}re
in the 60's \cite{Tesniere59}. Link grammar is a form of analytic grammar designed for linguistics, which derives syntactic structures by  examining the
positional relationships between pairs of words. This can be seen in
the model that the Davy Temperley and Daniel Slator provided
for English in this system: Their grammar deals with most of the
linguistic phenomena in English. Informally, a sentence is correct
in this system if it is possible to link all words according to the
links needed by each word, defined in the lexicon. Link represents
syntactic relations. It is shown in \cite{Sleator91} that link
grammars have the expressive power of context-free grammars.

We call sign system to the extension of the Ehresmann sketch notion.
A formal definition for an identical notion appeared is presented by Goguen in
\cite{Gog??}, in our version it is a library defined by a
multi-graph more, just as sketch, a set of commutative multi-diagrams, a
set of limit cones and a set of colimit cocones. In this context the limit and the colimit for a multi-diagram must be seen as a relation defined on a fuzzy set theory. We consider a semiotic system as a pair defined by a
sign system and one of its models, and they can be seen as an institution in Goguen sense \cite{goguen83}.

Our goal was to develop a mathematically precise theory of semiotics based on Ehresmann sketch notion. This structure try to internalized the formalization made by the \emph{Vienna Circle} in their \emph{International Encyclopedia of Unified Science}, by breaking out the filed, which they called "\emph{Semiotic}", into three branches: \emph{Semantics} (relation between signs and the things they refer to) \emph{Syntactics} (Relation of signs to each other in formal structures) and \emph{Pragmatics} (Relation of signs to their impacts on those processes which use them).

Signs appear as members of sign systems. Must signs are complex objects constructed from others, lower level signs. In this sense they can be seen as structures defined using signs and sign systems capture the systematic structure of signs. Marking an entity with a sign can be seen as a way to named entities or assign properties to entities, such as two entities marked with the some signs are identical. We defined the model of a sign system a consistent process for marking entities belonging to a fixed universe, preserving labeled relations between signs.

Following the spirit used by Peirce on \emph{Semiotics} we identify entities from the universe of discourse with some of this signs. We do this
by marking some of the objects or morphisms of a topos with
signs used on library specification using a library model. If, in the topos, the object classifier have a monoidal logics structure, and its operators are marked with signs from a library, we can use this library to specify monoidal
graphic logics. And a relation on a semiotic system must be seen as a
configuration having by interpretation a multi-morphism with target
a object marked as having associated a monoidal logic structure.
This allows the extension of the concepts used in logics to the
graphic logics associated to a semiotic system. We can use graphic
relations to define queries on the semiotic. An answer for it is a
fibre on the source, of its interpretation, defined by all the
"points" transformed in the true for the semiotic associated
ML-algebra. In this context a "set" of points is consistent with a
graphic relation if every point is associated by the relation
interpretation to the true of the semiotic monoidal logic. Since
monoidal logics can be seen as fuzzy logics or multi-valued logics,
logics definable in this framework are in reality fuzzy graphic
logics. This extends the logic used for Ehresmann sketches. And
allows making fuzzy the notion of a relation evaluation, an
important issue for practical applications.

Day living activities generate information which can be stored in information systems spread by different databases.  A query to a information system can be seen as a view of the data stored in the system, and it is presented by a dataset. Many times this information is useful to produce new knowledge about the reality. Given the amount of data stored this transformation must be automatized and this is the goal of fields of AI, like Machine Learning, see \cite{Michell86}.

In fuzzy set theory a dataset can be expressed by a relation. The interpretation of a word in a logic semiotic is called a model for the dataset if the relation is the multi-diagram limit. If we see a dataset as a way of codifying data, we can seen its model as a way of represent the knowledge in the graphic language, associated to a fixed logic semiotic. This relation between data and knowledge, is only useful, if we define the notion of structure $\lambda$-consistent with a relation. Where $\lambda$ is a logic value quantifying the degree of similarity between interpretation of a word and the concept defined codified in the dataset. In this sense the fact of a dataset be $\lambda$-consistent with a diagram catches the idea of approximation.

If the data, presented in a dataset, is $\lambda$-consistent with a set of diagrams we call  semiotic defined by this set of diagrams a theory $\lambda$-consistent with the data. Naturally this can be extended to databases. The knowledge available about a database can be codified in a semiotic and the quality of this knowledge is given by the way its model describe good approximations to data, associated to tables what can be, generated in a database state. Since different human specialists or machine learning algorithms express the extracted knowledge using, many times, different languages: the problem of knowledge integration is the problem of semiotic integration. The objective of integration is then to construct one semiotic that exploits all the knowledge that is available and as good performance. We describe methodologies for merging several separate theories. However this processes some times also requires the integration of languages and the associated logics which is simplified following our approach.

\textbf{Overview of the paper:}

We began section \ref{monoidal logics} describing a partially-ordered monoidal structure
for the set of truth-values used on the definition of our graphic
logics. The language used in this logics is based on possible
circuit configurations using a libraries of components, structure
defined on section \ref{specifying libraries}. As a framework for the definition of models
for this libraries we used a class of relations defined between sets and evaluated in a multi-valued logic, which are described in
section \ref{multi-morphisms}. For that we need to define composition of relations compatible with circuit gluing. In this sense composition must be seen as a total operator in the class of relations, relaxing the diagram equality, evaluated in a multi-valued logic, allowing in section \ref{multidiagrams}, the presentation of a generalized version for the notion of commutative diagrams. This is explored  in section \ref{bayesian inference}, on the definition of a version for Bayesian inference on fuzzy logics. In
section \ref{modeling libraries} the language defined by a libraries is seen as a set of circuits closed to the plug-in operation, every word is a string of component labels or signs and define a relation
between a family of inputs requirements and a family of output
structures, both identified using families of signs. We present how
libraries are modeled on the class of relations on section \ref{modeling libraries}. The descriptive power of languages defined using libraries allow defining structure what are not definable using first order basic theory. This is presented in section \ref{descritive power} by showing what the category defined using library models is not accessible. Accessible categories are known to be specified using Ehresmann sketch. We toke advantage of its specification power by enriching the structure of a library
with a structure similar to a Ehresmann sketch, in section \ref{Sings Semiotics}, we called to this specification tool a sign system or a
specification system. This enrichment is made using multi-diagrams instead of diagrams specified respecting library constrains and where limits and colimits are interpret as multi-morphisms and used in section \ref{multidiagrams}, on the definition of diagram commutativity evaluation on a multi-valued logic.  A specification system where we fixed a model we called a semiotic system, with this and an example we finish section \ref{Sings Semiotics}. On section \ref{logics} we use
signs systems to specify fuzzy logics. They are semiotics with special structure, for that we impose interpretations for some of the
sign system signs allowing the interpretation or words as evaluations of relations in a monoidal logic.  When some signs
are interpreted as ML-algebras operators the library was called a logic
library and the associated language was called a logic language. We formalize this concepts and describe when a diagram
defines a relation and an equation. However, some problems in Mathematics require more expressive languages than the ones defined using libraries. We improve the expressive power of libraries Lagrangian syntactic operators. An example is describe on section \ref{synopt}  allowing the definition of Differencial semiotics.  On section \ref{Concept description} we present how to
evaluate relations in a semiotic, and use it to define what we mean by
the level of consistence of a relation. This notion is extended to
relations $\lambda$-consistent with words in a semiotic. We emphasize the idea what a diagram defining a relation can be seen as a query to the semiotic. In this context a
$\lambda$-answer is a structure where the query is
$\lambda$-consistent. This notions are used on section \ref{fuzzy computability} the definition of
bottom and upper presentation to a structure $A$ in the a semiotic;
the bottom presentation is the lower structure in $A$ codified in
the semiotic language, and the upper presentation is the short
structure containing $A$ and codified in the fixed language. These
notion can be seen as two approximations to the concept, following the spirit of Pawlaks' of the standard version of rough set theory. They
define an interior and a closure operators for structures, allowing
the definition of a formal topology.  Which we use on the definition of a inference system for words evaluation on a semiotic system based on
properties of ML-algebras. What is made for descriptions can be made for relations evaluated in a multi-valued logic. In section \ref{fuzzy computability} a fuzzy relation is computable in a semiotic if it can be seen as a interpretation for a diagram defined in the associated language.  Section \ref{integration} is dedicated to the the integration of semiotics and models for concepts. The goal is construct one system that exploits all the knowledge that is available, allowing improve concept description by combining different description for the concept on the same concept possibility expressed using different languages. For specification reasons a string-based modal logic is present, on section \ref{reasoning}, where propositional variables are  interpreted as diagrams defined on a semiotic. This intents to be a meta-language to reasoning about models of concepts and knowledge.

\section{Monoidal logics}\label{monoidal logics}
Fuzziness is the rule than the exception in practical problems. A lot of research is being done on fuzzy sets and the associated fuzzy logics; we are specially interested in the possibility of extending the data specification paradigm using Ehresmann sketches to the fuzzy case. This paper is motivated to the introduction to $\Omega$-Categories given in \cite{Clementino04} and to $\Omega$-Sets given in \cite{valverde??}.

Ulrich H\"{o}hle introduced Monoidal Logic in 1995 in order to give a common framework to several first order non-classical logics, such as Linear logic, Intuitionistic logic and Lukasiewicz logic. A Monoidal Logic is a Full Lambek calculus with  exchange and weakening. We supposed what problems involving the specification of structures, using libraries of components, can be formulate in a framework defined by a set theory with a monoidal logic.

Recall that a algebra $(\Omega,\otimes,\leq,1)$  is a partially-ordered monoid
if $(\Omega,\otimes,1)$ is a monoid and  $\leq$  is a partial order
on $\Omega$  such that the operator $\otimes$ is monotone
increasing; i.e.
\[ x\leq x' \text{ and } y\leq y' \text{ imply } x\otimes y \leq x'\otimes y'.\]

An algebra $(\Omega,\otimes,\setminus,/,\leq,1)$  is a resituated
partially-ordered monoid if $(\Omega,\otimes,\leq,1)$ is a
partially-ordered monoid and moreover the following condition is
satisfied for all $x,y,z\in \Omega$;
\[ x\otimes y\leq z \Leftrightarrow y\leq x\setminus z \Leftrightarrow x\leq z/y. \]
This condition is called the \emph{law of residuation}, and $/$ and
$\setminus$ are called the right and left residual of $\otimes$,
respectively.

Any residuated partial-ordered monoid $\Omega$  such that
$(\Omega,\leq)$ forms a lattice and $(\Omega,\otimes)$  has a unit
it is called a \emph{residuated lattice}. More precisely, an algebra
\[(\Omega,\vee,\wedge,\otimes,\setminus,/,1)\]  is a residuated
lattice if
\begin{enumerate}
  \item $(\Omega,\otimes,1)$ is a monoid such that $\setminus$ and $/$ are the right and the left residual
of  $\otimes$, respectively, and
  \item $(\Omega,\vee,\wedge)$ is a lattice.
\end{enumerate}
When $\otimes$ is commutative, we call it a \emph{commutative residuated
lattice}. In any commutative residuated lattice, $x\setminus y=y/x$
hold for all $x,y$. In such a case, we use the symbol $\Rightarrow$
and write $x\Rightarrow y$ instead of  $x\setminus y$ (and of
$y/x$). Also the commutative residuated lattice is denoted by
$(\Omega,\vee,\wedge,\otimes,\Rightarrow,1)$.

\begin{defn}
A \emph{ML-algebra} is a bounded commutative residuated lattice where
$1=\top$, formally, is a system
$(\Omega,\otimes,\Rightarrow,\vee,\wedge,\perp,\top)$ satisfying:
\begin{enumerate}
  \item $(\Omega,\otimes,\top)$ is a commutative monoid,
  \item $x\otimes\top=x$ for every $x\in \Omega$,
  \item $(\Omega,\vee,\wedge,\perp,\top)$ is a bounded lattice,
  and
  \item the residuation property holds, \[\text{for all }x,y,z\in
  \Omega, x\leq y \Rightarrow z \text{ iff } x\otimes y\leq z.\]
\end{enumerate}
\end{defn}
In this paper we assume that ML-algebra $\Omega$ is non-trivial, i.e. $\top\neq\bot$.

A structure equivalente to a ML-algebra is presented in \cite{Clementino04} as a commutative and unital quantale where $\Omega$ is a complete lattice equipped with a symmetric and associative tensor product $\otimes$, with unit $\top$ and with right adjoint $\Rightarrow$. Considered $\Omega$ as a thin category, $\Omega$ is said to be symmetric monoidal-closed.

Logics having as models refinements of ML-algebras are called
\emph{monoidal logics}. In many-valued logics, such as fuzzy logics,
$\otimes$ is the standard truth degree function for conjunction
connective. Since operator $\otimes$ is monotone and have right adjoint, we have:

\begin{prop}
On a ML-algebra one has
\begin{enumerate}
  \item if $y\leq z$ then $x\otimes y\leq x\otimes z$,
  \item $x\leq y$ iff $(x \Rightarrow y)=\top$, and
  \item $x\Rightarrow z=\bigvee\{y:x\otimes y\leq z\}$.
\end{enumerate}
\end{prop}

And,

\begin{prop}\cite{klawonn??}\label{prop:implic}
In any ML-algebra the following equalities hold, for all  $x,y,z\in
\Omega$,
\begin{enumerate}
  \item $x\otimes(x\Rightarrow y)\leq x \wedge y$, and
  \item $(x\Rightarrow y)\otimes(y\Rightarrow z)\leq x\Rightarrow
  z$.
\end{enumerate}
\end{prop}

Every non-trivial Heyting algebra - with $\otimes=\wedge$ and $\top$ the top element - is an example of a ML-algebra, in particular the two element chain $2=\{false<true\}$ with the monoidal structure given by "and" and "true".

The complete real half-line $P=[0,\infty]$, with the categorical structure induced by the relation $\leq$ admits several interesting monoidal structures. If $\otimes=\wedge=\max$ it is a Heyting algebra. Another possible choice of $\otimes$ is +, note that in this case the right adjoint $\Rightarrow$ is given by truncated minus: $x\Rightarrow y=\max\{v-u,0\}$.

\begin{exam}[t-norm based fuzzy logic]
A t-norm is a function $\otimes$ used to define a ML-algebra structure on the real unit interval $[0,1]$. We may define a monoidal logic using a t-norm by taken the unite interval as the set of truth values and where the residuum of $\otimes$ is defined as the operation $x\Rightarrow y=\max\{z|x\otimes z\leq y\}$. The other truth function considered important in fuzzy logic are \emph{weak conjunction} $x.y=\min(x,y)$ and \emph{weak disjunction} $x+y=\max(x,y)$. However in the following we interpret a fuzzy logic on a ML-algebra $([0,1],\otimes,\Rightarrow,\vee,\wedge,0,1)$ where $\otimes$ is continuo t-norm and $\Rightarrow$ is its residuum, the lattice structures are given by $x\vee y=\min(x,y)$ and $x\vee y=\max(x+y,1)$.
The following are importante examples of fuzzy logics interpreted in ML-algebras defined by specific continuous t-norms:
\begin{itemize}
  \item \textbf{{\L}ukasiewicz} logic defined using the t-norm $x\otimes y=\max(0,x+y-1)$ and its residuum \[x\Rightarrow y=\min(1,1-x+y)\]
  \item \textbf{G\"{o}del logic} defined using the t-norm $x\otimes y=\min(x,y)$ and its residuum \[x\Rightarrow y= \left\{
                                      \begin{array}{cl}
                                        1 & \text{ if } x\leq y \\
                                        y & \text{ otherwise}\\
                                      \end{array}
                                    \right.
      \]
  \item \textbf{Product logic} defined using the t-norm $x\otimes y=x.y$ and its residuum \[x\Rightarrow y= \left\{
                                      \begin{array}{cl}
                                        1 & \text{ if } x\leq y \\
                                        y/x & \text{ otherwise}\\
                                      \end{array}
                                    \right.\]
\end{itemize}
\end{exam}

Particularly important to this work are the basic logics, with have by instances ML-algebras with are divisible, i.e. such that
\[
x\otimes(x\Rightarrow y)= x \wedge y.
\]
Examples of this type of logic are classic boolean logic and fuzzy logic like product, G\"{o}del and {\L}ukasiewicz. Note that must of the examples presented in the following are construct using product logics with the natural order in interval [0,1].

For the sequel we define
\[
a\Leftrightarrow b := (a\Rightarrow b)\otimes (b\Rightarrow a) \text{ and }
\neg a := a\Rightarrow \bot.
\]

Let
\[
(\Omega_i,\otimes_i,\Rightarrow_i,\vee_i,\wedge_i,\perp_i,\top_i)_{i\in I}
\]
be a finite family of ML-algebra. The product of this ML-algebras is the ML-algebra
\[
(\Pi_{i\in I}\Omega_i,\otimes,\Rightarrow,\vee,\wedge,\perp,\top)
\]
such that
\begin{enumerate}
  \item $\Pi_{i\in I}\Omega_i$ is the cartesian product of sets of truth values,
  \item $(\lambda_1,\lambda_2,\ldots,\lambda_n)\otimes(\alpha_1,\alpha_2,\ldots,\alpha_n)=(\lambda_1\otimes\alpha_1,\lambda_2\otimes\alpha_2,\ldots,\lambda_n\otimes\alpha_n)$,
  \item $(\lambda_1,\lambda_2,\ldots,\lambda_n)\Rightarrow(\alpha_1,\alpha_2,\ldots,\alpha_n)=(\lambda_1\Rightarrow\alpha_1,\lambda_2\Rightarrow\alpha_2,\ldots,\lambda_n\Rightarrow\alpha_n)$,
  \item $(\lambda_1,\lambda_2,\ldots,\lambda_n)\vee(\alpha_1,\alpha_2,\ldots,\alpha_n)=(\lambda_1\vee\alpha_1,\lambda_2\vee\alpha_2,\ldots,\lambda_n\vee\alpha_n)$,
  \item  $(\lambda_1,\lambda_2,\ldots,\lambda_n)\wedge(\alpha_1,\alpha_2,\ldots,\alpha_n)=(\lambda_1\wedge\alpha_1,\lambda_2\wedge\alpha_2,\ldots,\lambda_n\wedge\alpha_n)$,
  \item $\bot=(\bot_1,\bot_2,\ldots,\bot_n)$, and
  \item $\top=(\top_1,\top_2,\ldots,\top_n)$.
\end{enumerate}
This structure has associated two types of morphisms. The projections
\[
\pi_j:\Pi_{i\in I}\Omega_i\rightarrow\Omega_j,
\]
and the upper interpretations
\[
\begin{array}{cccl}
  \top_j: & \Omega_j & \rightarrow & \Pi_{i\in I}\Omega_i\\
          & \alpha_j & \mapsto & (\bot_1,\bot_2,\ldots,\alpha_j,\ldots,\bot_n)
\end{array}
\]
Note what upper interpretation is the right inverse to projection,
\[
\pi_j.\top_j=id_{\Omega_j}.
\]
We will use this structure as a vehicle for integration of ML-logics.

\section{Multi-morphisms}\label{multi-morphisms}

A surprising result discovered in Category Theory, presented by M. Makkai in \cite{Makki89}, is that the arrow specification language is absolutely expressive, in the sense that any construction having a formal semantic meaning can be described in the arrow language as well. Moreover, if basic object of interest are described by arrows then normally it turned out that many derived objects of interest can be also derived by arrows in a quite natural way \cite{Diskin99}. To define the universe we are going to deal with it is necessary and sufficient to define what we mean by a morphisms between objects of the universe.

Our universe for semantic modeling must be "essentially the same" as $Set$ but able to represent soft structures specified thought monoidal logics. Let $\Omega$ be a set with a ML-algebra structure $(\Omega,\otimes,\Rightarrow,\vee,\wedge,\perp,\top)$. We use as universe $Set(\Omega)$ defined using $Set$  having by entities $\Omega$-sets, i.e. sets $A$ furnished with a $\Omega$-valued map
\[
[\cdot=\cdot]:A\times A\rightarrow\Omega,
\]
which is symmetric and transitive in the sense that both
\[
[a=b]=[b=a]\text{ and } [a=b]\otimes[b=c]\leq [a=c],
\]
hold for all $a,b,c\in A$. This is called a \emph{similarity} in $A$. We will use Greek letters to denote  $\Omega$-sets, we write $\alpha:A$, to mean a $\Omega$-sets defined by set $A$ and a similarity $[\cdot=\cdot]_\alpha$, and it is interpreted as a relation evaluated in $\Omega$ or a distribution in $A\times A$. The diagonal of this fuzzy relation is used on definition of fuzzy sets with support $A$. For each $\Omega$-set $\alpha:A$ and $a\in A$ we define
\[[a]_\alpha=[a=a]_\alpha,\]
and called it the \emph{extend} of $a$. Then $[\cdot]_\alpha:A\rightarrow\Omega$ is a representation for the fuzzy set $\alpha$ codified thought similarity $[\cdot=\cdot]_\alpha$. An element $a$ is called \emph{global} in $\alpha:A$ if $[a]_\alpha=\top$.

Note that every set $A$ have a natural structure of $\Omega$-set defined by the equality $=$ in $A$, i.e. having by similarity
\[
[a=b]_A=\left\{
        \begin{array}{cc}
          \top & \text{ if }a=b \\
          \bot & \text{ if }a\neq b \\
        \end{array}
      \right..
\]The crisp similarity $[a=b]_A$, defined by the equality in $A$, is denoted by $1_A$.

Entities belonging to a $\Omega$-set $\alpha:A$ are characterized by a set of attributes $(A_i)_I$ if $A=\Pi_{i\in I}A_i$. Given $\bar{x}\in\Pi_{i\in I}A_i$, on the description of $\bar{x}$ many of the values associated to some of this attributes are "non-observable" or unknown. In this sense we will differentiate between two types of attributes: \emph{observable attributes} and \emph{non-observable attributes}. Let  $(A_i)_{i\in L}$ be a set of observable attributes in $A$, where $L\subseteq I$. We define an observable $\Omega$-set of $\alpha:A=\Pi_{i\in I}A_i$ as the  $\Omega$-set $\beta:B=\Pi_{i\in L}A_i$ such that
\[[\bar{a}=\bar{b}]_\beta=\bigvee_{\bar{x}=(\bar{c},\bar{a}),\bar{y}=(\bar{d},\bar{b})\in A}[\bar{x}=\bar{y}]_\alpha.\]

\begin{defn}[Observable description]
If $\alpha:A$ is a $\Omega$-set with a set of observable attributes $(A_i)_{i\in L}$, we call to every $\bar{a}\in \Pi_{i\in L}A_i$ an \emph{observable description} for an entity in $\alpha$.
\end{defn}

We define a multi-morphism in $Set(\Omega)$ as a tracking morphism between $\Omega$-sets $\alpha:A$ and $\beta:B$ as a map
\[
f:A\times B \rightarrow \Omega
\]
(usually called a $\Omega$-map or a $\Omega$-matrix in \cite{Clementino04}).
 If $f$ is a multi-morphism between $\alpha:A$ and $\beta:B$ in $Set(\Omega)$ we write $f:A\rightharpoonup B$ to identify $A$ as the source of $f$ and $B$ as the target of $f$. And if $\bar{a}$ and $\bar{b}$ are observable descriptions for entities in $\alpha:A$ and $\beta:B$ respectively we define
 \[f(\bar{a},\bar{b})=\bigvee_{\bar{x}=(\bar{c},\bar{a}),\bar{y}=(\bar{d},\bar{b})}f(\bar{x},\bar{y}).\]
 The complete partial order on the ML-algebra $\Omega$ induces a complete partial order on the set of multi-morphisms. Given two multi-morphisms between $\Omega$-sets $\alpha:A$ and $\beta:B$ in $Set(\Omega)$ \[f,g:A\times B \rightarrow \Omega\] we write $f\leq g$ if $f(a,b)\leq g(a,b)$, for every $(a,b)\in A\times B$.
 Graphically a multi-morphism \[f:A_0\times A_1\times A_2\rightharpoonup A_3\times A_4\times A_5\] is present in fig. \ref{multiarrow} by a multi-arrow
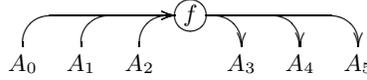
\begin{figure}[h]
 \[
 \small
\xymatrix @=5pt {
&&&*+[o][F-]{f}\ar `r[rd][rd]\ar `r[rrd][rrd]\ar `r[rrrd][rrrd]&&&\\
 A_0\ar `u[urrr][urrr]&A_1\ar `u[urr][urr]& A_2\ar `u[ur][ur]&&A_3&A_4&A_5
 }
\]
\caption{Multi-arrow.}\label{multiarrow}
\end{figure}
having by sources $A_0, A_1$ and $A_2$ and by targets $A_3, A_4$ and $A_5$.

We classify multi-morphisms preserve entity evaluation in $\Omega$:
 \begin{defn}[Total multi-morphism]\label{total}
 A multi-morphism $f:A\rightharpoonup B$ is \emph{total} in $\alpha:A$ if \[[a]_\alpha=\bigvee_b f(a,b),\]
 for every $a\in A$.
 \end{defn}

 \begin{defn}[Faithful multi-morphism]\label{total}
 A multi-morphism $f:A\rightharpoonup B$ is \emph{faithful} in $\beta:B$ if \[[b]_\beta=\bigvee_a f(a,b),\]
 for every $b\in B$.
 \end{defn}

 Note what, for every $\Omega$-set $\alpha:\prod_{i\in I}A_i$ we can use the similarity diagonal to define a multi-morphism by selecting a set of sources $(A_s)_{s\in S}$ sets and a set of targets $(A_t)_{s\in T}$ sets, with disjoint indexes, i.e. such that $S\cap T=\emptyset$. This multi-morphism is given by a map
 \[
 g(\bar{x},\bar{y},\bar{z})=\bigvee_{\bar{y}\in I\setminus(S\cup T)}[\bar{x},\bar{y},\bar{z}]_{\prod_{i\in I}A_i},
 \]
 for every $\bar{x}\in\prod_{s\in S}A_s$ and $\bar{z}\in\prod_{t\in T}A_t$, which defines
 \[ g:\prod_{s\in S}A_s\rightharpoonup\prod_{t\in T}A_t.\]

 Composition of multi-morphisms  is defined as matrix multiplication, the composition of
\[
f:A\rightharpoonup B \text{ and } g:B\rightharpoonup C,
\]
as the multi-morphism
\[ f\otimes g:A\rightharpoonup B, \]
given by
\[
(f\otimes g)(a,c)=\bigvee_b(f(a,b)\otimes g(b,c)).
\]
Note what if $f$ and $g$ are total and faithful then $f\otimes g$ is total and faithful. This composition have by identity for a $\Omega$-set $\alpha:A$ the multi-morphism \[1_A=[\cdot=\cdot]_\alpha:A\rightharpoonup A,\] defined by the equality in $A$, since for $f:A\rightharpoonup B$ we have $1_A\otimes f = f\otimes 1_B$.

\begin{prop}\label{prop:comprestriction} In $Set(\Omega)$ let $f:A\rightharpoonup A$ be a multi-morphism such that $1_A\leq f$. If in the ML-algebra $\Omega$ for every truth value $\alpha$, $\alpha\otimes\alpha\leq\alpha$, then \[1_A\leq f\otimes f\leq f.\] And, when the logic have more than two truth values, i.e. if $|\Omega|>2$, we have
\[f\otimes f = 1_A \text{ iff }f=1_A.\]
\end{prop}

The set of multi-morphisms defined between $\Omega$-sets $\alpha:A$ and $\beta:B$ is denoted by $Set(\Omega)|A,B|$. And, every map $f:A\rightarrow B$ in $Set$ defines a multi-morphism, with source $A$ and target $B$, given by
\[
\chi_f:A\times B\rightarrow\Omega \text{ where } \chi_f(a,b)=
\left\{
  \begin{array}{l}
    \top \text{ if } f(a)=b\\
    \bot \text{ if } f(a)\neq b\\
  \end{array}
\right..
\]
In this sense the hom-set $Set[A,B]$, of morphism between $A$ and $B$ in $Set$, define a subset of $Set(\Omega)|A,B|$. To keep notation simple in the sequel we will write $f:A\rightarrow B$ rather than $f:A\rightharpoonup B$ for the multi-morphism induced by a map. Then $f:A\rightarrow B$ defines a total multi-morphism from $\alpha:A$, having by similarity the equality, and it is a faithful multi-morphism to $\beta:B$, having by similarity
\[
[a=b]=
\left\{
  \begin{array}{l}
    \top \text{ if } a=b \text{ and }a\in f(A)\\
    \bot \text{ other wise}\\
  \end{array}
\right..
\]

The formula for multi-morphism composition became considerably easier if one of the multi-morphism is a set-map. For maps $f:A\rightarrow B$ and $g:B\rightarrow C$ and multi-morphisms $r:A\rightharpoonup B$ and $s:B\rightharpoonup C$ we have
\[
(f\otimes s)(a,c)=s(f(a),c),\text{ and }(r\otimes g)(a,c)=\bigvee_{b\in g^{-1}(c)}r(a,b).
\]
The operator of multi-morphism composition can be extended to multi-morphisms which are not composable, in the usual sense. Let
\[
f:A\rightharpoonup X\times W \text{ and } g:B\times X \rightharpoonup C,
\]
then we define
\[ f\otimes g:A\times B\rightharpoonup W\times C,\]
given by
\[
(f\otimes g)(a,b,w,c)=\bigvee_x(f(a,x,w)\otimes g(b,x,c)).
\]
In particular if $f:A\rightharpoonup B$, $g:C \rightharpoonup D$ and $B\neq C$ then \[f\otimes g:A\times C\rightharpoonup B\times D,\] is given by \[(f\otimes g)(a,c,b,d)=f(a,b)\otimes g(c,d).\] This reflects the independence between entities in $B$ and $C$ and we define:

\begin{defn}[Independence]
Two multi-morphisms $f$ and $g$ are called independent if
\[f\otimes g=g\otimes f.\]
\end{defn}

\begin{exam}(Keys in a relational database)\label{def:indexproduct}
The relational model for database management is a database model based on predicate logic and set theory. The fundamental assumption of the relational model is that data is represented as mathematical $n$-ary relations, an $n$-ary relation being a subset of the Cartesian product of $n$ domains. In the usual mathematical model, reasoning about such data is done in two-valued logic or three-valued logic. Data are operated upon by means of a relational calculus or relational algebra.

The relational model of data permits the database designer to create a consistent, logical representation of information. Consistency is achieved by including declared constraints in the database design, which is usually referred to as the logical schema.

A weight table $R$ in a database defined using attributes $(A_i)_I$ is a map in a ML-algebra
\[
R:\prod_{i\in I}A_i\rightarrow\Omega.
\]
in this sense a weight table is a $\Omega$-set $\alpha:\prod_{i\in I}A_i$.

Every weight table $R:A\times B\rightarrow\Omega$, may be describe as the multi-morphism $R:A\rightharpoonup B$, and can be decomposed using two weight tables
\[
D_0:A\times K\rightarrow\Omega
\]
\[
D_1:K\times B\rightarrow\Omega
\]
such that
\[
R=D_0\otimes D_1.
\]
In this case we call to $K$ a set of \emph{keys}, between $D_0$ and $D_1$, and write
\[D_0\otimes_KD_1\]
to denote the joint of $D_0$ and $D_1$ using the keys in $K$.

Generically, if $K_1,K_2,\ldots,K_n$ are sets of keys between $D_0$ and respectively $D_1,D_2,\ldots,D_n$ we write
\[
D=D_0\otimes_{K_1,K_2,\ldots,K_n}(D_1\otimes\ldots\otimes D_n),
\]
to denote the joint product
\[
D=(\ldots(((D_0\otimes_{K_1}D_1)\otimes_{K_2}D_2)\otimes_{K_3}\ldots)\otimes_{K_n}D_n),
\]
or
\[
D=D_0\otimes_{K_1}D_1\otimes_{K_2}D_2\otimes_{K_3}\ldots\otimes_{K_n}D_n.
\]
When the family $(K_i)$ of keys is defined by the same set $K$ the joint product is called the $K$ indexed product of $D_0,D_1,\ldots,D_n$ and denoted by
\[D=D_0\otimes_{K}D_1\otimes_{K}D_2\otimes_{K}\ldots\otimes_{K}D_n.\]
In this case $D$ is called the $K$-\emph{indexed product} of $D_0,D_1,D_2,\ldots,D_n$.
\end{exam}

Given the importance of multi-morphism composition in this work lets formalize that we mean by the multi-morphism composition:
\begin{defn}[Multi-morphism composition]\label{def:composition}
Given multi-morphisms $f$ and $g$ defined by $\Omega$-maps
\[
f:\prod_{i\in I(f)}A_i\rightarrow \Omega \text{ and } g:\prod_{j\in I(g)}B_j\rightarrow \Omega
\]
where  for every $i\in I(f)$ and $j\in I(g)$, $i=j$ iff $A_i=B_j$. Without selection of sources and targets sets for $f$ and $g$ we define
\[
(f\otimes g)(\bar{x},\bar{y})=f(\bar{x})\otimes g(\bar{y}),
\]
for every $\bar{x}\in\prod_{i\in I(f)}A_i$ and $\bar{y}\in\prod_{j\in I(g)}B_j$. However, if we select sets of sources $S(f)\subset I(f)$ and $S(g)\subset I(g)$, and sets of targets $S(f)\subset I(f)$ and $S(g)\subset I(g)$, such $S(f)\cap T(f)=\emptyset$ and $S(g)\cap T(g)=\emptyset$, we define
\[
(f\otimes g)(\bar{x},\bar{y})=\bigvee_{\bar{z}\in \prod_{i\in T(f)\cap S(g)}A_i}f(\bar{x},\bar{z})\otimes g(\bar{z},\bar{y}),
\]
for every $\bar{x}\in\prod_{i\in S(f)}A_i\times \prod_{j\in S(g)\setminus T(f)} B_j$ and $\bar{y}\in\prod_{i\in T(f)\setminus S(g)}A_i\times\prod_{j\in T(g)}B_j$.
\end{defn}
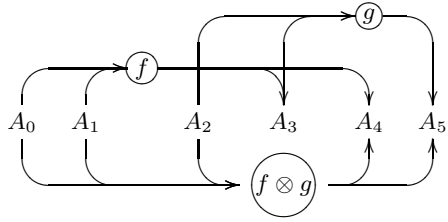
\begin{figure}[h]
\[
\small
\xymatrix @=7pt {
&&&&&*+[o][F-]{g}\ar `r[rdd][rdd] &\\
&&*+[o][F-]{f}\ar `r[rrd][rrd]\ar `r[rrrd][rrrd] &&&&\\
 A_0\ar `u[urr][urr]\ar `d[drrrr][drrrr]&A_1\ar `u[ur][ur]\ar `d[drrr][drrr]& &A_2\ar `u[uurr][uurr]\ar `d[dr][dr]&A_3\ar `u[uur][uur]&A_4&A_5\\
 &&&&*++[o][F-]{f\otimes g}\ar `r[ru][ru]\ar `r[rru][rru]&&\\
 }
\]
\caption{Multi-morphism composition.}\label{multimorphism composition}
\end{figure}

The transpose $f^\circ:B\rightharpoonup A$ of a multi-morphism $f:A\rightharpoonup B$ is defined by $f^\circ(b,a)=f(a,b)$. It is easy to see that
\[
(\cdot)^\circ:Set(\Omega)|A,B|\rightarrow Set(\Omega)|B,A|
\]
is order preserving and
\[
[\cdot=\cdot]^\circ=[\cdot=\cdot],\;\; (f\otimes g)^\circ=g^\circ\otimes f^\circ\text{ and } {f^\circ}^\circ=f.
\]
In this sense, if $f$ is a multi-morphism having by set of sources $\A$ and by set of targets $\B$ then $\B$ is the set of sources for $f^\circ$ and $\A$ its set of targets.
\begin{figure}[h]
\[
\small
\xymatrix @=7pt {
&&&*+[o][F-]{f}\ar `r[rd][rd]\ar `r[rrd][rrd]\ar `r[rrrd][rrrd]&&&\\
 A_0\ar `u[urrr][urrr]&A_1\ar `u[urr][urr]& A_2\ar `u[ur][ur]&&A_3\ar `d[dl][dl]&A_4\ar `d[dll][dll]&A_5\ar `d[dlll][dlll]\\
 &&&*+[o][F-]{f^\circ}\ar `l[lu][lu]\ar `l[llu][llu]\ar `l[lllu][lllu]&&&
 }
\]
\caption{Transpose.}\label{multimorphism:transpose}
\end{figure}
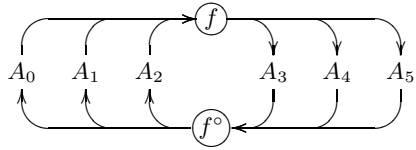

We classify multi-morphisms by its ability to preserve its domain or codomain truth values distribution.

\begin{defn}\label{def:isomorphism}
A multi-morphism $f:A\rightharpoonup B$ is an \emph{epimorphism} between $\alpha:A$ and $\beta:B$ if
\[f^\circ\otimes\alpha\otimes f=\beta.\]
It is a \emph{monomorphism} between $\alpha:A$ and $\beta:B$ when
\[\alpha =f\otimes\beta\otimes f^\circ.\]
Naturally, when the two conditions are valid $f$ is called an \emph{isomorphism} between $\Omega$-objects $\alpha:A$ and $\beta:B$.
\end{defn}

For each set-map $f:A\rightarrow B$ we have
\[
[\cdot=\cdot]_A=1_A\leq f\otimes f^\circ \text{ and } f^\circ\otimes f\leq 1_B=[\cdot=\cdot]_B, \]
i.e. $f$ is left adjoint to $f^\circ$, $f\dashv f^\circ$. If $f:A\rightharpoonup B$ is a multi-morphism and $1_A = f\otimes f^\circ$ and $f^\circ\otimes f= 1_B$ the multi-morphism $f$ is called \emph{orthogonal}.

In general given multi-morphisms
$f:A\rightharpoonup B$ and $g:B\rightharpoonup A$ we say that $f$ is the \emph{left adjoint} to $g$ for $\alpha:A$ and $\beta:B$ if
\[
[\cdot=\cdot]_\alpha\leq f\otimes g \text{ and } g\otimes f\leq [\cdot=\cdot]_\beta.\]

The tensor product on $\Omega$ can be naturally transported to $\Omega$-sets. More precisely, for $\Omega$-sets $\alpha:A$ and $\beta:B$, we denote by $\alpha\otimes \beta$ the $\Omega$-sets defined using the Cartesian product $A\times B$ in $Set(\Omega)$ and furnished with
\[
[(a_1,b_1)=(a_2,b_2)]_{\alpha\otimes \beta}=[a_1=a_2]_\alpha\otimes[b_1=b_2]_\beta.
\]
Then, for each $\Omega$-set $\alpha:A$, the functor
\[
\alpha\otimes\cdot:Set(\Omega)\rightarrow Set(\Omega),
\]
has a adjoint the hom functor $(\cdot)^\alpha:Set(\Omega)\rightarrow Set(\Omega)$ defined by  $\beta^\alpha=Set(\Omega)[\alpha,\beta]$ with the similarity given by
\[
[f=g]_{\beta^\alpha}=\bigwedge_{a\in A}\bigwedge_{b\in B}(f(a,b)\Leftrightarrow g(a,b)).
\]
and, for every $f\in Set(\Omega)[\beta,\gamma]$,
\[
f^\alpha = Set(\Omega)[\alpha,f]:Set(\Omega)[\alpha,\beta]\rightarrow Set(\Omega)[\alpha,\gamma],
\] such that $(f^\alpha)(g)=g\otimes f$.

Being monoidal-closed $\Omega$ has a natural structure as $\Omega$-set given by
\[
[x=y]_\Omega=(x\Leftrightarrow y)=(x\Rightarrow y)\otimes(y\Rightarrow x).
\]

Given similarities  $[\cdot=\cdot]_\alpha:A\rightharpoonup A$ and
$[\cdot=\cdot]_\beta:B\rightharpoonup B$. If we sets $A$ and $B$ are distinct,  applying composition definition we have  \[[\cdot=\cdot]_\alpha\otimes[\cdot=\cdot]_\beta:A\times B\rightharpoonup A\times B,\] given by \[([\cdot=\cdot]_\alpha\otimes[\cdot=\cdot]_\beta)(a_1,b_1,a_2,b_2)=[a_1=a_2]_\alpha\otimes[b_1=b_2]_\beta,\] and it is a similarity relation defining the $\Omega$-object  $\alpha\otimes\beta:A\times B$. More generically we define:

\begin{defn}[Product of $\Omega$-sets]\label{ProdSimil}
Given $\Omega$-sets $\alpha:A$ and $\alpha:B$ we define the product of $\alpha\otimes\beta$ as the $\Omega$-sets $\alpha\otimes\beta:A\times B$ given by
\[[\cdot=\cdot]_{\alpha\otimes\beta}:A\times B\times A\times B\rightharpoonup \Omega,\] such that
\[[(a_1,b_1)=(a_2,b_2)]_{\alpha\otimes\alpha}=[a_1=a_2]_{\alpha}\otimes[b_1=b_2]_{\beta}.\]
\end{defn}

By the transitivity imposed on the definition of similarity we have,   \[[\cdot=\cdot]_{\alpha\otimes\beta}=[\cdot=\cdot]_\alpha\otimes[\cdot=\cdot]_\alpha\leq[\cdot=\cdot]_\alpha.\]

\section{Bayesian inference in a basic logic}\label{bayesian inference}
The presented definition for multi-morphism composition $\otimes$ is compatible to the Bayes' theorem used on Bayesian inference when  $Set(\Omega)$ logic is a basic logic.

\begin{prop}[Bayes Rule]
Let $\Omega$ be a divisible ML-algebra. Given a faithful and total multi-morphism ${f:A\rightharpoonup B,}$ and observable descriptions $a$ and $b$ of entities in $\alpha:A$ and $\beta:B$, respectively. The equations
\begin{enumerate}
  \item $[a]_\alpha\otimes f(\beta|a) = f(a,\_)$ and
  \item $[b]_\beta\otimes f(\alpha|b) = f(\_,b)$,
\end{enumerate}
have solution, and they define $\Omega$-maps $f(\alpha|b):A\rightarrow \Omega$ and $f(\beta|a):B\rightarrow \Omega$, given by $f(\beta|a)=[a]_\alpha\Rightarrow f(a,\_)$ and $f(\alpha|b)=[b]_\beta\Rightarrow f(\_,b)$.
\end{prop}

\begin{proof}
In a divisible ML-algebra $\Omega$ we have $x\otimes(x\Rightarrow y)=x\wedge y$. Since $f$ is faithful $[a]_\alpha=\bigvee_cf(a,c)\geq f(a,c)$, then $[a]_\alpha\geq f(a,\_)$. Because $[a]_\alpha\wedge f(a,\_)=f(a,\_)$ we have
\[
[a]_\alpha\otimes([a]_\alpha\Rightarrow f(a,\_))=[a]_\alpha\wedge f(a,\_)=f(a,\_).
\]
And, we can use the same strategy to proof  $f(\alpha|b)=[b]_\beta\Rightarrow f(\_,b)$.
\end{proof}

We will interpret the $\Omega$-map $f(\beta|a)$ as a classifier in $B$, defined by relation $f$, for an entity described by $a$ using the basic monoidal logic $\Omega$.

Applying in the multi-morphism context the principles of Bayes inference: For faithful and total multi-morphisms $f:A\rightharpoonup B$ and $g:B\rightharpoonup C$, and $\Omega$-sets $\alpha:A$, $\beta:B$, $\gamma:D$, we have
\begin{center}
\begin{tabular}{rcl}
  $[a]_\alpha\otimes (f\otimes g)(\gamma|a)(c)$ & = & $(f\otimes g)(a,c)$ \\
   & = & $\bigvee_b f(a,b)\otimes g(b,c)$\\
   & = & $\bigvee_b [a]_\alpha\otimes f(\beta|a)(b)\otimes g(b,c)$, \\
\end{tabular}
\end{center}
then
\[
(f\otimes g)(\gamma|a)(c)=[a]_\alpha\Rightarrow ([a]_\alpha\otimes\bigvee_b  f(\beta|a)(b)\otimes g(b,c)),
\]
i.e.
\[(f\otimes g)(\gamma|a)=\bigvee_b f(\beta|a)(b)\otimes g(b,\_),\]
since in a divisible ML-algebra $\Omega$ we have $x\Rightarrow(x\otimes y)=x\wedge y$.

When $f:A\rightharpoonup C$ and $g:B\rightharpoonup D$ are independent we have
\[(f\otimes g)(\gamma\otimes \delta|a,b)(c,d)= f(\gamma|a)(c) \otimes g(\delta|b)(d).\]
Naturally, if $C=D$ we write $(f\otimes g)(\delta|a,b)(d)$ for $(f\otimes g)(\delta\otimes \delta|a,b)(d,d)$. And in this case we interpret the classifier  $(f\otimes g)(\delta|a,b)$ as the combination of two classifiers $f(\delta|a)$ and $g(\delta|b)$, defining entities of $\delta:D$ described by $a$ and $b$.

\begin{exam}[Binding]
Drugs are typically small organic molecules that achieve their desired
activity by binding to a target site on a receptor. The first step in
the discovery of a new drug is usually to identify and isolate the
receptor to which it should bind, followed by testing many small
molecules for their ability to bind to the target site. This leaves
researchers with the task of determining what separates the active
(binding) compounds from the inactive (non-binding) ones.  Such a
determination can then be used in the design of new compounds that not
only bind, but also have all the other properties required for a drug
(solubility, oral absorption, lack of side effects, appropriate duration
of action, toxicity, etc.).

The DuPont Pharmaceuticals provided a data set to KDD Cup 2001 consisting of 1909 compounds tested for their ability to bind to a target site on thrombin, a key receptor in blood clotting. Of these compounds, 42
were active (bind well) and the others were inactive. Each compound is
described by a single feature vector, of observable descriptions, comprised of a class value (A for
active, I for inactive) and 139,351 binary features, which describe
three-dimensional properties of the molecule. The definitions of the
individual bits were not included, they can be seen as not observable descriptions of a compound - we didn't know what each individual
bit means, only that they are generated in an internally consistent
manner for all 1909 compounds. Biological activity in general, and
receptor binding affinity in particular, correlate with various
structural and physical properties of small organic molecules. The task
proposed on KDD Cup 2001 by DuPont Pharmaceuticals  was to determine which of these three-dimensional properties are critical in this case and to learn to accurately predict the class value of a new compound.

Let $S$ be the set of available compounds and suppose what the process of compound classification in laboratory evacuate proposition "Compound $a$ is active" in the a fuzzy logic $\Omega=[0,1]$. A classification of each compound as active or inactive may be seen as a multi-morphism, in $Set([0,1])$,
\[
c:S\rightharpoonup\{A,I\}
\]
where $c(a,I)$ and $c(a,A)$ define the truth value of proposition "Compound $a$ is inactive" and "Compound $a$ is active", respectively in $\Omega=[0,1]$. Each compound in $S$ is described by a set of observable three-dimensional characteristics, measured in laboratory processes, and codified on the dataset. The similarity between compound must be codified by a $\Omega$-set $\alpha:S$. In the case of $\bar{x}$ be an observable three-dimensional structure of a compound in $\alpha:S$ it can be seen as a compound structure generalization. A description $x$ describe a class of compounds $\alpha:S$ and we defined the truth value of proposition "Compounds satisfying description $\overline{x}$ are active" by
\[
c(\beta|\overline{x})(A),
\]
where the $\Omega$-set $\beta:\{A,I\}$ codify the similarity between the stat of a compound be "active" or "inactive".
\end{exam}

In this example the best description $\bar{x}$ for an active compound in $S$, can be seen as the description that maximizes $c(\beta|\bar{x})(A)$. However from this notion emerges the need of have a way to codify observable descriptions and the existence of a framework for the selection of a best observable description. In the sequel we give a process to describe entities based on a graphic language. Words in this language are used to codifying relations between observable characteristics of entities in a multi-valued logic.

\section{Multi-diagrams}\label{multidiagrams}

A \emph{multi-diagram} in $Set(\Omega)$ is a multi-graph homomorphism $D:\G\rightarrow Set(\Omega)$ defined by mapping the multi-graph vertices to $\Omega$-sets and multi-arrows to multi-morphisms in $Set(\Omega)$.

Formally, if the multi-graph $\G$ is defined using nodes $(v_i)_{i\in L}$ and by a family of multi-arrows $(a_{IJ})$, where the multi-arrow $a_{IJ}$ have by source \[\{v_i:i\in I\subset L\},\] and by target \[\{v_j:j\in J\subset L\}.\] A multi-graph homomorphism $D:\G\rightarrow Set(\Omega)$ transform every node $v_i$ in a $\Omega$-sets $D(v_i)$ and each multi-arrow  \[a_{IJ}:\{v_i:i\in I\subset L\}\rightharpoonup\{v_j:j\in J\subset L\},\] in a multi-morphism \[D(a_{IJ}):\prod_{i\in I}D(v_i)\rightharpoonup\prod_{j\in J}D(v_j).\]

The usual definition of limit in $Set$ for a diagram can be extended to multi-diagrams in $Set(\Omega)$. For that we must see category $Set$ as the topos $Set({false,true})$, where ${false,true}$ define a two element chain with the monoidal structure given by the logic operator "and" and "true". Recall that the limit for a diagram or multi-diagram $D$ is defined as a $\{false,true\}$-set, denote by $Lim\; D$, which is a subobject of a cartesian product defined by the diagram vertices (see \cite{maclane71} or \cite{Borceux94}). We use these relation on the limit extension to multi-diagrams in $Set(\Omega)$. Since the cartesian product of $\Omega$-sets $(\alpha_i:A_i)$ was defined in \ref{ProdSimil} as the $\Omega$-set $\otimes_i\alpha_i:\prod_iA_i$ given by
\[
[\cdot=\cdot]_{\otimes_i\alpha_i}=\bigotimes_i[\cdot=\cdot]_{\alpha_i},
\]
we define
\begin{defn}[Limit of a multi-diagram]\label{lim}
Let $D:\G\rightarrow Set(\Omega)$ be a multi-diagram where $\G$ have by vertices $(v_i)_{i\in L}$. Its limit $Lim\;D$ is a subobject of the multi-diagram vertices cartesian product:
\[
Lim\; D \leq \prod_{i\in L} M(v_i)
\]
given by
 \[
Lim\; D: \prod_{i\in L} M(v_i) \rightarrow \Omega
\]
such that
\[
(Lim\;D)(\bar{x}_1,a_{i},\bar{x}_2,a_{j},\bar{x}_3)=[(\bar{x}_1,a_{i},\bar{x}_2,a_{j},\bar{x}_3)]_{\prod_{i\in L}M(v_i)}\otimes\bigotimes_{f:v_i\rightharpoonup v_j\in \G} D(f)(a_{i},a_{j}).
\]

\end{defn}
We see a limit as the result of applying the pattern used on the definition of each multi-morphism in the cartesian product of its vertices. This definition satisfies the usual universal property when the object classifier used $Set(\Omega)$ is the two-element chain, $2=\{false,true\}$, with the monoidal structure given by "and" and "true". In other words this definition coincide with the classical one on the context of classical logic.

Note what we can see the limit for a multi-diagram  $D$ as a multi-morphism by selecting a set of source $\Omega$-sets and a set of targets. The canonical multi-morphism associated to a multi-diagram $D$ have by source $s(D)$ the union of sources used to define the diagram multi-morphisms and have by target $t(D)$ the union of targets of $D$ multi-arrows.

\begin{exam}
By definition the multi-diagram $D$ given bellow
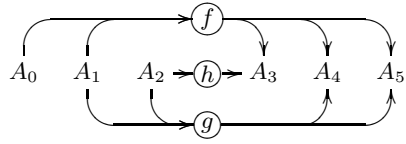
\begin{figure}[h]
\[
\small
\xymatrix @=7pt {
&&&*+[o][F-]{f}\ar `r[rd][rd]\ar `r[rrd][rrd]\ar `r[rrrd][rrrd]&&&\\
 A_0\ar `u[urrr][urrr]&A_1\ar `u[urr][urr]\ar `d[drr][drr]& A_2\ar `d[dr][dr]\ar[r]&*+[o][F-]{h}\ar[r] &A_3&A_4&A_5\\
 &&&*+[o][F-]{g}\ar `r[rru][rru]\ar `r[rrru][rrru]&&&
 }
\]
\caption{Multi-diagram.}\label{multidiagram1}
\end{figure}
have by limit the $\Omega$-map
\[
Lim\;D:A_0\times A_1\times A_2\times A_3\times A_4\times A_5\rightarrow \Omega
\]
given by
\[
_{(Lim\;D)(a_0,a_1,a_2,a_3,a_4,a_5)=[a_0,a_1,a_2,a_3,a_4,a_5]\otimes f(a_0,a_1,a_3,a_4,a_5)\otimes g(a_1,a_2,a_4,a_5)\otimes h(a_2,a_3).}
\]
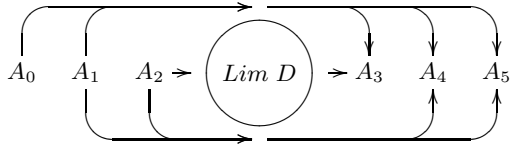
\begin{figure}[h]
\[
\small
\xymatrix @=7pt {
&&&\ar `r[rd][rd]\ar `r[rrd][rrd]\ar `r[rrrd][rrrd]&&&\\
 A_0\ar `u[urrr][urrr]&A_1\ar `u[urr][urr]\ar `d[drr][drr]& A_2\ar `d[dr][dr]\ar[r]&*++++[o][F-]{Lim\;D}\ar[r]&A_3&A_4&A_5\\
 &&&\ar `r[rru][rru]\ar `r[rrru][rrru]&&&
 }
\]
\caption{Multi-diagram limit functionality.}\label{multidiagram2}
\end{figure}
\end{exam}
The limit of a multi-diagram collapses the diagram into a multi-morphism by internalizing all the interconnections, thus delivering a multi-diagram as a whole.

In this sense the equalizer of a parallel pair of multi-morphisms $R,S:X\rightharpoonup Y$ is defined by
\[Lim(R=S):X\times Y\rightarrow \Omega\]
where
\[
Lim(R=S)(x,y)=[x,y]\otimes R(x,y)\otimes S(x,y).
\]

And, the pullback of $R:X\rightharpoonup U$ and $S:Y\rightharpoonup U$ is the multi-morphism
\[Lim(R\otimes_U S):X\times U\times Y\rightarrow \Omega\]
where
\[
Lim(R\otimes_U S)(x,u,y)= [x,u,y]\otimes R(x,u)\otimes S(y,u).
\]

Given a discrete multi-diagram  $D$ its limit is denoted by $\Pi_vD(v)$ given by
\[
\Pi_vD(v)(\bar{x})=[\bar{x}]_{\bigotimes_iD(i)},
\]
i.e. when $\bar{x}=(x_1,x_2,\ldots,x_n)$, $\Pi_vD(v)(\bar{x})=[x_1,x_2,\ldots,x_n]$.

The presented definition for limit simplifies the proof of:

\begin{prop}[Existence of limit in $Set(\Omega)$]
Every multi-diagram $D{:\G\rightarrow Set(\Omega)}$ have limit, i.e. exists a multi-morphism $f\leq \prod_{v\in\G}D(v)$ such that ${Lim\;D=f}$.
\end{prop}

But more interesting is the fact what we can show the opposite for basic logics:

\begin{prop}
If $\Omega$ is a divisible ML-algebra, then for every $\Omega$-map \[g:A_0\times A_1\times\ldots\times A_n\rightarrow\Omega\] and $\Omega$-map $(\alpha_i:A_i)$ such that \[g(x_1,x_2,\ldots,x_n)\leq [x_1,x_2,\ldots,x_n]\] there is a multi-diagram $D:\G\rightarrow Set(\Omega)$ such that \[Lim\;D=g\]
\end{prop}

We may proof this just by showing that multi-diagram
\[
f:A_0\rightharpoonup A_1\times\ldots\times A_n
\]
where \[f(x_1,x_2,\ldots,x_n)=[x_1,x_2,\ldots,x_n]\Rightarrow g(x_1,x_2,\ldots,x_n)\]
have by limit $g$.


We can see a multi-diagram as a way to express dependencies between classes of entities. When the limit of a diagram $D$ is an $\Omega$-object $\alpha$ we want to see the diagram as way to codify $\alpha$ or a model for $\alpha$. But for this type of relationship be useful we need to have a language to codify the diagram structure. Define this languages is one of the goals for this work.

Let $D$ be a multi-diagram with vertices $(v_i)$, having by arrow interpretations faithful and total multi-morphisms and let $a_i$ be an observable descriptions of an entity in $D(v_i)$. The limit in $Set(\Omega)$ of $D$, where $\Omega$ is a divisible ML-algebra, defines the classifier \[(Lim\;D)(D(v_1)|a_2,\ldots,a_n)\] such that
\[
[a_2,\ldots,a_n]\otimes(Lim\;D)(D(v_1)|a_2,\ldots,a_n)=\bigotimes_{f:v_i\rightharpoonup v_j\in \G} [a_i]\otimes D(f)(D(v_{j})|a_{i})(a_j),
\]
i.e.
\[
(Lim\;D)(D(v_1)|a_2,\ldots,a_n)=[a_2,\ldots,a_n]\Rightarrow\bigotimes_{f:v_i\rightharpoonup v_j\in \G} [a_i]\otimes D(f)(D(v_{j})|a_{i})(a_j),
\]
which can be seen as the combination of classifiers related through diagram $D$ to predicted $D(v_1)$. This expression is simplified when in $D$ we don't have multi-arrows with source $v_1$, we have
\[
(Lim\;D)(D(v_1)|a_2,\ldots,a_n)=\bigotimes_{f:v_i\rightharpoonup v_j\in \G} D(f)(D(v_{j})|a_{i})(a_j).
\]

We see the limit of diagram as a generalization for multi-morphism composition of a chain of composable multi-morphisms. This interpretation allows the definition of a semantic for circuits, when we assume a dependence between execution of circuit componentes. This point of view is also used to extend the classic notion of commutative diagram to fuzzy structures. For that we assume that a set of vertices $s(D)$ was selected in a diagram $D$. The set $s(D)$ is  called the set of \emph{sources} for diagram $D$.

\begin{defn}[Commutativity of multi-diagrams]
Let $D$ be a multi-diagram where we select $s(D)$ as set of sources. If $V$ is the cartesian product defined by all the vertices of $D$ not in $s(D)$. The multi-diagram $D$ is commutative for $s(D)$ if \[\bigvee_{\bar{n} \in V}(Lim\;D)(\bar{s},\bar{n})=\bigvee_{\bar{n}\in V}(\prod_i\;D(i))(\bar{s},\bar{n}),\] for every $\bar{s}\in \prod_{i\in s(D)}D(i)$. It is $\lambda$-commutative if  \[\left(\bigvee_{\bar{n}\in V}(Lim\;D)(\bar{s},\bar{n})\Leftrightarrow \bigvee_{\bar{n}\in V}(\prod_i\;D(i))(\bar{s},\bar{n})\right)\geq\lambda,\]for every $\bar{s}\in \prod_{i\in s(D)}D(i)$.
\end{defn}

In other words, a multi-diagram is commutative if the multi-morphism defined by its limit, with the selected sources, is total.

\begin{exam} Lets $Set([0,1])$ defined by the product logic, $\mathds{R}$ be the set of real number and $\oplus$ be a relation defined by the multi-morphism $\oplus:\mathds{R}\times\mathds{R}\rightharpoonup\mathds{R}$ given by Gaussian function  \[\oplus(x,y,z)=e^{-\frac{(z-x-y)^2}{2}}.\] The diagram $D$, presented in fig. \ref{equation1}, with sources $\alpha_0:\mathds{R}$ and $\alpha_1:\mathds{R}$,
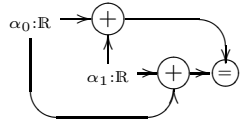
\begin{figure}[h]
\[
\small
\xymatrix @=7pt {
&&&&\\
 _{\alpha_0:\mathds{R}}\ar[r]\ar `d[ddrr]`r[rru][drr] &*+[o][F-]{+}\ar `r[rrd][rrd] &&&\\
&_{\alpha_1:\mathds{R}}\ar[u]\ar[r] &*+[o][F-]{+} \ar [r]&*+[o][F-]{_{=}}\\
&&&&\\
}
\]
\caption{Multi-diagram $D$ codifying $\alpha_0+\alpha_1=\alpha_1+\alpha_0$.}\label{equation1}
\end{figure}
where $=$ is defined as equality in $\mathds{R}$, is commutative for every $x_0,x_1\in\mathds{R}$, when we have by densities in $\alpha_0$ and $\alpha_1$, \[\alpha_0(x,y)=e^{-\frac{(x-x_0)^2}{2}-\frac{(y-x_0)^2}{2}}\text{ and } \alpha_1(x,y)=e^{-\frac{(x-x_1)^2}{2}-\frac{(y-x_1)^2}{2}}.\]
Because, using the definition presented to the multi-diagram limit, we have
\[
\begin{array}{rcl}
  (Lim\;D)(x,y,w) & = & \oplus(x,y,w)\otimes\oplus(y,x,w)\otimes[x,y,w]\\
                  & = & \oplus(x,y,w)\otimes\oplus(y,x,w)\otimes[x]\otimes[y]\otimes[w] \\
                  & = & e^{-\frac{(w-x-y)^2}{2}}.e^{-\frac{(w-y-x)^2}{2}}.e^{-\frac{(x-x_0)^2}{2}-\frac{(x-x_0)^2}{2}}.e^{-\frac{(y-x_1)^2}{2}-\frac{(y-x_1)^2}{2}}.1 \\
                  & = & e^{-(w-x-y)^2-(x-x_0)^2-(y-x_1)^2}
\end{array}
\]
then, since $e^{-(w-x-y)^2}\leq 1$, we have $e^{-(w-x-y)^2-(x-x_0)^2-(y-x_1)^2}\leq e^{-(x-x_0)^2-(y-x_1)^2}$, and
\[
\begin{array}{rcl}
  \bigvee_w(Lim\;D)(x,y,w) & = & e^{-(x-x_0)^2-(y-x_1)^2} \\
                           & = & [x]_{\alpha_0}\otimes[y]_{\alpha_1} \\
                           & = & \bigvee_w\alpha_0(x,x)\otimes\alpha_1(y,y)\otimes=(w,w).
\end{array}
\]
This proofs the commutativity for diagram $D$ when its sources have the fixed distributions. The diagram $D'$, presented on fig. \ref{equation2},
\begin{figure}[h]
\[
\small
\xymatrix @=7pt {
 _{0_{1_{\mathds{R}}}:\mathds{R}}\ar[r]&*+[o][F-]{+}\ar `r[rd][rd]&\\
  &&_{\alpha_0:\mathds{R}}\ar `l[lu][lu]\\
}
\]
\caption{Multi-diagram $D'$ codifying $0_{1_{\mathds{R}}}+\alpha_0=\alpha_0$.}\label{equation2}
\end{figure}
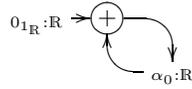
having by source the $[0,1]$-set $\alpha_0:\mathds{R}$, defined by distribution   \[\alpha_0(x,y)=e^{-\frac{(x-x_0)^2}{2}-\frac{(y-x_0)^2}{2}},\]
is also commutate, when the $[0,1]$-set $0_{\mathds{R}}:\mathds{R}$ is defined by distribution
\[0_\mathds{R}(x,y)=\lambda.e^{-\frac{x^2}{2}-\frac{y^2}{2}},\] where the parameter $\lambda$ is a truth value selected in $[0,1[$. Since
\[
\begin{array}{rcl}
  (Lim\;D')(x,y,x) & = & \oplus(x,y,x)\otimes[x,y,x]\\
                  & = & \oplus(x,y,x)\otimes[x]\otimes[y]\otimes[x]\\
                  & = & e^{-\frac{(x-x-y)^2}{2}}.e^{-\frac{(x-x_0)^2}{2}-\frac{(x-x_0)^2}{2}}.\lambda.e^{-\frac{y^2}{2}-\frac{y^2}{2}}.e^{-\frac{(x-x_0)^2}{2}-\frac{(x-x_0)^2}{2}} \\
                  & = & \lambda.e^{-(x-x-y)^2-2(x-x_0)^2-y^2}
\end{array}
\]
and
\[
\begin{array}{rcl}
  \bigvee_y(Lim\;D')(x,y,x) & = & \lambda.e^{-(x-x-0)^2-2(x-x_0)^2-0^2} \\
                           & = & \lambda.e^{-2(x-x_0)^2} \\
                           & = & \lambda\otimes[x]\otimes[x] \\
                           & = & \bigvee_y\alpha_0(x,x)\otimes0_{\mathds{R}}(y,y)\otimes\alpha_0(x,x).\\
\end{array}
\]
However, if we change in diagram $D'$ the interpretation of $\oplus$ using the new distribution
\[\oplus(x,y,z)=\lambda'.e^{-\frac{(z-x-y)^2}{2}},\] depending from a parameter $\lambda'\in[0,1[$. We have
\[
(Lim\;D')(x,y,x) = \lambda'.\lambda.e^{-(x-x-y)^2-2(x-x_0)^2-y^2},
\] thus
\[
\bigvee_v(Lim\;D')(x,y,x) = \lambda'.\bigvee_y\alpha_0(x,x)\otimes0_{\mathds{R}}(y,y)\otimes\alpha_0(x,x).
\]
Then, since we are working in a multiplicative logic, we have
\[
\left(\bigvee_v(Lim\;D')(x,y,x)\Leftrightarrow \bigvee_v(\alpha_0(x,x)\otimes0_{\mathds{R}}(y,y)\otimes\alpha_0(x,x)\right) \geq \lambda',
\]
which means that $D'$ is $\lambda'$-commutative.
\end{exam}

Naturally, if a diagram is $\lambda$-commutative, it also is $\lambda'$-commutative, when $\lambda'<\lambda$. When for every $\lambda>\bot$ the diagram isn't $\lambda$-commutative it is called a non-commutative diagram.

However, when we see a multi-diagram as a way of specify a architectural connectors, in the sense of \cite{Backhouse03}, we may want to interpret diagrams by collapsing the joint execution of its components, generalizing the notion of parallel composition.  We may do this through the symmetry between operators $\otimes$ and $\vee$ on classic logic. We define:

\begin{defn}[Colimit of multi-diagrams]\label{colim}
Given a multi-diagram $D$ in $Set(\Omega)$ with vertices $(v_i)$ the colimit is defined by the multi-morphism
\[
coLim\; D \leq \prod_i M(v_i)
\]
i.e.
 \[
coLim\; D: \prod_i M(v_i) \rightarrow \Omega
\]
given by
\[
(coLim\;D)(\bar{x}_1,a_{i},\bar{x}_2,a_{j},\bar{x}_3)=[\bar{x}_1,a_{i},\bar{x}_2,a_{j},\bar{x}_3)]_{\prod_iM(v_i)}\otimes\bigvee_{f:v_i\rightharpoonup v_j\in \G} D(f)(a_{i},a_{j}).
\]
\end{defn}
This definition allows the formalization of knowledge integration. Colimits capture a generalized notion of parallel composition of components in which the designer makes explicit what interconnections are used between components. We can see this operation as a generalization of the notion of superimposition as defined in \cite{Bosch99}.

The colimit for a multi-diagram  $D$ can be used to define multi-morphism by selection of a set of source and a set of targets. The canonical multi-morphism associated to a multi-diagram $D$, using colimit, have by source $s(D)$ the union of sources used to define the diagram multi-morphisms and have by target $t(D)$ the union of targets of $D$ multi-morphisms.

\begin{exam}
By definition the multi-diagram $D$, presented in fig. \ref{multidiagram1},
have by colimit the $\Omega$-map
\[
coLim\;D:A_0\times A_1\times A_2\times A_3\times A_4\times A_5\rightarrow \Omega
\]
given by
\[
_{(coLim\;D)(a_0,a_1,a_2,a_3,a_4,a_5)=[a_0,a_1,a_2,a_3,a_4,a_5]\otimes(f(a_0,a_1,a_3,a_4,a_5)\vee g(a_1,a_2,a_4,a_5)\vee h(a_2,a_3))}.
\]
\end{exam}

In this sense the coequalizer of a parallel pair of multi-morphisms $R,S:X\rightharpoonup Y$ is defined by the multi-morphism $colim(R=S):X\times Y \rightarrow \Omega$ given by
\[
coLim(R=S)(x,y)=[x,y]_{X\times Y}\otimes(R(x,y)\vee S(x,y)) .
\]
And the pushout of $R:X\rightharpoonup U$ and $S:Y\rightharpoonup U$ is the multi-morphism
\[coLim(R\otimes_U S):X\times U\times Y\rightarrow \Omega\] given by
\[
coLim(R\oplus_U S)(x,u,y)=[x,y]_{X\times Y}\otimes(R(x,u)\vee S(y,u)).
\]

When $D$ is a discrete diagram colimit coincide with the limit of $D$, and in this case, we write
\[
\coprod_vD(v)=\prod_vD(v).
\]

Naturally
\begin{prop}[Existence of coLimit in $Set(\Omega)$]
Every multi-diagram ${D:\G\rightarrow Set(\Omega)}$ have colimit.
\end{prop}

Since $Set(\Omega)$ have limit and colimit of multi-diagrams we use it, in the following, as "Universe of Discurse" to construct model for structures specified by diagrams on the monoidal logic described in $\Omega$.

\begin{exam}[Genome]
The genomes of several organisms have now been completely sequenced, including
the human genome.  Interest within bioinformatics is therefore shifting somewhat away from sequencing, to learning about the genes encoded in the sequence.  Genes code for proteins, and these proteins tend to localize in various parts of cells and interact with one another, in order to perform crucial functions.  A data set presented to KDD Cup 2001 consists of
a variety of details about the various genes of one particular type of organism.
The two tasks proposed for the Data Analysis
Challenge were to predict the functions and localizations of the
proteins encoded by the genes.  A gene/protein can have more than one function,
and more than one localization.  The other information from
which function and localization can be predicted includes the class of the
gene/protein, the phenotype (observable characteristics) of individuals with a
mutation in the gene (and hence in the protein), and the other proteins with
which each protein is known to interact. The dependencies associated to the problem may be expressed by the multi-diagram $D$ presented by fig. \ref{genomeattrib}.
\begin{figure}[h]
\[
\small
\xymatrix @=8pt{Class & & Phenotype  \\
          & Gene \ar@_{->}[lu]\ar@_{->}[ru]\ar@_{->}[ld]\ar@_{->}[rd]\ar@_{->}[r]& Gene \times Interation.type \\
          Function & & Localization
          }
\]
\caption{Dependencies between attributes.}\label{genomeattrib}
\end{figure}
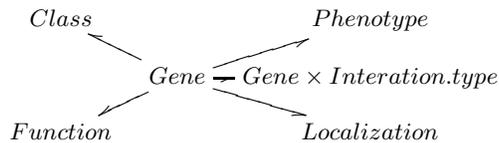
The diagram limit defines a morphism which caracterize the involved entities:
\[
Gene\times Class\times Phenotype \times Gene \times Interation.type\times Function\times Localization
\]
\[
 \begin{array}{c}
    \downarrow \\
    \Omega \\
  \end{array}
\]
This map can be seen as a data set were we can compute $\alpha:Function\times Localization$ describing the similarity between of pair on $Functions\times Localization$. Given a description $\bar{x}$ for a gene, the $\Omega$-map
\[
(Lim\;D)(\alpha|\bar{x})(f,l)
\]
reflect the truth value in $\Omega$ of the proposition "the class of genes characterized by  $x$ have function $f$ and localization $l$".
\end{exam}
In this sense a multi-diagram can be aggregate in a relation via its limit. In the following sections we will describe the inverse problem: Define a multi-diagram having by limit an "approximation" to a given multi-morphism. By this we mean, the possibility of express graphically a fuzzy relations between attributes aggregate in a multi-morphism codifying a data set.

\section{Specifying libraries of components}\label{specifying libraries}

In Computer Sciences a formal grammar is an abstract structure that describes a formal language. Formal grammars are classified into two main categories: generative or analytic.

The generative grammars are the must well-known kind. It is a set of rules by which all possible strings in the language to be described can be generated by successively rewriting strings from a designated start symbol. An analytic grammar, in contrast, is a set of rules that assume an arbitrary string to be given as input, and which successively reduces or analyzes the input string yield a final boolean, "yes/no", result indicating whether or not the input string is a member of the language described by the grammar.

The languages used in this work are expressed through a type of generative grammar where words are configurations defined using componentes selected from a library. Each component have associated a set of requisites and a configuration is valid in the language if every requisite for the used componentes are satisfied. The use of this type of structure and the definition of its semantic was motivated on the Architectural Connectors domain which emerged as a powerful tool for supporting the description of the overall organization of systems in terms of components and there interactions \cite{Fiadeiro97} \cite{Bass98} \cite{Perry92}. According to \cite{Allen97}, an architectural connector can be defined by a set of roles and a glue specification. The roles of a connector type can be instantiated with specific components of the system under construction, which leads to an overall system structure consisting of components and connector instances establishing the interactions between the components.

Let $Chains$ be the forgetful functor from the category of total ordered sets and its
homomorphisms to $Set$, the category of all sets. If we interpret a set $\Sigma$ as a set of symbols or signs, objects on the comma category $(Chains\downarrow \Sigma)$ can
be seen as words defined by strings over \emph{alphabet} $\Sigma$. A set of signs $\Sigma$ equipped with a partial order $\leq$ is called a \emph{ontology}. Given signs $\lambda_0$ and $\lambda_1$ on an ontology $(\Sigma,\leq)$ such that $\lambda_0\leq\lambda_1$, $\lambda_1$ is called a \emph{generalization} of $\lambda_0$ and  $\lambda_0$ is called a \emph{particularization} of $\lambda_1$.

Given $w\in (Chains\downarrow \Sigma)$, we write $w:|w|\rightarrow
\Sigma$, where $|w|$ denotes the chain used on the indexation $w$,
and it is interpreted as an ordered sequence of symbols from
$\Sigma$.

An ontology $(\Sigma,\leq)$ is called a \emph{bipolarized ontology}, if we have
a nilpotent operator $(\_)^+:\Sigma\rightarrow\Sigma$, such that
$\Sigma=\Sigma_I\cup\Sigma_O$, where $(\Sigma_I)^+=\Sigma_O$, and preserving the ontology structure, i.e. given signs $\lambda_0$ and $\lambda_1$  if $\lambda_0\leq\lambda_1$ then $\lambda_0^+\leq\lambda_1^+$.
Set $\Sigma_I$ is called the set of $\Sigma$ \emph{input symbols}
and $\Sigma_I$ is called the set of $\Sigma$ \emph{output symbols}.
If the symbol $\lambda$ is an input symbol, $\lambda^+$ is called
the dual of $\lambda$ and it is an output symbol. A bipolarized
ontology will be denoted by $(\Sigma^+,\leq)$.

Using lifting we define for every word
$w\in (Chains\downarrow \Sigma^+)$, the words:
\begin{enumerate}
  \item $o(w)\in (Chains\downarrow\Sigma_O)$ defined by all output symbols in $w$ and
  \item $i(w)\in (Chains\downarrow \Sigma_I)$ defined by all input symbols in $w$.
\end{enumerate}
\begin{figure}[h]
\[
\xymatrix{\ar@{}[dr]|(.3)\lrcorner {|i(w)|}\ar[r]^{i(w)} \ar[d]_{\subseteq}& \Sigma_I \ar[d]^{\subseteq} \\
          |w|\ar[r]^{w} & \Sigma\ }\quad\quad\quad
\xymatrix{\ar@{}[dr]|(.3)\lrcorner {|o(w)|}\ar[r]^{o(w)} \ar[d]_{\subseteq}& \Sigma_O \ar[d]^{\subseteq} \\
          |w|\ar[r]^{w} & \Sigma\ }
\]
\caption{Pullbacks used to select input and output signs from word $w$.}\label{inoutsign}
\end{figure}
Given a word $w$, we define  $\Sigma(w)$ as the set of symbols used
in $w$,  $\Sigma(w)\subset\Sigma^+$. On an bipolarized ontology, let $w$ and $w'$ be two words,
$w\otimes w'$ is a substring from concatenation $w.w'$
inductively described by the following algorithm:

\begin{alg}\label{op:StringGluing}
    (Input: $w,w'\in \Sigma^+$ Output: $w_i.w'_i$)
    \begin{enumerate}

    \item let $w_0=w$ and $w'_0=w'$.
    \item let $\lambda$ be the first output symbol, in $w_i$ having its dual  $\lambda^+$ in
    $w'_i$ or one of its generalizations :
            \begin{enumerate}
            \item $w_{i+1}$ is generated removing the first occurrence of $\lambda$ from
    $w_i$;
            \item $w'_{i+1}$ is generated removing the first occurrence of $\lambda^+$ or a $\lambda^+$ generalization
    from $w'_i$
            \end{enumerate}
    \item the step 2 is repeated while there are signs in $\Sigma(w_{i+1})$ with dual or generalization in $\Sigma(w_{i+1})$.
    \end{enumerate}
\end{alg}

In this sense we can see the word $w\otimes w'$ as the result of the ordered elimination of output symbols on $w$ and input symbols on $w'$ linked by duality.

From the definition of operator $\otimes$ we can proof:
\begin{prop}
For every pair of words $w$ and $w'$ in a bipolarized ontology
$(\Sigma,\leq)$ we have:
    \begin{enumerate}
    \item $w\otimes  (w'\otimes  w'') = (w\otimes   w')\otimes  w''$;
    \item $w'\otimes w = w\otimes  w'$, if $w'$ (and $w$) not have the dual neither one of its generalizations of signs from $w$ (and $w'$, respectively);
    \item $w\otimes \bot=\bot \otimes w  = w$.
    \end{enumerate}
\end{prop}

For our goal of finding a framework for library specification,
we supposed processes inputs and outputs requirements codified over
signs from a polarized ontology $(\Sigma^+,\leq)$. Thus \emph{the
universe of libraries} having components requirements codified over
the polarized ontology $\Sigma^+$ can be seen as the comma category
\[(Chains \downarrow(Chains \downarrow \Sigma^+)).\]
With this we mean that a library is a list for componentes specified using words defined over $\Sigma^+$.

A \emph{library specification} is a map $L:|L|\rightarrow
(Chains\downarrow \Sigma^+)$, where each node in the chain $|L|$ is
called a \emph{component label} or a \emph{sign} in the library. Given a component label
$r\in |L|$ we can see $L(r)=w:|w|\rightarrow \Sigma^+$ as the
specification of the component input requirements, $i(w)$, and its
output requirements $o(w)$. In this sense we see a library as an
oriented multi-graph having by multi-arrows a selection of
objects in $(Chains\downarrow \Sigma^+)$ and having by nodes objects
from $(Chains\downarrow \Sigma_I)$.

Let $L\in(Chains \downarrow(Chains \downarrow \Sigma^+))$, if
$L(r)=w$, $r$ is interpreted as a dependence between families of
nodes $i(w)$ and $o^+(w)$. And in this case we write
\[r:i(w)\rightarrow o^+(w),\] or for short $r\in L$, defining
a multi-arrow in the  multi-graph $\G(L)$ associated to the library
$L$.

A homomorphism between libraries is a morphism in $(Chains \downarrow(Chains
\downarrow \Sigma^+))$. Every morphism $f:L_0\rightarrow L_1$
between libraries $L_0$ and $L_1$ have associated a multi-graph
homomorphism $\G(f):\G(L_0)\rightarrow \G(L_1)$, defining a correspondence between signs and a correspondence between component labels in $L_0$ and $L_1$, preserving component requirements.

Naturally, we may define an order relation between libraries, we write $L_0\leq L_1$ if for every $r\in |L_0|$ we have $r\in |L_1|$, i.e. every componente existent in $L_0$ is in $L_1$. In this case $L_0$ is called a \emph{sublibrary} of $L_1$ and the associated homomorphism $f:L_0\rightarrow L_1$ is called the \emph{library inclusion}.

We denote by $L^\ast$ the \emph{free monoid} on $L$ for operator $\otimes$. Formally, given a library $L\in (Chains \downarrow(Chains
\downarrow \Sigma^+))$ we define $L^\ast$ as the $\otimes$-closure of $L$, i.e. it is the least library in $(Chains \downarrow (Chains\downarrow \Sigma^+))$ such that:
\begin{enumerate}
  \item every word generated using signs of $L$ is a label in $L^\ast$;
  \item the empty word defines a label for a component having empty requisites $\perp:\perp\rightarrow\perp$;
  \item $L$ is a sublibrary of $L^\ast$;
  \item if $s\in L^\ast$ such that $s=r_1\otimes r_2\otimes \ldots\otimes r_n$ then
  \[L^\ast(s)=L(r_1)\otimes L(r_2)\otimes\ldots\otimes L(r_n).\]
\end{enumerate}
Note what the empty word $\bot$ is a label in $L^\ast$. Since $L$ is a sublibrary of $L^\ast$ the requirements of a label in $L^\ast$ can be interpreted as the requirements for the plugging of the components used on the label definition. In this sense a word in $L^\ast$ can be seen as a circuit defined by the plugging of components from $L$.

A library is a formal system where we may stratify in different levels of abstraction. The level of abstraction of a circuit is define by the number of steps of refinement need to obtain an equivalent circuit using only atomic components. In order to presente what we mean by a circuit refinement, note that in an ontology $(\Sigma^+,\leq)$ the order defined for sign can be lifted to words. We write $\lambda_0\ldots\lambda_n\leq\lambda_0'\ldots\lambda_n'$ if and only if $\lambda_i\leq\lambda'_i$ in $(\Sigma^+,\leq)$. When for two words from ontology we have $w\leq w'$, $w'$ is called a \emph{generalization} of $w$ on the ontology.

Circuit refinement is based on the notion of \emph{semantic for a library} in $L$, and it is a pair of equivalence relations
 $(\equiv_l,\equiv_w)$, where $\equiv_l$ is defined for labels in $L^\ast$ and $\equiv_w$ is
 defined for words in $\Sigma^\ast$, such that:
\[
\text{If } l_0\equiv_l l_1 \text{ then } L^\ast(l_0)\equiv_w L^\ast(l_1).
\]
In $(L^\ast,\equiv_l,\equiv_w)$ a label $s$ is called a \emph{decomposable componente} if there are words $s_0$ and $s_1$ such that:
\[
s\equiv_l s_0\otimes s_1
\]
  We called to a labels that can't be decomposable  an \emph{atomic componente}. In this sense if a label in $L^\ast$ is atomic, it is a label in library $L$.

A \emph{normal form} presentation for a label $s\in L^\ast$ is a sequence of atomic components
\[
(r_0,r_1,r_2,\ldots,r_n)
\]
such that
\[
s\equiv_l r_0\otimes r_1\otimes r_2\otimes\ldots\otimes r_n.
\]

Given a library $L$ we call \emph{library of atomic components} of $(L^\ast,\equiv_l,\equiv_w)$ to the library \[L_{at}\leq L\] such that,  $s\in L_{at}$ if and only if $s$ is atomic in $(L^\ast,\equiv_l,\equiv_w)$.

We want to describe structures, like Architectural connectors, using a graphic language to describe the global organization of complex structures having by resource simplest ones. For that, each component is associated to a graphic presentation and the circuit specification result of the linkage between components satisfying a set of rules and a glue specification \cite{Allen97}. The essence of our approach is to provide a general framework that gives circuit explicit semantic status. To formalize this, for a library $L$ we associated a  multi-graph $\G(L)$, having by nodes symbols from $\Sigma_I$, and by multi-arcs componentes such that each component $s$ have by input $i(L(s))$ and output $o(L(s))$.
By $\G^\ast(L)$ we denote the comma category
\[(Mgraph\downarrow\G(L))\]
having by objects homomorphisms defined between a multi-graph and the multi-graph $\G(L)$.

Trivially, for a library $L$ with polarized ontology
$(\Sigma^+,\leq)$, any diagram $D\in \G^\ast(L)$ can be
codified as a library \[L(D)\in (Chains\downarrow(Chains\downarrow
\Sigma^+)),\]
having as component labels multi-arcs, from $D$, each one must be associated to a word defined concatenating, in a single word, the multi-arc sources vertices labels and its targets vertices dual labels. Given a diagram $D\in \G^\ast(L)$ defined through:
\begin{enumerate}
  \item the source map $i:|D|\rightarrow (Chains\downarrow \Sigma_I)$ and
  \item the target map $o:|D|\rightarrow (Chains\downarrow \Sigma_I)$.
\end{enumerate}
Any library $L(D)$ have associated two words, defined using symbols from  $\Sigma_I$; This words are its
input requisites $i(D)$ and its output structures $o(D)$, where:
\begin{enumerate}
  \item $i(D)$ is the word define concatenating labels belonging to vertices without input multi-arc;
  \item $o(D)$ is a word defined by concatenation of the dual of labels belonging to vertices without output multi-arc.
\end{enumerate}

On the category $\G^\ast(L)$, for every pair of diagrams
$D$ and $D'$, we define the diagram $D\otimes D'$ by gluing
together vertices with equal labels belonging to $o(D)$ and $i(D')$, taken others as distinct. The order used to gluing vertices must respect the order given by the chain of symbols used to define the words $o(D)$ and $i(D')$.

When $\otimes$  is restricted to pairs of diagrams $D$ and $D'$
such that $i(D)=o(D')$, we used this operator as a "composition"
between relations specified using multi-graphs. With it we define a
category having by objects words from $(Chains \downarrow \Sigma_I)$
and by morphisms diagrams from $\G^\ast(L)$. Given a
diagram $D$, the fact of $i(D)=w$ and $o(D)=w'$ is denoted by
$D:w\rightharpoonup w'$. Given diagrams $D$ and
$D'$ from $\G^\ast(L)$, if $D'$ is a subobject of $D$, denoted by writing $D'\leq D$, if here is an epimorphism in $\G^\ast(L)$ from $D'$ to $D$, or equivalently, if there is a decomposition $D=D''\otimes D' \otimes D'''$.

A diagram $D$ is \emph{decomposable} if there are two not null subobjects $D'$ and $D''$ such that $D=D'\otimes D''$. If a diagram isn't decomposable it is called \emph{atomic}. Let $\G_{at}^\ast(L)$ be the class of atomic diagrams in $\G^\ast(L)$. We can see atomic diagrams as building blocks for generate diagrams. A functor $F:\G^\ast(L)\rightarrow \G^\ast(L)$, where $L$ have a semantic, is called a \emph{diagram refinement} if:
\begin{enumerate}
  \item $F(D\otimes D')\equiv_l F(D)\otimes F(D')$, i.e. the refinement of a diagram is semantically equivalente to the refinement for its parts;
  \item $F(D)\equiv_l D$, i.e. the refinement of a diagram is semantically equivalent to it self;
  \item $F(D)=D$ if and only if $D\in \G_{at}^\ast(L)$, i.e. atomic elements cannot be simplified.
\end{enumerate}
A refinement can be seen as a rewriting rule allowing unpacking subdiagram encapsulations. If a diagram is a fixed-point for the refinement function we say it is in \emph{normal form}. And since diagrams are finite structure, for every diagram $D$ we can find, at least, a representation of $D$ in normal form in a finite number of steps.

A diagram refinement $F:\G^\ast(L)\rightarrow \G^\ast(L)$ have the nice property of defining a partial order in $\G^\ast(L)$, denoted by $\leq_F$ and where $D \leq_F D'$ is true if $F(D')=D$, i.e. if $D$ is a refinement of $D'$ through $F$, and in this case we call to $D'$ a \emph{generalization} of $D$.

\begin{exam}[Signatures as libraries]\label{ex:signature}
A signature can be expressed through a library of components, where each component represents a function symbol where its arity is codified on the component requirements. Formalizing this: Following \cite{Makki89} a signature
$\Sigma=(S,T,ar)$,  with type symbols from $S$, consists of a finite
set $T$ of function symbols (or operators) $f,g,\ldots$ where each
function $f$ has an arity ${ar(f)=(<a_i>_I,b)}$ defined by a chain of
input type symbols and one output type symbol. In this case we write
$i(f)=<a_i>_I$ and $o(f)=<b>$. We can sort the set $T$ of symbols
function and taking these symbols as labels of components having its
requirements codified over the polarized alphabet $S^+$ generated from
$S$. The set $S^+$ is defined adding a new dual symbol $a^+$ for
each type symbol $a$ in $S$. The library associated to the signature
$\Sigma$, will be denoted by $L(\Sigma )$.

A constant, of type $a$, is a function symbols in a signature with
arity $(<>,a)$, i.e. without input and having $a$ as output. It
is usual to take a countable infinite set of variables for each type
used on the signature. They are codified on a diagram using
inputs on components without associated links, and defining the
set of diagram sources.
\end{exam}

Bellow we present some examples of libraries associated to models generated by machine learning algoritmos in \cite{Michell86} and used on the following for the presentation of examples by describing fuzzy structures :
\begin{exam}[Binary Library $\L_B(S,C)$]\label{binarylibrary}
Binary libraries are define using a set of component labels $C$ and a set of type signs $S$. A binary library $\L_B(S,C)$ presuppose the existence of a sign $l\in S$, interpretable as the set $\Omega$ of truth values in $Set(\Omega)$, and a
constant $\top:l^+$ in $C$ interpreted as true. In $S$ we must have  defined  components of type $=_s:ssl^+$, one for each symbol $s\in S$, interpreted as a
similarity $[\cdot=\cdot]$ on the interpretation for $s$, and also components of type $c:s^+$, where $c\in C$ and $s\in S$, interpreted as a constants selected on the interpretation for $s$.

Since data sets or tables can be codified using this primitives, binary libraries are also called data set libraries
\end{exam}

\begin{exam}[Linear Library $\L_L(S,C)$]\label{linearlibrary}
Linear libraries $\L_L(S,C)$, extend binary libraries, are defined
using similarities $=_s:ssl^+$, components specified as
$\geq_s:ssl^+$, for each symbol $s\in S$, interpretable as a total order
and constant components $c:s^+$, where $c\in C$ and $s\in S$.

Linear library are associated to processes for the discretization of continuous domains in this sense they are also called grid libraries.
\end{exam}

\begin{exam}[Additive Library $\L_A(S,C)$]\label{addlibrary}
An additive libraries $\L_A(S,C)$ is an extension to a linear libraries. They are defined
using equality and order components $=_s:ssl^+$, $\geq_s:ssl^+$,
with a component specified as $+_s:sss^+$, interpretable as an
addition for all symbol $s\in S$, and constantes $b:s^+$ where $b\in
C$ and $s\in S$.
\end{exam}

\begin{exam}[Multiplicative Library $\L_M(S,C)$]\label{multlibary}
Multiplicative libraries $\L_M(S,C)$ are extensions to additive libraries. They are
defined using components $=_s:ssl^+$, $\geq_s:ssl^+$,$+_s:sss^+$,
also have a component $\times_s:sss^+$, for each symbol $s\in S$,
interpretable as a multiplication and constantes $b:s^+$ for $b\in C$ and $s\in
S$.
\end{exam}

\section{Modeling libraries and Graphic Languages}\label{modeling libraries}

Since libraries and multi-graphs have structural compatibility it is natural to assume the soundness for library semantic in $Set(\Omega)$ as equivalente to the library structural preservation.

A model for a library $(L^\ast,\equiv_l,\equiv_w)$ in $Set(\Omega)$ is a multi-graph homomorphism $M$ from the \emph{library parser graph} $\G(L^\ast)$ to $Set(\Omega)$,
\[M:\G(L^\ast)\rightarrow Set(\Omega)\] such that
\begin{enumerate}
  \item equivalente componentes are interpreted as the same multi-morphism, i.e. $M(r)=M(r')$ if $r\equiv_l r'$;
  \item transform componente gluing in multi-morphism composition, i.e. $${M(r\otimes r')=M(r)\otimes M(r');}$$
  \item preserves componente requirements and truth value distribution, i.e. if $r:w\rightharpoonup w'$ then $$M(r):M(w)\rightharpoonup M(w') \text{ and } M(r)^\circ\otimes M(w)\otimes M(r)= M(w');$$
  \item preserves sign ontological structure, i.e. if sign $l$ is a generalization for sign $l'$ (i.e. if $l'\leq l$) then $M(l')\leq M(l)$;
  \item words are mapped to as chains of $\Omega$-set products defined in \ref{ProdSimil}, i.e. if $w=s_1s_2\ldots s_n$ then $M(w)=\Pi_iM(s_i)$;
  \item equivalente words are mapped to the same $\Omega$-set defined in \ref{def:isomorphism}, i.e. if $w\equiv_w w'$ then $M(w)= M(w')$.
\end{enumerate}
In other words a model transform multi-arcs into multi-morphisms preserving its structure and the semantic induced through relations $\equiv_l$ and $\equiv_w$. Property (3) imposes the preservation of truth values distribution by componente interpretation. The class of models for a library $(L^\ast,\equiv_l,\equiv_w)$ is used, in the sequel, on definition of a category of models \[Mod(L^\ast,\equiv_l,\equiv_w).\]

A model for $L^\ast$ can be defined lifting interpretation of atomic components to circuits. For that we must note that, since a model preserves componente gluing, it can be defined by fixing interpretations for its atomic components. This is expressed by the following completion principle:

\begin{prop}[Universal property]
Let $L_{at}$ be the sublibrary defined by atomic components in ${(L^\ast,\equiv_l,\equiv_w)}$. Every multi-graph homomorphism \[M:\G(L_{at})\rightarrow Set(\Omega),\] defines a unique model \[M^\ast:\G(L^\ast)\rightarrow Set(\Omega),\] for $(L^\ast,\equiv_l,\equiv_w)$, such that \[M^\ast\circ i=M,\] where $i$ is the homomorphism defined by library inclusion.
\end{prop}

The proof to this result is made defining $M^*(r)=\otimes_iM(r_i)$ if the circuit $r$ have a normal form given by a sequence
$
(r_0,r_1,r_2,\ldots,r_n)
$, i.e. the $r_i$'s are atomic and
\[
r\equiv_l r_0\otimes r_1\otimes r_2\otimes\ldots\otimes r_n.\]

For every component label $r\in L$ we called \emph{a realization}
for $r$ through $M$ in $Mod(L^\ast,\equiv_l,\equiv_w)$ or a \emph{Chu representation} of $r$ to the a \emph{epi multi-morphism} \[M(r):M(i(r))\rightharpoonup M(o(r)) \text{ such that } M(r)^\circ\otimes M(i(r))\otimes M(r)=M(o(r)).\] In this case, if $r=
r'\otimes r''$ in $L^\ast$, then $M(r)$ can be decomposed in $Set(\Omega)$
as
\[M(r')\otimes M(r'').\]

Since library refinement preserves semantics, it is idempotent with regard to library models, given a model $M\in Mod(L^\ast,\equiv_l,\equiv_w)$ and if $F:L^\ast\rightarrow L^\ast$ is a refinement in library $(L^\ast,\equiv_l,\equiv_w)$, we have
\[
M( F^n(D))=M(D),\text{ for every configuration } D\in\G(L^\ast).
\]
This property can be used to characterize refinement:
\begin{prop}
A library homomorphism $F:L^\ast\rightarrow L^\ast$ having by fixed-points atomic components is a refinement in $(L^\ast,\equiv_l,\equiv_w)$ if and only if for every model $M\in Mod(L^\ast,\equiv_l,\equiv_w)$ we have
\[
M\circ F^n=M
\] for each natural $n$.
\end{prop}

In this sense, for every component $r$ which is a refinement for $r'$ by $F$, i.e. $r\leq_Fr'$, we have, for every library model $M$, $M(r')=M(r)$.


A library can be seen as an analytic grammar and we can use them to characterize languages. We define the \emph{graphic language} associated to a library $L$ as the set of valid finite configurations using components indexed by
$L$. A configuration of components $D$ is \emph{valid or allowed} in $L$ if
\[
D\in \G^\ast(L),
\]
i.e. if $D$ is a multi-graph homomorphism
\[D:\G\rightarrow \G(L).\] Formally, given a library $L\in (Chains\downarrow (Chains\downarrow
\Sigma^+))$, a graphic word $D$ defined by $L$ is a finite
configuration
\[D\in (Mgraph\downarrow \G(L)).\] In this sense a word in the language is a multi-graph homomorphisms where the multi-arrows are library components. Since the homomorphism $D$ have $\G(L)$ as codomain it satisfy library constrains and it can be seen, and is called, the \emph{parsing} of word $D$.

The \emph{graphic language defined by} library $L$, is denoted by $Lang(L)$,
and it is the comma category
$\G^\ast(L)=(Mgraph\downarrow \G(L))$ of allowed configurations in
$L$. Given an allowed configuration $D:\G\rightarrow \G(L)$ we call
 $\G$ the \emph{configuration shape} and  to $D(\G)$ a
\emph{word} or \emph{diagram} on the language defined by $L$.

Given a configuration $D\in Lang(L)$ and a model $M\in
Mod(L^\ast,\equiv_l,\equiv_w)$, we define
\begin{enumerate}
  \item $i(M(D))=M(i(D))$ and
  \item $o(M(D))=M(o(D))$,
\end{enumerate}
where $M(i(D))$ and $M(o(D))$ denote $\Omega$-sets in $Set(\Omega)$ used to give meaning to multi-diagram input and output vertices.

We may collapse the structure of a word interpretation on a multi-morphism using limits.

\begin{defn}[Limits as multi-morphisms]\label{def:word interpret}
Given a model $M\in Mod(L^\ast,\equiv_l,\equiv_w)$, the \emph{interpretation for a configuration}
$D:\G\rightarrow \G(L)$ through $M$ is  $Lim\;MD$ a multi-morphism
\[ M(i(D))\rightharpoonup M(o(D)).\]
In this case we write $M(D)$ to denote the multi-morphism $Lim\;MD$.
\end{defn}

Naturally we define

\begin{defn}[coLimits as multi-morphisms]\label{def:colim interpret}
Given a model $M\in Mod(L^\ast,\equiv_l,\equiv_w)$ and a diagram
$D:\G\rightarrow \G(L)$. Its colimit $coLim\;MD$ can be seen as a multi-morphism
\[coLim\;MD : M(i(D))\rightharpoonup M(o(D)).\]
\end{defn}

By definition the model for a library preserves component decomposition, which can be extended to multi-diagrams when interpreted in a basic logic.
\begin{prop}
Let $\Omega$ be a basic logic. If multi-diagram $D$ is a word on the language defined by library $L$ and if it can be obtain by gluing diagrams $D_1$ and $D_2$, i.e. $D=D_1\otimes D_2$. An interpretation for word $D$  is the result of composing the interpretation of $D_1$ and $D_2$ , i.e.
\[M(D)=[\cdot=\cdot]_{\otimes H}\Rightarrow (M(D_1)\otimes M(D_2)),\]
where $H=M(i(D_2)\cap o(D_1))$ is the set of $\Omega$-sets that are sources for diagram $D_2$ and targets for $D_1$.
\end{prop}

Note that, if $M(D_1)^\circ\otimes\alpha\otimes M(D_1)=\beta$ and $M(D_2)^\circ\otimes\beta\otimes M(D_2)=\gamma$ then \[M(D_1\otimes D_2)^\circ\otimes\alpha\otimes M(D_1\otimes D_2)=M(D_2)^\circ\otimes M(D_1)^\circ\otimes\alpha\otimes M(D_1)\otimes M(D_2)=\gamma.\]
In a basic logic $\Omega$, this strategy can be extended to configurations colimit:
\[colim\; M(D_1\otimes D_2)=[\cdot=\cdot]_{\otimes H}\Rightarrow (colim\;M(D_1)\otimes colim\;M(D_2)),\]
where $H=M(i(D_2)\cap o(D_1))$.

\section{Library descriptive power}\label{descritive power}

Lets define now the structure for the category of models for a library, $Mod(L^\ast,\equiv_l,\equiv_w)$. It has by objects models \[M:\G(L^\ast)\rightarrow Set(\Omega),\]
and by morphisms natural transformations. In this context, taking $D=\G(L^\ast)$ the graphic library structure, a natural transformation from model $M_1$ to model $M_2$ is a pair of epi multi-morphisms $(f,g):M_1\Rightarrow M_2$, such that
\[ f \otimes M_2(D)= M_1(D)\otimes g.\]
Naturally, by definition of epi multi-morphism,  \[f^\circ\otimes M_1(i(D))\otimes f=M_2(i(D))\text{ and }g^\circ\otimes M_1(o(D))\otimes g=M_2(o(D)).\]

\begin{figure}[h]
\[
\xymatrix{M_1(i(D)) \ar@_{->}[r]_f\ar@_{->}[d]_{M_1(D)} & M_2(i(D))\ar@_{->}[d]_{M_2(D)} \\
          M_1(o(D)) \ar@_{->}[r]_g& M_2(o(D))
          }
\]
\caption{Natural transformation $(f,g)$.}\label{naturaltransf}
\end{figure}

 We use the composition for multi-morphism to define the composition of natural transformation. Given two natural transformations $(f_1,g_1):M_1\Rightarrow M_2$ and $(f_2,g_2):M_2\Rightarrow M_3$ we define
\[
(f_1,g_1)\otimes (f_2,g_2)=(f_1\otimes f_2,g_1\otimes g_2).
\]
A model $M$ have by identity in $Mod(L^\ast,\equiv_l,\equiv_w)$ the natural transformation \[(1_{i(M(D))},1_{o(M(D))}),\] where both epi multi-morphisms are defined using the identity relation in $\Omega$.

The usual limits and colimits in $Mod(L^\ast,\equiv_l,\equiv_w)$ are computed based on that are made in the category of $\Omega$-sets and epi multi-morphism, $epi$-$Set(\Omega)$. Where the product (in usual sense) exists for two $\Omega$-sets $\alpha:A$ and $\beta:B$, if
\[\bigvee_{a,a'}\alpha(a,a')=\top \text{ and } \bigvee_{b,b'}\beta(b,b')=\top,\]and it is defined by object $\alpha\otimes\beta:A\times B$ and the usual projections $\pi_A:A\times A\rightarrow A$ and $\pi_B:B\times B\rightarrow B$ in $Set$ and codified as a function in $Set(\Omega)$. Note that, for instance, $\pi_A$ is a epi multi-morphism, since $\pi_A^\circ\otimes(\alpha\otimes\beta)\otimes \pi_A$ define the multi-diagram, presented in fig. \ref{product},

\begin{figure}[h]
\[
\small
\xymatrix @=10pt {
&&A\ar[r]&*+[o][F-]{\alpha}\ar[r]&A\ar[dr]&&\\
 A\ar[r]&*+[o][F-]{\pi_A^\circ}\ar[ur]\ar[dr]&&&&*+[o][F-]{\pi_A}\ar[r]&A\\
&&B\ar[r]&*+[o][F-]{\beta}\ar[r]&B\ar[ur]&&\\
 }
\]
\caption{Multi-morphism $\pi_A^\circ\otimes(\alpha\otimes\beta)\otimes \pi_A$.}\label{product}
\end{figure}
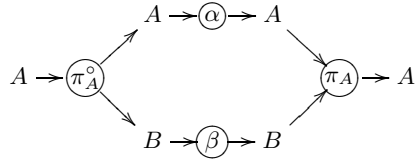

by composition we have
\[\bigvee_{b'}\bigvee_{b''}\bigvee_{a'}\bigvee_{a''} \pi_A^\circ(a,b',a')\otimes \alpha(a',a'')\otimes \beta(b',b'')\otimes\pi_A(a'',b'',a''') =\]
\[=\bigvee_{b'}\bigvee_{a'}\pi_A^\circ(a,b',a') \otimes\bigvee_{a''}\bigvee_{b''}(\alpha(a',a'')\otimes\beta(b',b'')\otimes\pi_A(a'',b'',a''') )=\]
\[=\bigvee_{b'}\pi_A^\circ(a,b',a) \otimes\bigvee_{b''}(\alpha(a,a''')\otimes\beta(b',b'')\otimes\pi_A(a''',b'',a''') )=\]
\[=\bigvee_{b'}\bigvee_{b''}(\alpha(a,a''')\otimes\beta(b',b''))=
\alpha(a,a''')\otimes\bigvee_{b'}\bigvee_{b''}\beta(b',b'')=\alpha(a,a''')\]
 The category $epi$-$Set(\Omega)$ doesn't has an initial object. However, for $\Omega$-object $\alpha:A$ with a unique factorization $\alpha=f^\circ\otimes f$, the multi-morphism  $f:\emptyset\rightharpoonup A$ is the only epi multi-morphism, since $f^\circ\otimes f= \alpha$. But there is only one epi multi-morphism to $\emptyset:\emptyset$ given by the empty relation $\emptyset:A\rightharpoonup \emptyset$ since $\emptyset^\circ\otimes\alpha\otimes\emptyset=\emptyset$.

 Given a pair of epi multi-morphisms $f,g:A\rightarrow B$ from $\alpha:A$ to $\beta:B$. There is a equalizer for $f$ and $g$, and it is given by $\gamma:A$ if and only if
 \[
 (f\otimes g)^\circ\otimes\gamma\otimes (f\otimes g)=\beta.
 \]
Every pair of epi multi-morphisms $f$ and $g$ has a coequalizer and it is given by $\gamma:B$ such that
\[
 \gamma:=(f\otimes g)^\circ\otimes\alpha\otimes (f\otimes g).
 \]
And, every family $(f_i:\alpha_i\rightharpoonup \beta)$ of epi multi-morphisms have wild pushouts given by the product $\otimes_i\alpha_i:\prod_iA_i$ and its projections $\pi_j:\prod_iA_i\rightarrow A_j$. Since, $\pi_j^\circ\otimes_i\alpha_i\pi_j=\alpha_j$ and $f_j^\circ\otimes\pi_j^\circ\otimes_i\alpha_i\pi_j\otimes f_j=f_j^\circ\otimes\alpha_j\otimes f_j=\beta$. Because $epi$-$Set(\Omega)$ have coequalizers and  wild pushouts it has connected colimits (see \cite{Borceux94}). By definition of natural transformation in $Mod(L^\ast,\equiv_l,\equiv_w)$ we have:

\begin{prop}
The category $Mod(L^\ast,\equiv_l,\equiv_w)$ has connected colimits in the usual sense.  \end{prop}

Since $epi$-$Set(\Omega)$ and $Mod(L^\ast,\equiv_l,\equiv_w)$ have connected colimits they have directed colimits, i.e. exist colimit for diagrams like $D:(I,\leq)\rightarrow epi$-$Set(\Omega)$  where $(I,\leq)$ is a poset and if the vertices are models then its colimit is a model (see definition in \cite{Adamek94}).

Following \cite{Adamek94}, a category is accessible, provided that has directed colimits and has a set $\A$ of presentable objects such that every object is a direct colimit of objects from $\A$. And accessible categories can be characterized by:

\begin{prop}\cite{Adamek94}
Each small category with split idempotents is accessible
\end{prop}
Where, a category has split idempotents if for every morphism $f:A\rightarrow A$ with $f.f=f$ there exist a factorization $f=i.p$ where $p.i=id_A$.

\begin{exam}
If $\Omega$ is a ML-algebra with at least three logic values $\bot<\lambda<\top$. The multi-morphism \[f=\left[
     \begin{array}{cc}
       \top   & \alpha \\
       \alpha & \top \\
     \end{array}
   \right]
\]
is an epi multi-morphism and is idempotent but non-splitable, since for every factorization $f=i\otimes p$, $p\otimes i\neq id$ by proposition \ref{prop:comprestriction}.
\end{exam}

The fact described by this example are the rule in $epi$-$Set(\Omega)$ and $Mod(L^\ast,\equiv_l,\equiv_w)$. Since must of the idempotents aren't split idempotents the have:

\begin{prop}
Let $\Omega$ is a ML-algebra with more than two logic values. Then categories of models of libraries $Mod(L^\ast,\equiv_l,\equiv_w)$, aren't accessible.
\end{prop}

This means that, we can't use Ehreasman sketches to specify the category of models of a library \cite{Adamek94}, i.e. model categories can't be axiomatizable by basic theories in first-order logic.

\section{Sign systems and Semiotics}\label{Sings Semiotics}

Lets new find what is the basic structure need on a library to define useful fuzzy structures.

\begin{exam}[Signatures as libraries]
Let $\Sigma$  be a signature. By example \ref{ex:signature} a
signature can be seen as a library. The set $T(\Sigma)$ of terms
defined by $\Sigma$ is given, in \cite{Makki89}, as the least set satisfying:
\begin{enumerate}
  \item each variable $x$ is in $T(\Sigma)$, and if it
has type $b$ we write $x:b$;
  \item if $f\in \Sigma$   with $ar(f)=(<a_i>_I,b)$ and $<t_i>_I$
is a list of terms from $T(\Sigma)$, with $t_i:a_i$ for every $i\in
I$, then $f(t_i)_I$ is in $T(\Sigma)$, and it has type $b$, and we
write in this case $f(t_i)_I:b$.
\end{enumerate}

A term is closed if it doesn't contain variables.  We have
\[T(\Sigma)\cong Lang(L(\Sigma ,S))\] if and only if we impose the
existence of:
\begin{enumerate}
  \item especial symbols $c$ and $v$ in $S$, as described in example \ref{ex:signature} ;
  \item diagonal components $\lhd$  in $|L(\Sigma ,S)|$ with
$i(\lhd)=a$  and $o(\lhd)=<a>_I$ , in $L(\Sigma)$, one for each
alphabet symbol $a$ in $|L(\Sigma ,S)|$ and each "class" of chain
equivalences $I$. These components represent dependencies in the term
structure inherent to the use of the same variable more than once
on the term definition.
\end{enumerate}

If $t$ is a term in $T(\Sigma)$ then it can be represented as a
multi-graph homomorphism \[t:\G\rightarrow \G(L(\Sigma,S)).\] This
graph homomorphism is usually called the parsing graphs for $t$. If
variables involved on the definition of a term are different then
the diagram $t$ shape, $\G$, is a tree. On fig. \ref{terms} we present allowed configurations
defining terms \[f(x:b,g(y:c,z:d):e):a \text{ and }
f(x:b,g(y:c,x:b):e):a.\]
\begin{figure}[h]
\[
\small
\xymatrix @=10pt {
\ar[rrd]_b&  & & \\
\ar[rd]_c&   &*+[o][F-]{f}\ar[rr]_a  & &\\
&*+[o][F-]{g}\ar[ru]_e&  & \\
\ar[ru]_d&   & &}
\quad\quad
\xymatrix @=10pt {
\\
\ar[rrd]^c&&   &*+[o][F-]{f}\ar[rr]_a  & &\\
\ar[r]_b&*+[o][F-]{\lhd}\ar[rru]^b\ar[r]_b&*+[o][F-]{g}\ar[ru]_e&  & \\
}
\]
\caption{Multi-morphism $f(x:b,g(y:c,z:d):e):a$ and
$f(x:b,g(y:c,x:b):e):a$.}\label{terms}
\end{figure}
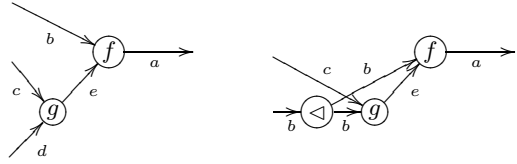

Let now $\Sigma$  be a signature with a special collection of
functional symbols denoted by $=_a$, one for each data type symbol a
used on the signature, where \[ar(=_a)=(<a,a>,l).\] In the arity of
$=_a$, the symbol $l$ must be seen as an identification for the set
of truth-values and should have associated two constant operators
$T$ and $F$, both with arity $(<>,l)$, used to identifying the true
and the false on the associated logic framework. By practical
reasons, given two terms of type $a$, $t:a$ and $s:a$, instead of
writing the term $=_a(t,s)=T$, we use the usual infix notation and
write $t=_as$ calling it an equation of type $a$. In a signature
with this characteristics we called relational symbol to every
functional symbol $f$ with arity of type \[ar(f)=(<a_i>_I,l).\] In
this sense the symbol $=_a$ is relational, and this sort of
signature $\Sigma$ is said to have a logic structure.

A formula $f$, on a signature $\Sigma$ with a logic structure, is a
term which has by representation a multi-graph
homomorphism having by output a relation, i.e. $o(f)=<l>$.
\end{exam}

The above example requires the existence in $\Sigma$ of a special symbol "$l$" and the existence of special component labels "="
and "$\lhd$", having predefined interpretations. This necessity can be seen frequently in other examples. For this type of labels we will define restrictions to the language  model structures by fixing interpretation to some signs and to some structures definable in the library.  We specify this type of structures using Ehresmann sketches defined by multi-graphs. We called  specification systems to this generalization.

\begin{defn}
A \emph{specification system} $S$, using a library $L$ ( or a
\emph{sign system} $S$ using $L$) is a structure $S=(L,\E,\U,co\U)$
where
\begin{enumerate}
  \item $\E\subset Lang(L)$ is a set of \textbf{finite diagrams}, interpreted as a total multi-morphism (see definition \ref{total}),
  \item $\U$ is a set of tuples $(f,D,i(D),o(D))$ where $f$ is a component and $D$ is a \textbf{finite configurations}, such that $f$ is interpreted as the multi-morphism defined by the limit of $D$ having sign $i(D)$ as input vertices and sign $o(D)$ as output vertices (see definition \ref{def:word interpret}), and
  \item $co\U$ is a set of tuples $(f,D,i(D),o(D))$ where $f$ is a component and $D$ is a \textbf{finite configurations}, such that $f$ is interpreted as the multi-morphism defined by the colimit of $D$ from sign $i(D)$ to $o(D)$ (see definition \ref{def:colim interpret}).
\end{enumerate}
\end{defn}

In a specification system  $S=(L,\E,\U,co\U)$, while the set $\E$ define structural proprieties to be preserved by its models, sets $\U$ and $co\U$ impose restrictions to the structure for sign interpretations.

Extending the definition of model of a small Ehresmann sketches to interpretations of sign systems in $Set(\Omega)$ we have:
\begin{defn}[Model for a specification system]\label{Modelspec}

 A model $M$ in $Mod(L^\ast,\equiv_l,\equiv_w)$, for library $L$, is a model for the specification system
$S=(L,\E,\U,co\U)$ if:
\begin{enumerate}
  \item for every pair $D\in \E$,  $M(D)$ is a total multi-morphism;
  \item for every $(f,D,i(D),o(D))\in \U$, $M(f)$ is the multi-morphism defined by $Lim\;MD$ from $M(i(D))$ to $M(o(D))$;
  \item for every $(f,D,i(D),o(D))\in co\U$, $M(s)$ is the multi-morphism defined by $coLim\;MD$ from $M(i(D))$ to $M(o(D))$ .
\end{enumerate}

\end{defn}
The category defined by models for a specification system in
$Set(\Omega)$ and natural transformations between interpretations is denoted by $Mod(S)$. And we call \emph{the sketch structure of} $S$ to $(\E,\U,co\U)$. By definition, the category $Mod(S)$ is a full subcategory of $Mod(L^\ast,\equiv_l,\equiv_w)$. And, since the category of library models don't have split idempotents relations evaluated in $\Omega$:

\begin{prop}
Let $\Omega$ be a not trivial ML-algebra with more than two truth values. Given a specification system $S$, the category $Mod(S)$ isn't an accessible category.
\end{prop}

Since the category $Set$ is a full subcategory of $Set(\Omega)$ each map can be seen as a relation and because we can codify commutative diagrams using total morphisms and limit cones using the limit structure, defined in \ref{lim}, trivially we have:
\begin{prop}
Every model for a Ehresmann Limit sketch in  $Set$,
defined using finite diagrams, can be specified by a sign systems in
$Set(\Omega)$ for every not trivial ML-algebra $\Omega$.
\end{prop}

By this, every algebraic theory (see definition in \cite{Adamek94}) has a fuzzy version defined in $Set(\Omega)$.

Since in the following we will work over models of specification systems we define a semiotic as a sign system furnished with a model. This structure associates syntactic and semantic components to a language on the Goguen's  institution spirit \cite{goguen83}. Formally
\begin{defn}[Semiotic system]
A \emph{semiotic system} is a pair $(S,M)$ defined by a sign system
$S=(L,\E,\U,co\U)$ and a model $M\in Mod(S)$.
\end{defn}
We will denoted by $Lang(S)$ the language associated to a sign
system $S=(L,\E,\U,co\U)$ or associated to a semiotic system
$(S,M)$.

In the context of information systems: we can see a system specification as a database structure and a semiotic defined with this structure as a database state. Each database update induces a change in the database state, implying a semiotic change since it reflects a change in the system attributes relations codified on the database tables. The information system is then a semiotic since it is usually defined as a database instance or state. Then the information stored in the information system can be queried in the associated semiotic.

Lets see an example of a semiotic and how we can query it using limits.
\begin{exam}\cite{Cohen01}
For the IDA'01 - Robot Data Challenge - series of vectors of binary data was generated by the perceptual system of a mobile robot. We suspect the generated time series contains several patterns (where a pattern must be see as a structure in the data that is observed, completely or partially, more than once) but we not know the pattern boundaries, the number of patterns, or the structure of patterns. We suspect that at least some patterns are similar, but perhaps no two are identical. The challenge is to find the patterns and elucidate their structure. A supervised approach to the problem might involve learning to recognize patterns given known examples of patterns.

The robot dataset is a time series of 22535 binary vectors of length 9, generated by a mobile robot as it executed 48 replications of a simple approach-and-push plan. In each trial, the robot visually located an object, oriented to it, approached it rapidly for a while, slowed down to make contact, and attempted to push the object. In one block of trials, the robot was unable to push the object, so it stalled and backed up. In another block the robot pushed until the object bumped into the wall, at which point the robot stalled and backed up. In a third block of trials the robot pushed the object unimpeded for a while. Two trials in 48 were anomalous.

Data from the robot's sensors were sampled at 10Hz and passed through a simple perceptual system that returned values for nine binary variables. These variables indicate the state of the robot and primitive perceptions of objects in its environment. They are: STOP, ROTATE-RIGHT, ROTATE-LEFT, MOVE-FORWARD, NEAR-OBJECT, PUSH, TOUCH, MOVE-BACKWARD, STALL. For example, the binary vector [0 1 0 1 1 0 1 0 0] describes a state in which the robot is rotating right while moving forward, near an object, touching it but not pushing it. Most of the 512 possible states are not semantically valid, however the robot's sensors are noisy and its perceptual system makes mistakes.

The dataset was segmented into episodes by hand. Each of 48 episodes  contains zero or more instances of seven episode types, labelled A, B1, B2, C1, C2, D and E. The entire corpus contains 356 instances of these episode types.

We may use domain knowledge to define a library $L$ which can be used to characterize relations between attributes. The easier way to do this library is by directly specifying its parsing graphic $\G(L)$. For that we fixed as signs \emph{Move}, \emph{Objects}, \emph{Path}, \emph{Node}, \emph{Episode}, \emph{Rotate}, \emph{Stalled} and \emph{Class} and by specifying components \emph{pushing}, \emph{direction}, $stat_1$, \emph{proximity}, \emph{source}, \emph{target}, \emph{start}, \emph{end}, \emph{direction}, $stat_2$ and \emph{type} having its constrains defined in the graph bellow.
\begin{figure}[h]
    \begin{center}
    \includegraphics[width=150pt]{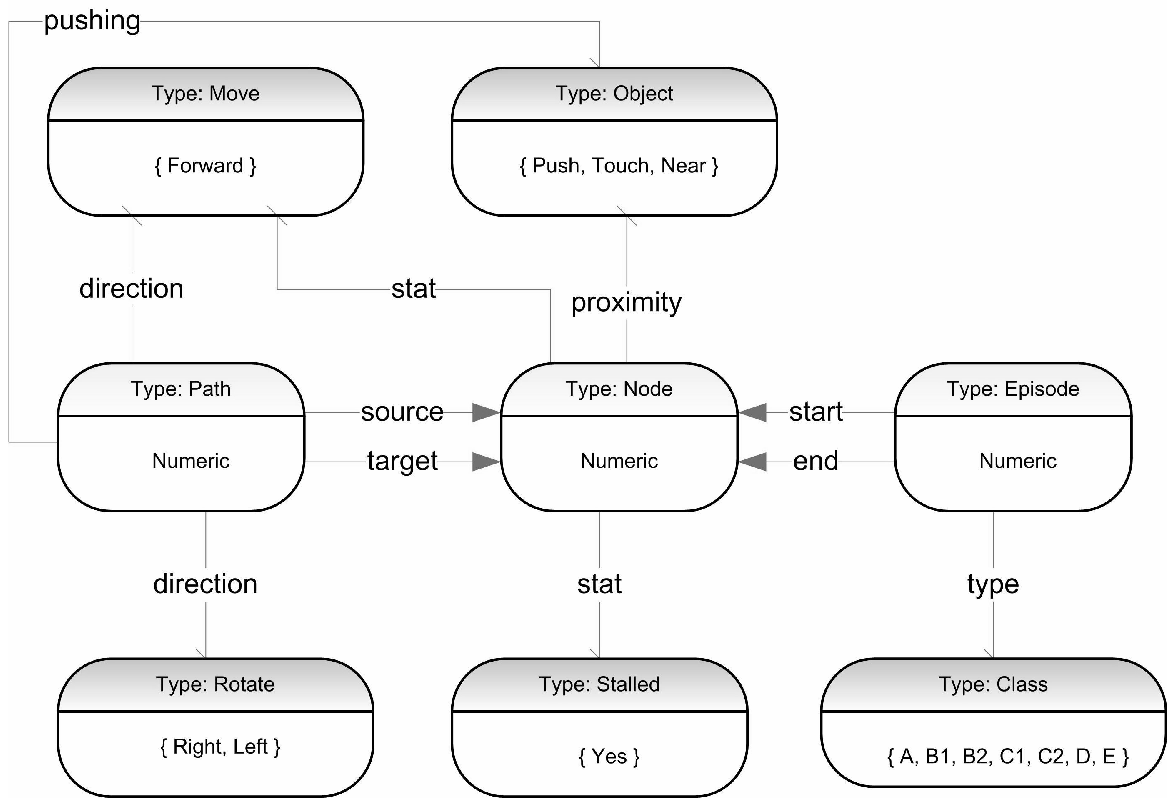}
    \includegraphics[width=150pt]{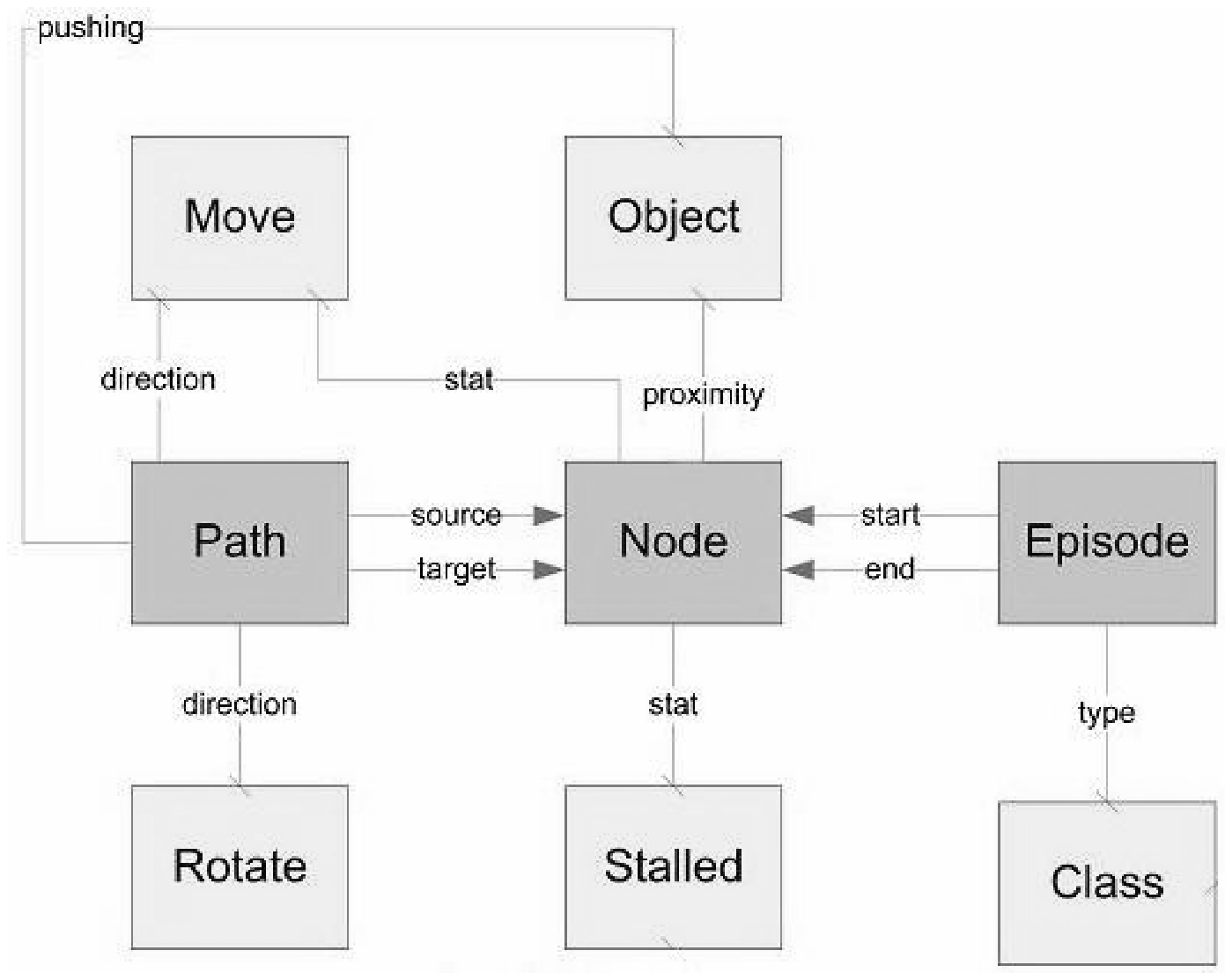}
    \end{center}
    \caption{Parsing Graph and example of a query.}\label{parsinggraph}\label{query1}
\end{figure}

This structure defines a semiotic when we assign to it a model by fixing an interpretation to each sign and for each component. Note that the defined semiotic is consistent with the available data if there is a word in the associated graphic language having as part of its the given dataset. However, for practical proposes, the lack of expressive power for the used language make this notion of consistence to restrictive. We relaxed this by defining what we mean by a semiotic $\lambda$-consistent with the data: a specification expressed in the graphic language, is $\lambda$-consistent with the data, if there is an  interpretation having a part that is "good approximation" to the given data.

A semiotic selected for the parsing graph from fig. \ref{parsinggraph} have the signs interpretation domains equipped with a similarity relation. The limit for the diagram in fig. \ref{query1}, where we identify the diagram sources $\{Move,$ $Object,$ $Rotate,$ $Stalled\}$, the target $\{class\}$ and as auxiliary signs
$\{Path,Node,$ $Episode\}$, is a "good" approximation to the dataset. This limit can be seen as a $\Omega$-set \[\alpha:Move\times Object\times Rotate\times Stalled\times Class.\] The discrepancies between $\alpha$ and the real data must then be seen as information that are not semantically valid for defined semiotic. This type of limit can be seen as a view of the data described by the semiotic. However the information in the generated limit isn't adequate to be used for solve the proposed problem of patterns detection, using machine learning algorithms. It doesn't codify the structure of time series generated by the robot perceptual system, since it doesn't use temporal relation between stats.

We used limits of admissible configuration to extract potential useful information from the universe modeled by the semiotic. The existence of patterns associate to robot stall on first three states of each episode should be detected in the limit for diagram $D(a)$ in fig. \ref{query2}, using the adequate machine learning tools.
\begin{figure}[h]
    \begin{center}
    \includegraphics[width=150pt]{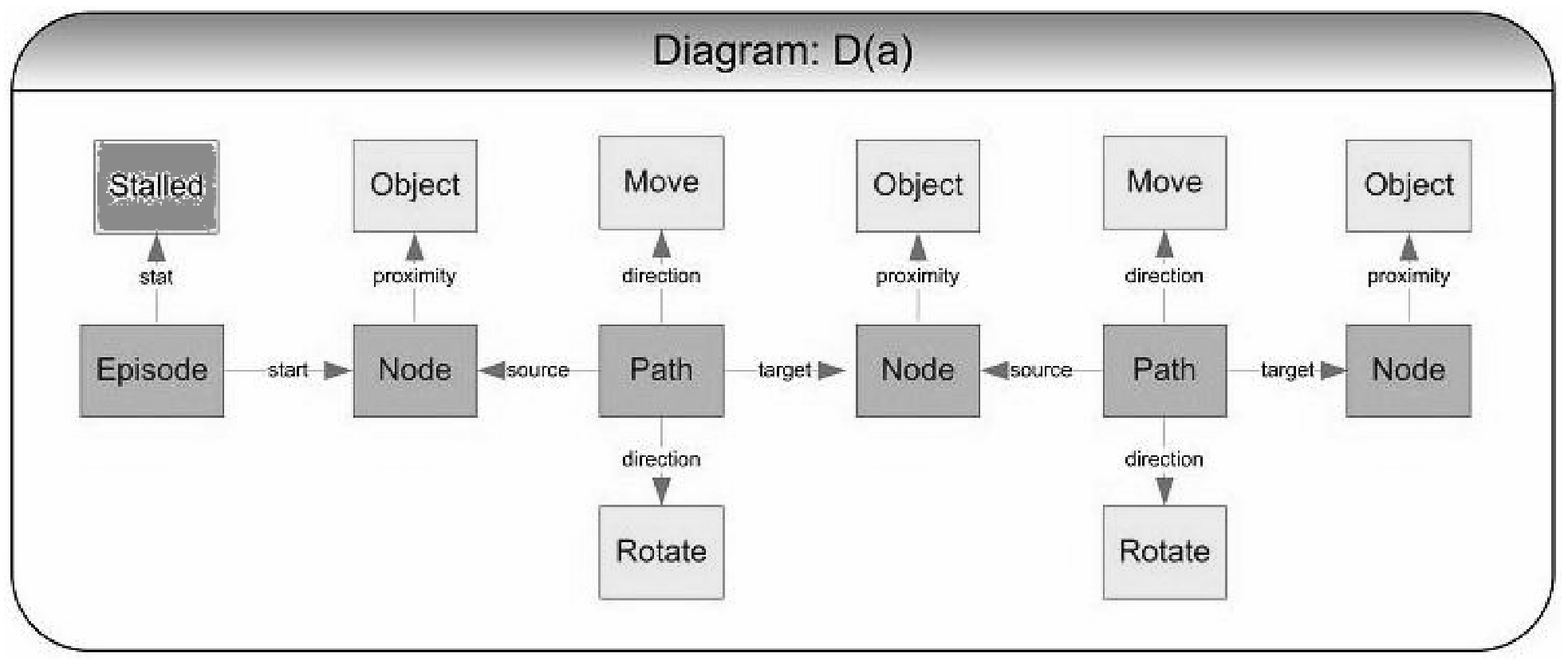}
    \includegraphics[width=170pt]{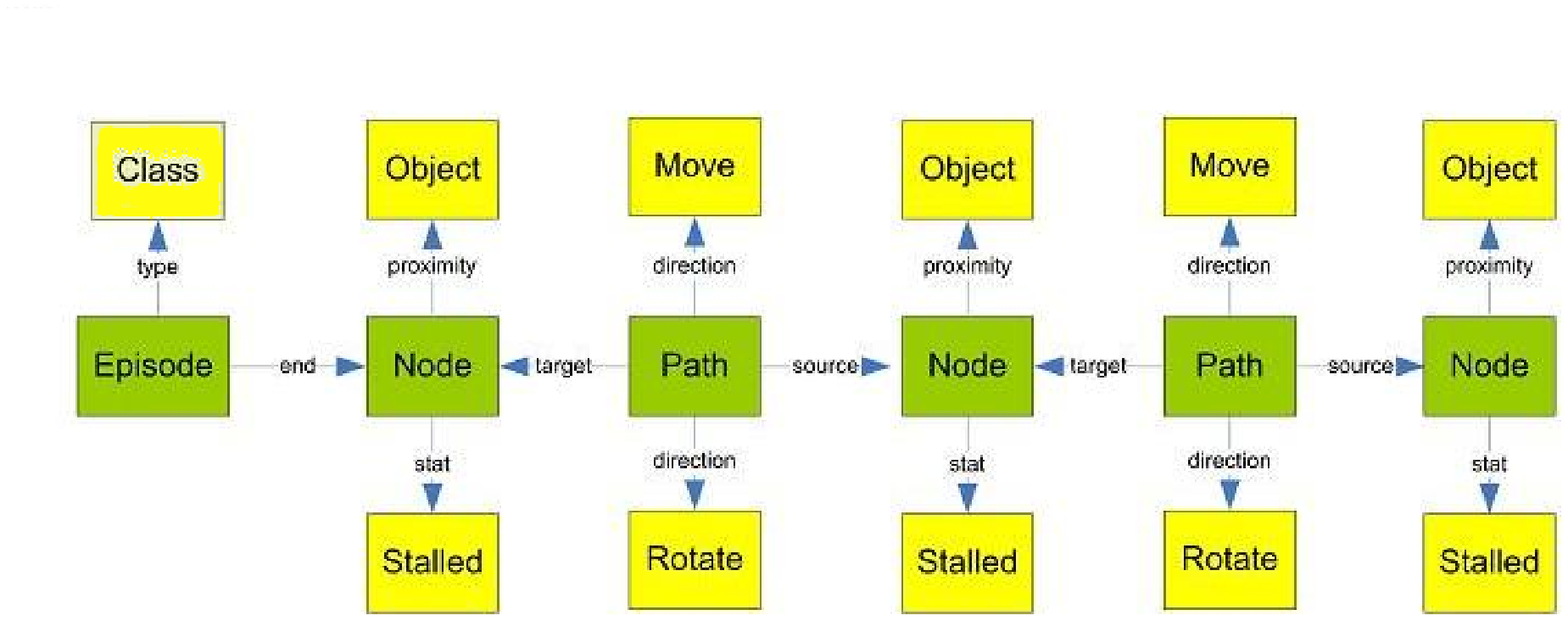}
    \end{center}
    \caption{Queries $D(a)$ and $D(b)$.}\label{query2}\label{query3}
\end{figure}
However, to classify episodes it seems to be more relevant the last robot stats. Patterns of information associated to the last three states of a robot are present in the limit for diagram $D(b)$ from fig. \ref{query3}.

If we wish to use, for episode class prediction, the relational information available for three consecutive states we must change used library. We have a problem, the defined library structure doesn't have sufficient expressive power to describe this type of query. We may improve the library expressive power by adding to the specification system more a component: associating each automata stat to its episode. The query can be defined by diagram $D(c)$ in fig. \ref{query3}.

The update of libraries can be made also to restricted its interpretations. In fig. \ref{enrichsemiotic} we enriched the sign system by imposing a restriction to its interpretation by adding two equalizers. We used them to impose that in a path the source and the target must be different. This can be codify by restricting the equalizer between components \emph{source} and \emph{targets} to a initial relation. If we want also consider only episode having more than a stat the limit of diagram define by component \emph{start} and \emph{end} must be the initial relation. The specification bellow is enriched also with new components $Lim\;D(a)$, $Lim\;D(b)$, $Lim\;D(c)$ and $Lim\;D(a)$ interpreted as the limit of presented queries $D(a)$, $D(b)$, $D(c)$ and $D(d)$, respectively, add as signs interpreted as the source of this new components. This new signs add a new level to the sign ontology. They are more general than the signs defined initially.

\begin{figure}[h]
    \begin{center}
    \includegraphics[width=150pt]{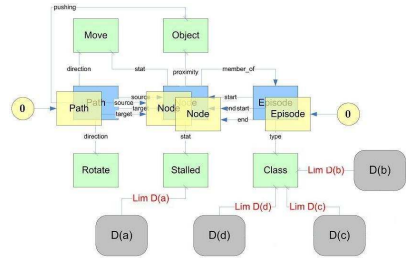}
    \end{center}
    \caption{Enriched sign system.}\label{enrichsemiotic}
\end{figure}

Note what the limit of diagram $D(c)$, presented in fig. \ref{query3}, describe available information about signs having by interpretation three consecutive stats, and the limit of $D(d)$ describes three consecutive stats of episodes such that the robot stall.

\begin{figure}[h]
    \begin{center}
    \includegraphics[width=150pt]{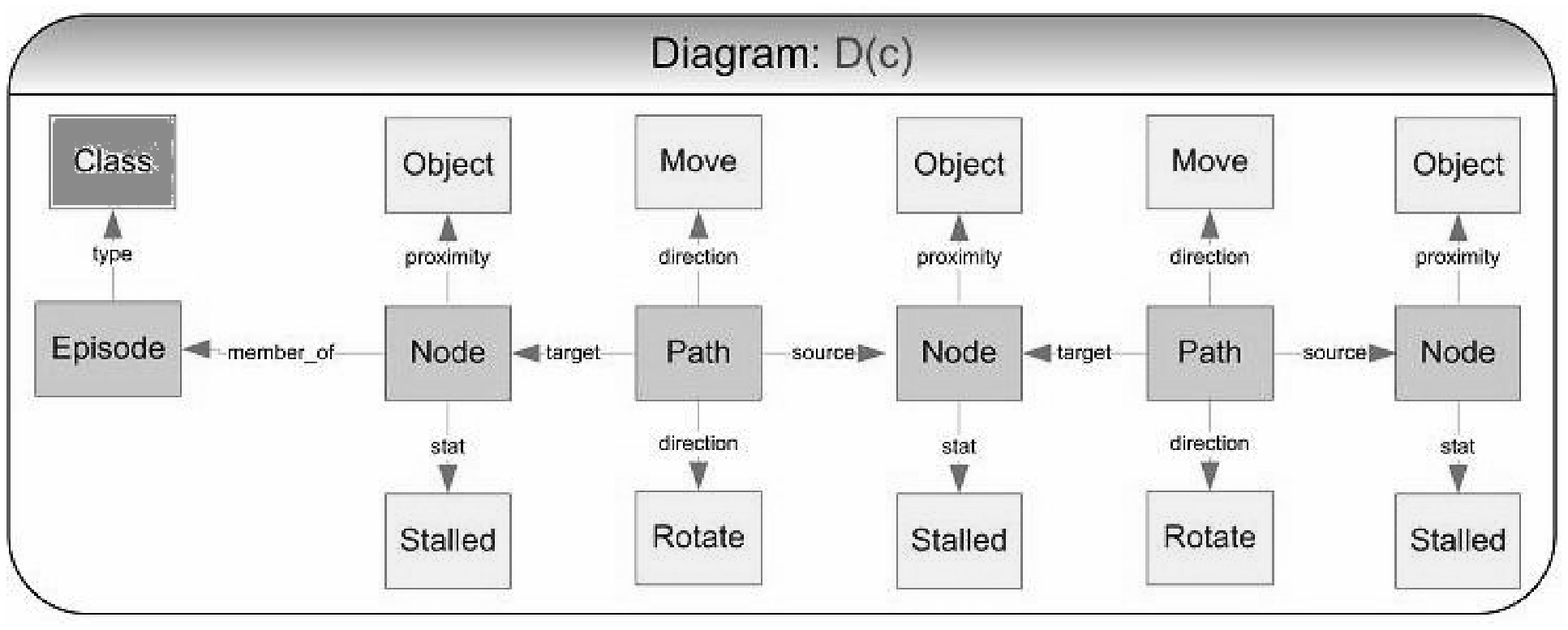}
    \includegraphics[width=170pt]{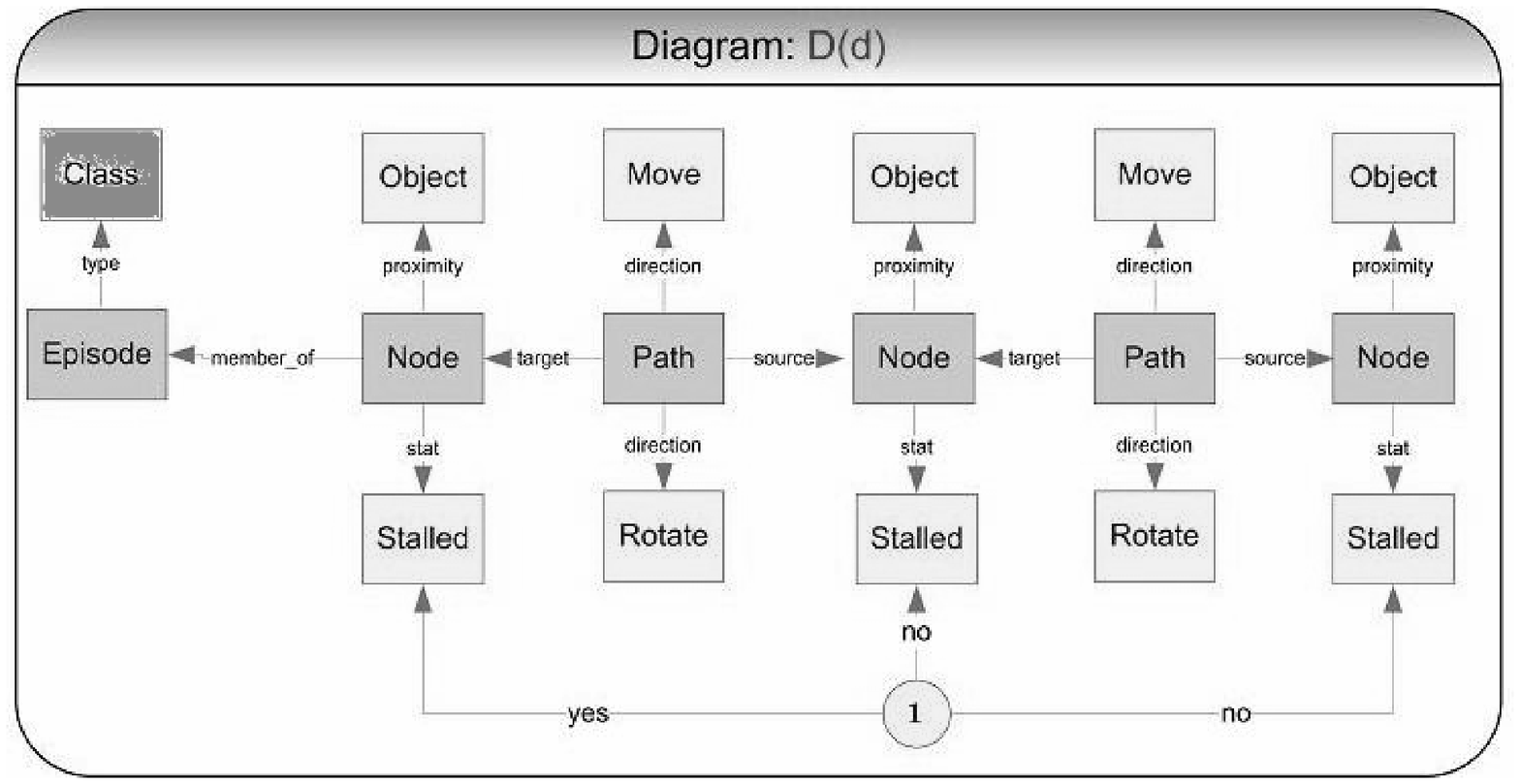}
    \end{center}
    \caption{$D(c)$ querying three consecutive stats and $D(d)$ querying three consecutive stats of episodes where the robot stall.}\label{query4}
\end{figure}
\end{exam}

In the example we used limits as a way to query the structure of a semiotic. On the following we will formalize some of the concepts presented in this example, particulary that we mean by an "approximation" to the given data.

\begin{exam}[Neural Networks \cite{Michell86}]
A neural network is a network of simple processing elements (neurons), which can exhibit complex global behavior, determined by the connections between the processing elements and element parameters.  Formally, a network is a function $f$ defined as a composition of other functions $g_i(x)$, with can further be defined as a composition of other functions. This dependencies can be conveniently represented as a network structure, with arrows depicting the dependencies between variables. What has attracted the most interest in neural network is the possibility of be parameterized to learning a task. Given a specific task to solve, and a class of functions $F$ defined using a network structure and a class os neurons, learning means using a set of observations, in order to find $f^\ast\in F$ with solves the task in a optimal sense.

In this sense we can see an artificial Neural Network as an admissible diagram defined in the semiotic codifying  every possible parametrization of each processing element. More precisely, the associated sign system describe the possible Neural Network structures and a model for it represents a set of parameterizations describing the behave of each processing element.

There are three major learning paradigms, each corresponding to a particular learning task \cite{Bishop96}. These are supervised learning, unsupervised learning and reinforcement learning and can be seen as different ways of search the space defined by models of admissible configurations in order to find the optimal solution, i.e. the model what best fits the data.

Given a diagram $D$ on a Neural Network Semiotic the multi-morphism $Lim\;D$ describe the functional behavior for the network $D$ when applied to its input space. In this sense, a network $D$ learned a concept describe in a dataset if part of its limit $Lim\;D$ is a good approximation to the dataset.
\end{exam}

More examples can be taken from applications of generic programming, see for instance \cite{Fiadeiro97}.

\section{Logics}\label{logics}

We can see a semiotic as a formal way to specify a problem Universe of Discurse. We are particulary interest on the specification of universes where its entities can be characteristics by propositions on monoidal logics. For that we define:
\begin{defn}
A semiotic system $(S,M)$, with
\[S=(L:|L|\rightarrow (Chains\downarrow \Sigma^+),\E,\U,co\U),\]
is a \emph{logic semiotic system}, if:
\begin{enumerate}
  \item exists a sign in $S$ interpreted as the support $\Omega$ of a ML-algebra having as operators interpretations of
component labeled with the signs $\vee ,\wedge ,\otimes,
\Rightarrow ,\bot$  and $\top$,
  \item for every string $w=s_0s_1\ldots s_n$ of sign in $S$ there is a label $=_w$ interpreted by $M$ as the similarity $\bigotimes_i^n[\cdot=\cdot]_i$ where $[\cdot=\cdot]_i$ is the similarity on the $\Omega$-set $M(s_i)$, defined in \ref{ProdSimil},
  \item for every sign $s$ in $S$ and every natural number $n$ there is a
  component in $S$ labeled by $\lhd^n_s$ and interpreted by $M$ as a
  diagonal \[\lhd^n_{M(s)}:M(s)\rightharpoonup
  \prod^n_{i=1}M(s)\text{, and}\]
  given by \[\lhd^n_{M(s)}(a,a_1,a_2,\ldots,a_n)=\bigotimes_i^n[a=a_i].\]
  \item for every sign $s$ in $S$ and every natural number, $n$ there is a
  component in $S$ labeled by $\rhd^n_s$ and interpreted by $M$ as a
  codiagonal relation \[\rhd^n_{M(s)}:\prod^n_{i=1}M(s)\rightharpoonup
  M(s),\] given by \[\rhd^n_{M(s)}(a_1,a_2,\ldots,a_n,a)=\bigotimes_i^n[a_i=a].\]
  \item we suppose the existence of a disjoint decomposition for the set of signs $\Sigma$, given by $\Sigma_{aux}$
  and $\Sigma_{pri}$, where signs in $\Sigma_{aux}$ are called \emph{auxiliary}, and
  for every pair $(s,u)\in \Sigma_{pri}\times \Sigma_{aux}$ there are
  components $r(s,u):su^+$ and $r(u,s):us^+$ in $L$, called \emph{rename component}, and interpreted by $M$ as the
  identity in $M(s)$, i.e. \[M(r(s,u))=id_{M(s)} \text{ and } M(r(u,s))=id_{M(s)}.\]
\end{enumerate}
\end{defn}
In a semiotic logic the signs $\triangleleft$ and $\triangleright$ are used to relate together similar component inputs and similar components outputs. Rename components are used as a mechanism to codify the wires or links between components inputs and outputs on the diagram.

As usual, equations can be specified by commutative diagrams, in fig. \ref{identity} we
specify the property $e_s+x= x+e_s= x$ (existence of identity $e_s$)
using a commutative diagram. A model $M$ for an additive
library $\L_A(S,C)$ defines an additive operator with identity if
the diagram limit defines a total multi-morphism.
\begin{figure}[h]
\[
\small
\xymatrix @=7pt {
 *+[o][F-]{\top} \ar `r[rrrrrr][rrrrrrdd]&  & & & & & & \\
 *+[o][F-]{e_s}\ar[r]& *+[o][F-]{\lhd} \ar[rr] \ar `d[ddrr][ddrr]& &*+[o][F-]{+} \ar `r[rr][rrd]& & & & \\
 x\ar[rr]      & &*+[o][F-]{\lhd} \ar `r[ur][ur] \ar `r[dr][dr] \ar `d[ddrrr][ddrrr]& & &*+[o][F-]{=}\ar[r] &*+[o][F-]{\rhd}  \\
                & & &*+[o][F-]{+} \ar[r]& *+[o][F-]{\lhd}\ar `r[ur][ur] \ar `r[dr][dr]& & & \\
                & & & & &*+[o][F-]{=} \ar `r[ruu][ruu] & &
}
\]
\caption{Diagrams codifying $e_s+x= x+e_s= x$.}\label{identity}
\end{figure}

This diagram presented in fig. \ref{identity} can be codified in string base notation by the chain of signs:

$\lhd^3_sr(s,x)r(s,x)r(s,x)e_s\lhd^2_sr(s,y)r(x,s)+r(s,z)r(x,s)r(y,s)+\lhd^2_s$ $r(s,w)$ $r(s,w)$
$r(z,s)$ $r(w,s)=$ $r(w,s)r(x,s)=\top\rhd^3_\Omega$.

By this we want to say that every diagram defining a word in a graphic language can be codified using string-based notation using rename components.

\begin{exam}
Let $Set(\Omega)$ be defined using the ML-algebra  $\Omega=([0,1],\otimes,\Rightarrow,\vee,\wedge,0,1)$, where  $([0,1],$ $\otimes,$ $\Rightarrow)$ is a model for product logic and $([0,1],\vee,\wedge,0,1)$ the usual complete lattice defined in $[0,1]$ by relation "less than". The $\Omega$-set defined in $A=\{0,1,2\}$ by the similarity relation given by the table
\begin{center}
\small
\begin{tabular}{c|ccc}
  $[\_\;=\;\_]$ & 0 & 1 & 2 \\
  \hline
  0 & 1.0 & 0.5 & 0.0 \\
  1 & 0.5 & 1 & 0.5 \\
  2 & 0.0 & 0.5 & 1.0 \\
\end{tabular}
\end{center}
and the additive relational operator
\begin{center}
\small
\begin{tabular}{c|ccc}
  $\_\;+\;0$ & 0 & 1 & 2 \\
  \hline
  0 & 1.0 & 0.5 & 0.0 \\
  1 & 0.5 & 1 & 0.5 \\
  2 & 0.0 & 0.5 & 1.0 \\
\end{tabular}$\;$
\begin{tabular}{c|ccc}
  $\_\;+\;1$ & 0 & 1 & 2 \\
  \hline
  0 & 0.5 & 1.0 & 0.5 \\
  1 & 0.0 & 0.5 & 1.0 \\
  2 & 1.0 & 0.5 & 0.5 \\
\end{tabular}$\;$
\begin{tabular}{c|ccc}
  $\_\;+\;2$ & 0 & 1 & 2 \\
  \hline
  0 & 0.0 & 0.5 & 1.0 \\
  1 & 1.0 & 0.0 & 0.5 \\
  2 & 0.5 & 1.0 & 0.0 \\
\end{tabular}
\end{center}
define a model for an additive library and the operator have by identity $e_s=0$, since the diagram in fig. \ref{identity} have by limit
\begin{center}
\small
\begin{tabular}{c|c}
  A & $\Omega$ \\
  \hline
  0 & 1.0=[0] \\
  1 & 1.0=[1] \\
  2 & 1.0=[2] \\
\end{tabular}
.
\end{center}
\end{exam}

In a logic semiotic system $(S,M)$ the language $Lang(L)$, is called a
\emph{logic language}. Logic semiotics have sufficient expressive power to distinguish between diagrams defining relations and diagrams defining entities. For that, we classify the words as:

\begin{defn}(Graphic relations)
A diagram $D\in Lang(S)$, for a logic semiotic $(S,M)$, is called a \emph{relation}
when its output $o(D)$ is interpreted by $M$ as the set of truth values $\Omega$ on $Set(\Omega)$.
\end{defn}
\begin{defn}
A relation $D$ is called an \emph{equation} if diagram $D$ can be decomposed
as \[D=I\otimes D_0\otimes D_1\otimes \lq=\lq,\] where $I=\lhd^{n_1}_{s_1} \ldots
\lhd^{n_m}_{s_m}$ is defined through realizations of diagonals,
$D_0$, and $D_1$, are diagrams satisfying $o(D_0)=o(D_1)$ and '=' is a
diagram defined using the unique component, interpreted as a similarity relation, and
satisfying ${i(\lq=\lq)=o(D_0).o(D_1)}$. In this case we simplify notation by denoting the diagram $D$
 as $D_0=D_1$. Note what, diagram $I$ codifies the dependencies between interpretations of signs used as inputs for diagram $D_0\otimes D_1\otimes \lq=\lq$, relating together signs having similar values.
\end{defn}
In definition diagram $I=\lhd^{n_1}_{s_1} \otimes\ldots\otimes\lhd^{n_m}_{s_m}$ captures in a graphic logic the dependence relations defined on string-based logic by repeating bounded variables on a proposition.
\begin{figure}[h]
\[
\begin{picture}(65,65)(0,0)
\multiput(19,33.75)(.04435484,-.03360215){186}{\line(1,0){.04435484}}
\put(27.25,27.5){\line(0,1){41}}
\multiput(27.25,68.5)(-.05536913,-.03355705){149}{\line(-1,0){.05536913}}
\multiput(55.5,29.5)(.034911717,.033707865){623}{\line(1,0){.034911717}}
\multiput(77.25,50.5)(-.035773026,.033717105){608}{\line(-1,0){.035773026}}
\put(25.25,63.75){\line(0,-1){13}}
\multiput(25.25,50.75)(.16666667,.03358209){201}{\line(1,0){.16666667}}
\multiput(58.75,57.5)(-.17487047,.03367876){193}{\line(-1,0){.17487047}}
\multiput(25.25,42.75)(.03125,-1.3125){8}{\line(0,-1){1.3125}}
\multiput(25.5,32.25)(.21644295,.03355705){149}{\line(1,0){.21644295}}
\put(57.25,57.75){\circle*{2.06}} \put(73.25,50){\circle*{2.92}}
\put(25.75,62.75){\circle*{2.06}} \put(25.25,59.75){\circle*{2.06}}
\put(25.5,57){\circle*{2.06}} \put(25.25,54){\circle*{2.24}}
\put(25.75,40.75){\circle*{2}} \put(25.75,37.25){\circle*{1.8}}
\put(25.75,34.75){\circle*{2}} \put(20,61.25){\circle*{2.06}}
\put(19.5,58.25){\circle*{2.24}} \put(19.5,52){\circle*{2}}
\put(19.5,46.5){\circle*{2}} \put(19.5,42){\circle*{2}}
\put(23.25,47){\makebox(0,0)[cc]{$_I$}}
\put(33.75,56.25){\makebox(0,0)[cc]{$_{D_1}$}}
\put(33,37.25){\makebox(0,0)[cc]{$_{D_2}$}}
\put(65.25,49.75){\makebox(0,0)[cc]{$_=$}}
\put(76.5,44.75){\makebox(0,0)[cc]{$_{\Omega}$}}
\put(55.5,70.75){\line(0,-1){40.75}}
\multiput(25.25,43)(.19512195,-.03353659){164}{\line(1,0){.19512195}}
\put(19,63.5){\line(0,-1){29}} \put(19.5,38){\circle*{1.58}}
\put(56.75,38.25){\circle*{1.58}}
\end{picture}
\]
\caption{Sketch for a multi-morphism of type $D=I\otimes D_0\otimes D_1\otimes \lq=\lq$.}\label{identity2}
\end{figure}
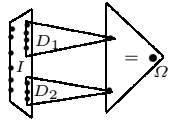
The relation $D\in Lang(S)$ is called \emph{true} by $M$ if
the limit $M(D\otimes\top\otimes\rhd^2_\Omega)$ is a total multi-morphism. In this sense a equation $D$ is universal if the interpretation of $D.\top.\rhd^2_\Omega$ by $M$ is
a total multi-morphism.

Given a logic semiotic system $(S,M)$, let $Lang_R(S,M)$ be the
subcategory of $Lang(S)$ having by objects diagrams defining relations. Using the operation of diagram gluing we define for pairs of relations $D_0,D_1\in
Lang_R(S,M)$ the following operators between diagrams:
\begin{enumerate}
  \item $D_0\otimes D_1$ is the diagram $I\otimes D_0\otimes D_1\otimes \lq\otimes\lq$,
  \item $D_0\Rightarrow D_1$ is the diagram $I\otimes D_0\otimes D_1\otimes \lq\Rightarrow\lq$,
  \item $D_0\wedge D_1$ is the diagram $I\otimes D_0\otimes D_1\otimes \lq\wedge\lq$
  and
  \item $D_0\vee D_1$ is the diagram $I\otimes D_0\otimes D_1\otimes \lq\vee\lq$,
\end{enumerate}
where $I$ is defined through realization of diagonals, linking
together inputs having the some meaning by $M$. This allows the definition of new relations from pairs of simplest ones, and they correspond to the lifting part of $\Omega$ structure to diagrams in $Lang_R(S,M)$.

In the following we present some useful examples of logic semiotics important to caracterize the expressive power of language used by some machine learning algorithms:
\begin{exam}[Binary semiotic $S_B(S,C)$]
A binary semiotic is a logic semiotic where sets $S$ and $C$ define a binary library $\L_B(S,C)$ (presented in example \ref{binarylibrary}). We call to this sort of semiotic dataset semiotic or table semiotic since we can use instantiations of relations in $Lang_R(S_B(S,C))$ to codify datasets or tables.

The use of binary semiotic can be seen in machine learning algorithm used to generate decision rules like Apriori described in \cite{Michell86}.
\end{exam}

\begin{exam}[Linear semiotic $S_L(S,C)$]
A linear semiotic $S_L(S,C)$ extends a binary semiotic. It is defined by a linear library $\L_L(S,C)$ (presented in example \ref{linearlibrary}) where $\geq:ssl^+$ is interpreted as a partial
order in the interpretation of sign $s$, $M(s)$, for all symbol $s\in S$. This type of relation is codified in the model of a linear
semiotic $S_L(S,C)=(L,\E,\U,co\U)$ if it includes in the set $\E$ the diagrams
presented in fig. \ref{linear}, codifying propositions represented in string-based first-order logic by \[\forall x:x\geq
x,\,\,\,\, \forall x,y:(x\geq y\wedge y\geq x) \Rightarrow (x=y)\] and \[\forall
x,y,z:(x\geq y\wedge y\geq z) \Rightarrow (x\geq z)\] and are "preserved" by its models.
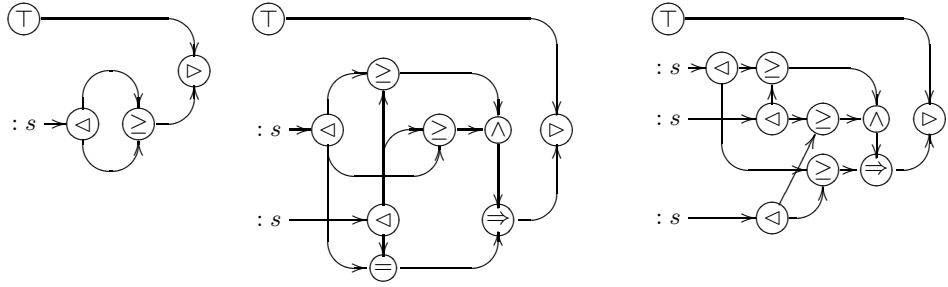
\begin{figure}[h]
\[
\small
\xymatrix @=8pt {
 *+[o][F-]{\top}\ar `r[rrrd][rrrd]&&&\\
 &&&*+[o][F-]{\rhd}&\\
 :s\ar[r]&*+[o][F-]{\lhd}\ar `u[ru]`r[rd][r] \ar `d[rd]`r[ru][r]&*+[o][F-]{\geq}\ar `r[ru][ru]&&\\
 &&&&
}
\xymatrix @=8pt{
 *+[o][F-]{\top}\ar `r[rrrrrdd][rrrrrdd]&&&&&&&\\
&&*+[o][F-]{\geq}\ar `r[rrd][rrd] &&&&&\\
:s\ar[r]&*+[o][F-]{\lhd}\ar `d[dddr][dddr] \ar `u[ur] [ur] \ar `d[drr] `r[rru] [rr]& &*+[o][F-]{\geq} \ar [r]&*+[o][F-]{\wedge} \ar [dd]&*+[o][F-]{\rhd}&&\\
&&&&&&&\\
:s\ar[rr]&&*+[o][F-]{\lhd}\ar[d] \ar[uuu] \ar `u[uur] [uur] &&*+[o][F-]{\Rightarrow} \ar `r[ruu][ruu]&&&\\
&&*+[o][F-]{=}\ar `r[rru] [rru]&&&&&\\
}
\xymatrix @=6pt{
 *+[o][F-]{\top}\ar `r[rrrrrdd][rrrrrdd]&&&&&&&\\
:s\ar[r]&*+[o][F-]{\lhd}\ar[r] \ar `d[ddrr][ddrr] &*+[o][F-]{\geq}\ar `r[rrd][rrd] &&&&&\\
:s\ar[rr]&&*+[o][F-]{\lhd}\ar[u]\ar[r] &*+[o][F-]{\geq} \ar [r]&*+[o][F-]{\wedge} \ar [d]&*+[o][F-]{\rhd}&&\\
&&&*+[o][F-]{\geq}\ar[r]&*+[o][F-]{\Rightarrow} \ar `r[ru][ru]&&&\\
:s\ar[rr]&&*+[o][F-]{\lhd} \ar[uur] \ar `r[ru] [ru]&&&&&&\\
}
\]
\caption{Diagrams codifying $\forall x:x\geq
x,\,\,\,\, \forall x,y:(x\geq y\wedge y\geq x) \Rightarrow (x=y)$ and $\forall
x,y,z:(x\geq y\wedge y\geq z) \Rightarrow (x\geq z)$.}\label{linear}
\end{figure}

 Lets present an example on  $Set(\Omega)$ having $\Omega$ the a structure of ML-algebra where  $([0,1],\otimes,\Rightarrow)$ is a model for product logic and $([0,1],\vee,\wedge,0,1)$ the usual complete lattice defined in $[0,1]$. For the $\Omega$-set defined with support $A=\{0,1,2,3,4\}$ and the similarity relation
\begin{center}
\small
\begin{tabular}{c|ccccc}
  $[\_\;=\;\_]$ & 0 & 1 & 2 & 3 & 4 \\
  \hline
  0 & 1.0 & 0.5 & .25 & 0.0 & 0.0 \\
  1 & 0.5 & 1.0 & 0.5 & .25 & 0.0 \\
  2 & .25 & 0.5 & 1.0 & 0.5 & .25 \\
  3 & 0.0 & .25 & 0.5 & 1.0 & 0.5 \\
  4 & 0.0 & 0.0 & .25 & 0.5 & 1.0 \\
\end{tabular}
\end{center}
and the relational operator defined form $A$ to $A$ by:
\begin{center}
\small
\begin{tabular}{c|ccccc}
  $[\_\;\geq\;\_]$ & 0 & 1 & 2 & 3 & 4 \\
  \hline
  0 & 1.0 & 0.5 & .25 & 0.0 & 0.0 \\
  1 & 1.0 & 1.0 & 0.5 & .25 & 0.0 \\
  2 & 1.0 & 1.0 & 1.0 & 0.5 & .25 \\
  3 & 1.0 & 1.0 & 1.0 & 1.0 & 0.5 \\
  4 & 1.0 & 1.0 & 1.0 & 1.0 & 1.0 \\
\end{tabular}
\end{center}
When this relations are used for the sign interpretation in the three diagrams presented on fig. \ref{linear} they have by limit, respectively,
\[
\small
\begin{tabular}{c|c}
  $A$ & $\Omega$ \\
  \hline
  0 & 1 \\
  1 & 1 \\
  2 & 1 \\
  3 & 1 \\
  4 & 1 \\
\end{tabular}\;\;
\begin{tabular}{c|c}
  $A\times A$ & $\Omega$ \\
  \hline
  (0,0) & 1 \\
  (1,0) & 1 \\
  $\vdots$ & $\vdots$ \\
  (3,4) & 1 \\
  (4,4) & 1 \\
\end{tabular}\;\;\text{ and }
\begin{tabular}{c|c}
  $A\times A\times A$ & $\Omega$ \\
  \hline
  (0,0,0) & 1 \\
  (1,0,0) & 1 \\
  $\vdots$ & $\vdots$ \\
  (3,4,4) & 1 \\
  (4,4,4) & 1 \\
\end{tabular}
.
\]
Which grants the commutativity of each diagram.

We call to this type of semiotics \textbf{grid semiotics}. They appear associated algorithms of machine learning used to generate decision rules  like the C4.5Rules of J.R. Quinlan see \cite{Michell86}.
\end{exam}

\begin{exam}[Additive semiotic $S_A(S,C)$]
A \emph{additive semiotic} $S_A(S,C)$ is a linear semiotic
$S_L(S,C)$ such that it is a additive library $\L_A(S,C)$ (presented in example \ref{addlibrary}) where
$(M(s),M(+:sss^+))$ is a monoid, for all symbol $s\in S$, and the interpretation for $e_s:s^+$ is the monoid identity, with $e_s\in C$. Monoid proprieties can be codified in the semiotic structure if the model
transform each of the diagrams presented in fig. \ref{additive} in a total multi-morphism, for each sign $s$ in the semiotic.
\begin{figure}[h]
\[
\small
\xymatrix @=7pt {
 *+[o][F-]{\top}\ar `r[rrrrrdd]'[rrrrrdd]&&&&&&&\\
:s\ar[r]&*+[o][F-]{\lhd}\ar[r] \ar `d[ddrr]`r[rru][drr] &*+[o][F-]{+}\ar `r[rrd][rrd] &&&&&\\
:s\ar[rr]&&*+[o][F-]{\lhd}\ar[u]\ar[r] &*+[o][F-]{+} \ar [r]&*+[o][F-]{=} \ar [r]&*+[o][F-]{\rhd}&&\\
&&&&&&&\\
}  \xymatrix @=7pt {
 *+[o][F-]{\top}\ar `r[rrrd]'[rrrd]&&&&\\
*+[o][F-]{e_s}\ar[r]&*+[o][F-]{+}\ar[r]&*+[o][F-]{=}\ar[r]&*+[o][F-]{\rhd}&\\
:s\ar[r]&*+[o][F-]{\lhd}\ar [u]\ar `r[ru][ru]&&&\\}
\xymatrix @=6pt {
 *+[o][F-]{\top}\ar `r[rrrrrrd][rrrrrrd]&&&&&&&\\
:s\ar[r]&*+[o][F-]{\lhd}\ar[r] \ar `d[dddrr]`r[rrrru][drrrr] &*+[o][F-]{+}\ar [rr] &&*+[o][F-]{+}\ar[r]&*+[o][F-]{=}\ar[r]&*+[o][F-]{\rhd}&\\
:s\ar[rr]&&*+[o][F-]{\lhd}\ar[u]\ar[r] &*+[o][F-]{+} \ar [rr]& &*+[o][F-]{+}\ar[u]&&&\\
:s\ar[rrr]&&&*+[o][F-]{\lhd}\ar[u]\ar`r[ru][ruu] &&&&\\
&&&&&&&\\
}
\]
\caption{Diagrams codifying $\forall x,y:x+y=y+x, \forall x:x+e_s=x$ and $\forall
x,y,z:(x+y)+z=x+(y+z)$.}\label{additive}
\end{figure}
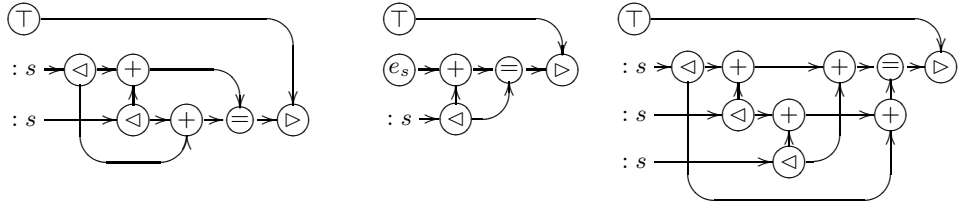
\end{exam}

\begin{exam}[Multiplicative Semiotic $S_M(S,C)$]
A \emph{multiplicative semiotic} $S_M(S,C)$ is an additive semiotic
$S_L(S,C)$ defined by a multiplicative library $\A_M(S,C)$ (presented in example \ref{multlibary}) where
$(M(s),\times)$ is a commutative semigroup, for all symbol
$s\in S$.

In this semiotics we can codify rules defined using regression like the rules generated by machine learning algorithms like M5, of J.R. Quinlan, described in \cite{Quinlan93}.
\end{exam}

The definition of Domain of Discurse structure may impose restrictions to signs interpretations. In a binary semiotic we may impose rules for sign interpretation defined by Horn clauses of type:
\[(x_1=c_1\wedge x_2=c_2 \wedge x_3=c_3) \Rightarrow y=c_4.\]
The expressive power of linear semiotic allows the codification of rules like
\[(x_1\leq c_1\wedge x_2\leq c_2 \wedge c_3\leq x_1) \Rightarrow
y\leq c_4,\] and on multiplicative semiotic sign interpretation can be restricted to semiotics defined by models satisfying regression rules like:
\[(x_1\leq c_1\wedge x_2\leq (c_3\times x_1+c_2) \wedge x_3=c_3) \Rightarrow
y\leq c_4\times x_1+c_5\times x_2+c_6\times x_3+c_7.\] This type of regression rules can be codified by
diagrams like the one represented on fig. \ref{regrassion}, where frames present the obvious
subdiagrams.
\begin{figure}[h]
\[
\small
\xymatrix @=12pt {
*+[o][F-]{\top}\ar `r[rrrrrrd][rrrrrrd]&&&&\\
x_1:s\ar[r]&*+[o][F-]{\lhd}\ar[rr] \ar `d[r]`/1pt[rrd][drr] \ar `d[dddrrrrr][dddrrrrr]&&*+[F-]{_{x_1\leq c_1}}\ar `r[rrd][rrd]&&&*+[o][F-]{\rhd}&\\
x_2:s\ar[rr]&&*+[o][F-]{\lhd}\ar[r]\ar `d[ddrrrr][ddrrrr]&*+[F-]{_{c_3x_1+c_2\leq x_2}}\ar[rr]&&*+[o][F-]{\wedge}\ar[d]&\\
x_3:s\ar[rrr]&&&*+[o][F-]{\lhd}\ar[r]\ar `d[drrr][drrr]&*+[F-]{_{x_3=c_3}}\ar[r]&*+[o][F-]{\wedge}\ar[r]&*+[o][F-]{\Rightarrow}\ar[uu]\\
y:s\ar[rrrrrr]&&&&&&*+[F-]{_{y=c_4\times x_1+c_5\times x_2+c_6\times
x_3+c_7}}\ar[u]&}
\]
\caption{Diagram codifying a regression rule.}\label{regrassion}
\end{figure}
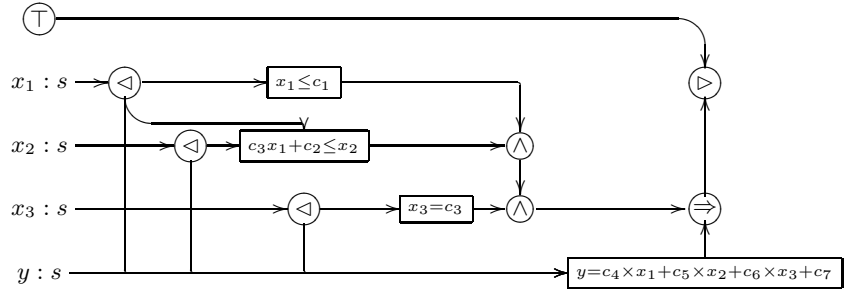

Models generated by some of the Data Mining and Machine Learning tools can be codified in one of this sign systems. By this we mean that we can extract rules from models generated by learning algorithms, what reflect the available data structure. This type of structure can be used on the enrichment of the sign system, employed to specify the information system, allowing the definition of constrains to its models compatible with, the existent data or views for, the reality. By this we want to say that the best description we may have for an information system is the best generalization available for the stored data.

\section{Lagrangian syntactic operators}\label{synopt}
The expressive power of our specification languages can be increased using \emph{Lagrangian syntactic operators} or \emph{sign operators}. This operator are defined at the level of sign systems signs or/and components, and must be preserved via sign systems models, allowing the generation of new signs or components based on others signs or components.

An example of this operators, with evident applicability, are the differential operators. For that we define:

\begin{defn}[Differential semiotic]
A logic semiotic system $(S,M)$ is called a \emph{differential semiotic
system} if it is a multiplicative semiotic where exists a sign $R$ interpreted as the support for a ring
$(M(R),\times,+,0,1)$, and labels $\partial_wf$, in $S$, for components $f:i(f)\rightarrow R$  in $S$ with output $R$, and $w$ a word over its input symbols such that:
\begin{enumerate}
 \item $w\leq i(f)$,
 \item $\partial_wf:d(f)\rightarrow R$,
 \item and the following for component label defined below, using $'\times'$ and $'+'$ on infix notation, we must have:
  \begin{enumerate}
    \item if $f\equiv_l d_0\times d_1$ then $M(\partial_w(d_0\times d_1))=M(\partial_w(d_0)\times d_1+d_0\times \partial_w(d_1))$,
    \item if $f\equiv_l d_0+d_1$ then $M(\partial_w(d_0+d_1))=M(\partial_w(d_0)+\partial_w(d_1))$, and
    \item if $ww'\leq i(f)$ then $M(\partial_{ww'}f)=M(\partial_{w}(\partial_{w'}f))=M(\partial_{w'}(\partial_{w}f))$
  \end{enumerate}
\end{enumerate}
\end{defn}

The component label operator $\partial$ allows the characterization of
multi-morphisms impossible of represente on the associated logic
semiotic. For instance, in a differential semiotic system we may interpret a component $f:xy\rightarrow R$ as a
multi-morphism satisfying the conservative law when the following
diagram is interpreted by the model as a total multi-morphism.
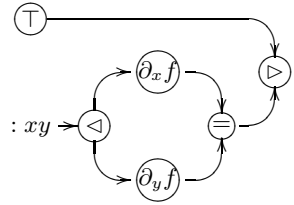
\begin{figure}[h]
\[
\small
\xymatrix @=7pt {
 *+[o][F-]{\top}\ar `r[rrrrd][rrrrd]&&&&\\
 &&*+[o][F-]{\partial_xf}\ar `r[rd][rd]&&*+[o][F-]{\rhd}&\\
 :xy\ar[r]&*+[o][F-]{\lhd}\ar `u[ru][ru] \ar `d[rd][rd]&&*+[o][F-]{=}\ar `r[ru][ru]&&\\
 &&*+[o][F-]{\partial_yf}\ar `r[ru][ru]&&&
}
\]
\caption{Conservative law in a diferencial semiotic.}\label{conservative}
\end{figure}

Lets  see a naif application:
\begin{exam}[Modeling traffic]
In Lighthill-Whithams-Richards (LWR) model (presented in \cite{Lighthill55} and \cite{Richards56}), the traffic state is represented from a macroscopic point of view by the function $\rho(x,t)$ which represents the density of vehicles at position $x$ and time $t$. The dynamics of the traffic are represented by a conservation law expressed as
\[
\frac{\partial \rho}{\partial t}+\frac{\partial \rho v}{\partial x}=0,
\]
where $v=v(x,t)$ is the velocity of cars at $(x,t)$. The main assumption of LWR model is that the drivers instantaneously adapt their speed in function of the surrounding density:
\[v(x,t)=V(\rho(x,t))
\] the function $f(\rho)=\rho V(\rho)$ is then the "flow rate" representing the number of vehicles per time unit. We have then
\[
\frac{\partial \rho}{\partial t}+\frac{\partial f(\rho)}{\partial x}=0.
\]
The model is defined for a single unidirectional road. And it define a diferencial semiotic having the base library presented in fig. \ref{baselibary}
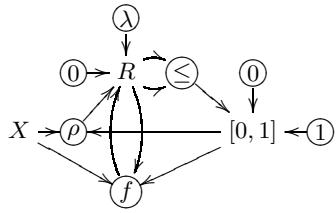
\begin{figure}[h]
\[
\xymatrix @=8pt {
\tiny
 &&*+[o][F-]{\lambda}\ar[d]&&&\\
 &*+[o][F-]{0}\ar[r]&R \ar@/^/[r] \ar@/_/[r]\ar@/^/[dd]&*+[o][F-]{\leq}\ar[dr]&*+[o][F-]{0}\ar[d]&\\
 X\ar[r]\ar[rrd]&*+[o][F-]{\rho}\ar[ru]&&&[0,1]\ar[lll]\ar[lld]&*+[o][F-]{1}\ar[l]\\
 &&*+[o][F-]{f}\ar@/^/[uu]&&&\\
 }
\]
\caption{The library for a sing unidirectional road model.}\label{baselibary}
\end{figure}
and where the associated sign system have by total diagrams
\[
\partial_t\rho + \partial_xf(\rho)=0 \text{ and } 0\leq \rho \leq\lambda,
\] where the least condition describes the road maximal density. A model for this sign system can be seen as an admissible car distribution on the road.

The above library doesn't has descriptive power to specify a road network. In fig. \ref{networklibrary}, we present an extension to the initial library. This new library allows the graphic modeling of a singles network
\begin{figure}[h]
\[
\xymatrix @-1pc {
\ar[rrd]_{X_1}&&&&\\
  &&\ar[rr]_{X_3}&&\\
\ar[rru]_{X_2}&&&&
}
\]
\caption{A road network defined by the junction of two incoming roads and one outgoing road.}\label{network}
\end{figure}
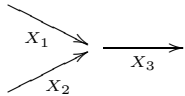
defined by the junction of two incoming roads $X_1,X_2$ and one outgoing road $X_3$ with single direction. The semiotic of this problem is defined using the library from fig. \ref{networklibrary}
\begin{figure}[h]
\[
\small
\xymatrix @=8pt {
 &*+[o][F-]{\lambda_1}\ar[rd]&*+[o][F-]{\lambda_2}\ar[d]&*+[o][F-]{\lambda_3}\ar[ld]&&\\
 &*+[o][F-]{0}\ar[r]&R \ar@/^/[r] \ar@/_/[r]&*+[o][F-]{\leq}\ar[dr]&*+[o][F-]{0}\ar[d]&\\
 X_1\ar[r]\ar[rdd]&*+[o][F-]{\rho_1}\ar[ru]&&&[0,1]\ar[lll]\ar[dd]\ar[ddl]&*+[o][F-]{1}\ar[l]\\
 &&&&&\\
 &*+[o][F-]{c}\ar[rrruu]\ar[rrrd]&&*+[o][F-]{\rho_2}\ar[luuu]&*+[o][F-]{\rho_3}\ar[lluuu]&\\
 X_2\ar[rrru]\ar[ru]&&&&X_3\ar[u]&
}
\]
\caption{The library for the presented road network.}\label{networklibrary}
\end{figure}
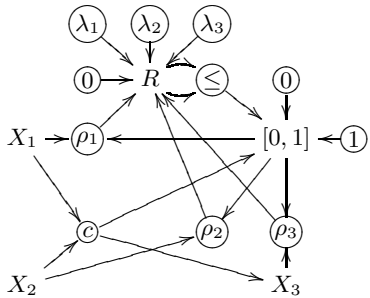
where we also add three "flow rate" components \[f_1:X_1\times R \times[0,1]\rightharpoonup R,\,\,f_2:X_2\times R \times[0,1]\rightharpoonup R\text{ and }f_3:X_3\times R \times[0,1]\rightharpoonup R\] one for each roads and initial condition constants $\rho_{1,0},\rho_{2,0},\rho_{3,0}:\bot\rightharpoonup R$. Network traffic model restrictions are described by the following conservative laws and flow restrictions in each roads:
\begin{enumerate}
  \item $\partial_t\rho_1 + \partial_xf_1(\rho_1)=0$  and $0\leq \rho_1 \leq\lambda_1$,
  \item $\partial_t\rho_2 + \partial_xf_2(\rho_2)=0$  and $0\leq \rho_2 \leq\lambda_2$, and
  \item $\partial_t\rho_3 + \partial_xf_3(\rho_3)=0$  and $0\leq \rho_3 \leq\lambda_3$.
\end{enumerate}
In order to complete the model description, we define the mechanism that occurs at the junction. A first condition is the conservation of flows
\[
f_1(\rho_1(0,t))+f_2(\rho_2(0,t))=f_3(\rho_3(0,t)),\;\forall t,
\]
One of the elementary problem we can study, and from which a model exists, is the Riemann problem. For a Riemann problem at a junction, we take as initial condition a constant density on the three roads:
\begin{enumerate}
  \item $\rho_1(x,0)=\rho_{1,0}$,
  \item $\rho_2(x,0)=\rho_{2,0}$, and
  \item $\rho_3(x,0)=\rho_{3,0}$.
\end{enumerate}

Models for this sign system may be totally unrealistic. For instance, constants,  $\rho_1 = \lambda_1$, $\rho_2 = \lambda_2$ and $\rho_3 = 0$ can be associated to a possible semiotic, however this model although is clearly counterintuitive (except obviously in the presence of a red light at the entrance of the third road). A natural way to have a realistic model, for a particular road junction, is by adding rules describing the behavior of the drivers at the junction. We may enriched the sign system by adding extra information using models generated by machine learning algorithms for the data. This requires the integration of the defined semiotic with a semiotic associated to the language used in the description of machine learning model, problem described in the following.
\end{exam}

\section{Concept description}\label{Concept description}
 A \emph{concept description}, using attributes $(A_i)$,  on the logic semiotic system $(S,M)$ is a $\Omega$-map
\[d:\prod_i A_i\rightarrow\Omega,\]
where the family $(A_i)$ is a sequence of $\Omega$-sets. A concept description is defined in a semiotic $(S,M)$ if the sequence $(A_i)$ is defined using interpretation of signs in the language $Lang(S)$. For short, we write $d\in M(S)$ when we  want to select a concept description in the semiotic $(S,M)$. Note what, $d$ defines a relation in a monoidal logic and it may not be the interpretable by $M$ of a relation in $Lang(S)$. Intuitively, a concept description can be seen as the state of knowledge about a concept at a given moment. A database specified by the sign system $S$ can codify the concept $d$ if there is a model $M\in Mod(S)$ and a diagram $D\in Lang(S)$ such that
\[M(D)=d.\]
In this case  we say that the query $D$ on the information system defined by semiotic $(S,M)$ have by answer $d$.

Given two concept descriptions \[d_0:\prod_i A_i\rightarrow\Omega \text{ and } d_1:\prod_i A_i\rightarrow\Omega,\] we write $d_0\leq d_1$ if $d_0(\bar{x})\leq d_1(\bar{x})$ for every $\bar{x}$ in $\Pi A_i$. And we should note that every $\Omega$-set $A$ defines a concept description by the extend map $[\cdot]:A\rightarrow \Omega$. In this sense we will see every $\Omega$-set as a concept description. And in $Set(\Omega)$ every $\Omega$-set with support $\Pi A_i$ have associated a complete lattice, of concept descritivos, having by bottom $\bot:\Pi A_i\rightarrow \Omega$ and by top $\top:\Pi A_i\rightarrow \Omega$. Particularly, the limit $M(D)$ is a concept description for every $D\in Lang(S)$.

Some concept descriptions can be codified by a semantic, others doesn't. Given a pair of concept descriptions \[d_0:\Pi
A_i\rightarrow\Omega\text{ and }d_1:\Pi A_i\rightarrow\Omega,\]
we define
\[\Gamma(d_0,d_1)=d_0\Leftrightarrow d_1.\] The biimplication $\Gamma$ is a $\otimes$-\emph{similarity
relation} in $\Pi A_i$ since:
\begin{enumerate}
  \item $\Gamma(d_0,d_0)=\top$ (reflexivity)
  \item $\Gamma(d_0,d_1)=\Gamma(d_1,d_0)$ (symmetry)
  \item $\Gamma(d_0,d_1)\otimes\Gamma(d_1,d_2)=\Gamma(d_0,d_2)$
  (transitivity) (by proposition \ref{prop:implic})
\end{enumerate}
When $\Omega=\{\bot,\top\}$ is a two valued logic, $\Gamma$ is clearly an equivalence relation
on $\Pi A_i$.

\begin{defn}
Given a semiotic $(S,M)$, and a concept $d$. A diagram $D$ of a $\lambda$-codification or a $\lambda$-description for $d$ if
\[
\Gamma(d,M(D))\geq \lambda.
\]
In this case we also say that $MD$ is an approximation to the concept $d$. In this sense, relation $D$ is an hypothesis describing the concept presented by $d$, selected on language $Lang(S)$.
\end{defn}

This definition may be extended to concept descriptions having different support sets. Given concept descriptions
\[d_I:\Pi_{i\in I} A_i\rightarrow\Omega\text{ and }d_J:\Pi_{j\in J} A_j\rightarrow\Omega,\] such that exist a projection map $\pi:\Pi_{i\in I} A_i\rightarrow\Pi_{j\in J} A_j$, we write $d_J\preceq d_I$, we call $d_J$ a $\lambda$-projection of $d_I$ if
\[
\Gamma(\pi\otimes d_I, d_J)\geq \lambda.
\]

Given a concept description $d$ and an hypothesis $D$ in $Lang(S)$ the quality of $D$ as a description for $d$ is given by:
\[[d=D]=\bigwedge_{\bar{x}}\Gamma(d,D)(\bar{x}),\]
where
\begin{enumerate}
  \item $\Gamma(d,D)(\bar{x})=(\pi\otimes M(D) \Leftrightarrow d)$ if $d\preceq M(D)$:
  \item $\Gamma(d,D)(\bar{x})=( M(D) \Leftrightarrow \pi\otimes d)$ if $M(D)\preceq d$.
\end{enumerate}

We see a model as the fuzzy answers to a query on a semiotic for what we define:
\begin{defn}\label{def:lambda model}
A concept $d$ is a $\lambda$-model for $D$ in $Lang(S)$ if $d\preceq M(D)$ or $M(D)\preceq d$ and the diagram presented in fig. \ref{lambdamodel}
is a pullback such that \[\Gamma_\lambda(d,D)=\Pi A_i,\] where $[\lambda,\top]$ is a chain on lattice $\Omega$. In this case we write
\[d\models_\lambda \forall D.\]
\end{defn}
\begin{figure}[h]
\[
\xymatrix @-1pc { \ar @{} [dr] |{\lrcorner}
 \Gamma_\lambda(d,D) \ar[rr] \ar[dd]^{\subset} &&[\lambda,\top] \ar[dd]^{\subset}\\
 &&\\
 \Pi A_i \ar[rr]_{\Gamma(d,D)} &&\Omega
 }
\]
\caption{Diagram evaluation.}\label{lambdamodel}
\end{figure}

If in the above pullback diagram we have
\[\Gamma_\lambda(d,D)\subseteq \Pi A_i \text{ and } \Gamma_\lambda(d,D)\neq \emptyset\]
we write
\[d\models_\lambda \exists D,\]
or, when we want be more formal,
\[d\models_\lambda \forall_B D,\]
where $B=\Gamma_\lambda(d,D)$. This notation is also used as $d\models_\lambda \forall_C D$ when $C\subseteq \Gamma_\lambda(d,D)$.

When $d\models_\top\forall D$, we write $d\models\forall D$,  and $d$ can be seen as
the answer to the query $D$ on the information system given by $(S,M)$. Similarly, if $d\models_\top\exists D$, we write $d\models\exists D$, part of $d$ is $\lambda$-consistente with interpretation for $D$ in the semiotic $(S,M)$.

Note what, $\forall_B D$ may not be seen as a formula on the language associated to the used semiotic. Because the language may not have sufficient expressive power to define $B$. However if \[B=\Gamma_\lambda(d,D)=\Gamma_\top(\chi_B,D'),\] i.e. if domain $B$ can be specified using diagram $D'$ in the language we write \[d\models_\lambda \forall_{D'} D.\]
Note what, in this case, for every description $d$ we have
\[
d\models_\lambda \forall_{D'} D\Leftrightarrow d\models_\lambda \forall D'\Rightarrow D.
\]
When $d_0\models_{\lambda_0} \forall_{B_0} D_0$ and $d_1\models_{\lambda_1} \forall_{B_1} D_1$ we have
\[d_0\otimes d_1\models_{\lambda_0\otimes\lambda_1} \forall_{B_0\cap B_1} D_0\otimes D_1,\]
\[d_0\vee d_1\models_{\lambda_0\vee\lambda_1} \forall_{B_0\cap B_1} D_0\vee D_1,\]
\[d_0\Rightarrow d_1\models_{\lambda_0\Rightarrow\lambda_1} \forall_{B_0\cap B_1} D_0\Rightarrow D_1.\]
From the proposed definition every diagram has a $\lambda$-model since:

\begin{prop}
In a logic semiotic $(S,M)$ if $D$ is a relation defined on language $Lang(S)$ then
\[MD\models \forall D.\]
\end{prop}

And from the presented notion of similarity, defined by biimplication, we have also as $\lambda$-models for $D$ concepts $\lambda$-similar to its interpretation $MD$:

\begin{prop}
If $\Gamma(d,MD)\geq \lambda$ then $d\models_\lambda \forall D$.
\end{prop}

Naturally, we used the similarity definition to formalize what we mean by concepts consistent with relations.

\begin{defn}[Consistence]
Given a semiotic $(S,M)$. A relation $D$ from $Lang(S)$ is
\emph{consistent with} $d\in M(S)$ if $d\models \forall D$, and it
is $\lambda$-\emph{consistent with} $d$ when $d\models_\lambda
\forall D$. The relation $D$ is consistent with part of $d$ if
$d\models \exists D$ and it is $\lambda$-\emph{consistent with a
part of} $d$ when $d\models_\lambda \exists D$.
\end{defn}

The set of hypotheses consistent with $d$ is denoted by
$Hy_{(S,M)}(d)$. For every $\lambda\in\Omega$, the set of hypotheses
$\lambda$-consistent with $d$ is denoted by
$\lambda$-$Hy_{(S,M)}(d)$. And, for a chain of truth values in $\Omega$
\[
\top\geq \lambda_0\geq \lambda_1\geq \ldots \geq \lambda_n,
\]
we have
\[
Hy\text{-}{(S,M)}(d)\subseteq \lambda_0\text{-}Hy_{(S,M)}(d)\subseteq  \lambda_1\text{-}Hy_{(S,M)}(d)\subseteq  \ldots \subseteq  \lambda_n\text{-}Hy_{(S,M)}(d).
\]

\begin{exam}[Description consistent with a dataset]
Let $(S,M)$ be a binary semiotic having by signs $A,B,C,D$ and let \[d:M(A)\times M(B)\times M(C)\times M(D)\rightarrow \Omega,\] be a finite crisp concept description, i.e. $d(\bar{x})=\top$ or $d(\bar{x})=\bot$ for every entity $\bar{x}$, and the number of entities $\bar{x}$ such that $d(\bar{x})=\top$ is finite. Then there is a word $D$ in the language, associated to the semiotic, consiste with $d$ called the dataset used to describe $d$.

Let be more specific, suppose that signs $A,B,C,D$ have the some interpretation, let $M(A)=M(B)=M(C)=M(D)=[0,1]$. And suppose that:
\begin{enumerate}
  \item $d(1.0,0.5,0.2,0.2)=\top$
  \item $d(1.0,1.0,0.2,0.2)=\top$, and
  \item $d(1.0,1.0,0.0,0.2)=\top$
\end{enumerate}
 are the only tuples true in relation $d$. This relation is consistent with the diagram
\[_{(A=1.0 \otimes B=0.5 \otimes C=0.2 \otimes D=0.2)\otimes (A=1.0 \otimes B=1.0 \otimes C=0.2 \otimes D=0.2)\otimes (A=1.0 \otimes B=1.0 \otimes C=0.0 \otimes D=0.2)}\]or $d$ is the answer to the query defined by the diagram,
usually represented using table notation by:
\begin{figure}[h]
\begin{center}
\begin{tabular}{|c|c|c|c|}
  \hline
  A & B & C & D \\
  \hline
  1.0 & 0.5 & 0.2 & 0.2 \\
  1.0 & 1.0 & 0.2 & 0.2 \\
  1.0 & 1.0 & 0.0 & 0.2 \\
  \hline
\end{tabular}
\end{center}
\caption{Dataset.}\label{dataset}
\end{figure}
\end{exam}

\section{Fuzzy computability}\label{fuzzy computability}
When the interpretation of a diagram is consistent with a multi-morphism we consider the multi-morphism computable in the semiotic. Formally:
\begin{defn}[Computability]
Given a semiotic $(S,M)$. A multi-morphism $f:A\rightharpoonup B$ is computable in $(S,M)$ if there is a diagram $D$ in $Lang(S)$:
\begin{enumerate}
  \item having $A$ as input, $A=i(D)$, $B=o(D)$ by output, and
  \item codify $F$, i.e. $f\models \forall D$.
\end{enumerate}
The multi-morphism $f$ is $\lambda$-computable in $(S,M)$ if $A=i(D)$, $B=o(D)$ and  $f\models_\lambda
\forall D$. These notions are very restrictive. We  relaxed them by calling to a diagram $D$ a specification to compute part of $f$ if
$d\models \exists D$. When the domain of the computable part of $f$ can be described by a diagram $D'$ we write \[f\models \forall_{D'} D.\]
\end{defn}

When $f\models \forall D$, with $A=i(D)$ and $B=o(D)$, we call diagram $D$ a \emph{program} or a \emph{specification}, in language $Lang(S)$, and its image by $M$ is an implementation for $f:A\rightharpoonup B$.

In this sense every, and only, interpretation of words from $Lang(S)$ are computable in the semiotic $(S,M)$. And, since words in $Lang(S)$ are generated from atomic componentes we have:
\begin{prop}
If $f$ and $g$ are computable multi-morphisms in the semiotic $(S,M)$ the $f\otimes g$ is also computable in $(S,M)$.
\end{prop}

And, since $Lang(S)$ is defined by finite diagrams, every finite diagram $D$ in $Set(\Omega)$, having by arrows computable multi-morphisms, has by limit a computable multi-morphism.

The set of interpretations of words from $Lang(S)$ and computable multi-morphisms define a category, denoted by $Hy_{(S,M)}$. In this category we write $f:d_1\rightarrow d_2$ if $f$ is a computable multi-morphism and $d_1$ and $d_2$ are consistent, descriptions in the semiotic, satisfying $d_1\otimes f=d_2$. Note what, if $D$ is consistent with $d_1$ and $D_f$ is the specification for $f$ then the diagram $D\otimes D_f$ is consistent with $d_1\otimes f$.

Generically, if $(d\Leftrightarrow MD)\geq \lambda$ and $(f\Leftrightarrow MD_f)=\top$ then $(d\otimes f \Leftrightarrow MD\otimes f)\geq \lambda$, i.e. $(d\otimes f \Leftrightarrow MD\otimes MD_f)\geq \lambda$. Formally:
\begin{prop}
 Let $d$ be a description $\lambda$-consistent with $D$ and $f$ a computable multi-morphism specified by $D_f$. Then $d\otimes f$ is a description  $\lambda$-consistent with diagram $D\otimes D_f$.
\end{prop}

In this sense a computable multi-morphism is known as a pre-processing tool in the data mining community. This allows the definition of $\lambda$-$Hy_{(S,M)}$, the category of concept $\lambda$-consistentes and  computable multi-morphisms in the semiotic $(S,M)$. Naturally, the limit and the colimit, in the usual sense, of finite diagrams in $\lambda$-$Hy_{(S,M)}$ define computable relations. We call this type of finite diagrams \emph{mining schemas}. And, given a mining schema $D$, in semiotic $(S,M)$, and a $\lambda$-consistent concept $d$, the limit $Lim\;D$ defines a computable multi-morphism and $d\otimes Lim\;D$ is a $\lambda$-consistent concept, interpreted as the output of schema $D$ when applied to concept $d$.

As usual we extend the notion of computability defining:
\begin{defn}[Turing computable]
A concept $d$ is called \emph{Turing computable} in the semiotic $(S,M)$ if there is a diagram $D$, possible infinite but enumerable, such that \[Lim\;D=d.\]
\end{defn}

Computability is usually associated with state-based systems. The interpretation, in a semiotic, of a stat must be time dependent. Given the presented static definition of sign interpretation we only catch the dynamic beaver using an ontological hierarchy. We see the possible interpretation of a sign as a class of structures used as possible instantiation for it during the system execution. We achieved this using a syntactic operator linking together signs in a same class representing different views for the same entity. The class of related signs using the syntactic operator must have the same generalization sign in the sign ontology. The existence of this type of syntactic operator, in a semiotic, defines what we called a syntactic operator in section \ref{synopt}.

\begin{defn}[Temporal semiotics]
A \emph{temporal semiotic} is a semiotic $(S,M)$, defined by a library $L:|L|\rightarrow (Chains\downarrow \Sigma^+)$, and having a syntactic operator \[t:\Sigma^+\rightarrow \Sigma^+\] such that:
\begin{enumerate}
  \item preserves polarization of signs, $t(s^+)=t(s)^+$, for every $s\in\Sigma$;
  \item preserves concatenation, $t(w_0.w_1)=t(w_0).t(w_1)$ for every pair of words $w_0,w_1$;
  \item preserves components functionality, if $f:w\rightarrow w'$, it must exists a component \[t(f):t(w)\rightarrow t(w').\]
\end{enumerate}
We imposed the existence of a component \[t(r):i(t(w))\rightarrow o(t(w))\] for every component $r:i(w)\rightarrow o(w)$, and an ontological hierarch for signs time invariant relating time dependent sings, i.e. if $s_1=t(s_0)$ then it must exist a sign $s$ such that $s_1\leq s$, $s_0\leq s$ and $s=t(s)$. In this sense, every sequence of time dependent signs \[s_0,t(s_0),t(t(s_0)),t(t(t(s_0))),\ldots\]  have by generalization the same sign $s$ on the ontology. We call $s$ a \emph{time invariant sign}.
\end{defn}

In a temporal semiotic $(S,M)$, if $r:i(w)\rightarrow o(s(w))$ is a component in the semiotic then its interpretation $M(r):M(i(w))\rightharpoonup M(o(s(w)))$ is called a \emph{coalgebra}. A sign $s\in\Sigma$ is \emph{time-invariant} in the semiotic if $M(s)=M(t(s))$.

A \emph{temporal logic semiotic} is a semiotic which is a logic semiotic and a temporal semiotic.

\begin{exam}[Fuzzy Turing machine]
A fuzzy Turing machine, with tape define using signs from $F$, can be defined as a word in language associated to a temporal logic semiotic $(S,M)$. And the interpretation for this word can be seen as an execution for it.  The machine structure can be codified in a sign system $S$ with library $L$ having by signs a set of machine stats, $Q$, and having by components the Turing machine instructions with labels in a set $I$.

Each of the instructions in $I$ has conditional form: it tells what to do, depending on whether the symbols distribution being scanned (the distribution of symbols in the scanned square). Namely, there are three classes of things that can be done:
\begin{enumerate}
  \item Print: Change signs distribution in place of whatever is in the scanned square;
  \item Move one square to the right;
  \item Move one square to the left;
\end{enumerate}
So depending on what instruction is being carried out and on what distribution is being scanned, the machine or its operator will perform one or another of these actions.

An instruction define a link between two stats and are codified as component labels with the following structure.
\begin{enumerate}
  \item $q_0[f] q_1$ if in stat $q_0$ the scanned distribution is changed using component $f$ interpretation and then change to stat $q_1$;
  \item $q_0[d_0:L]_\lambda q_1$ if in stat $q_0$ is reading a distribution $d$ and  $d\otimes M(d_0)\geq\lambda$ then move left and change to stat $q_1$;
  \item $q_0[d_0:R]_\lambda q_1$ if in stat $q_0$ is reading a distribution $d$ and  $d\otimes M(d_0)\geq\lambda$ then move right and change to stat $q_1$.
\end{enumerate} An instruction is executed if its condition is verified.

 In this sense a diagram in $Lang(L)$, defined using signs time invariant, is a \emph{Turing machine specification} with stats in $Q$ and tape signs from $F$.  Every refinement of a Turing machine specification in $Lang(L)$, defined using only time variant sings, is called a \emph{flow chart} and codifies a Turing machine execution. However to garante the correct interpretation of an instruction we have, for each stat $q_i\in Q$ in the sign system,  signs \[q_i^{(r)},q_i^{(m)},q_i^{(l)},q_i^{(hr)},q_i^{(tr)},q_i^{(hl)},q_i^{(tl)}\]
 where $q_i^{(r)}$ is interpreted as the tape right half, $q_i^{(m)}$ the reading square, $q_i^{(l)}$ is interpreted as the tape left half. And, for each tape halfs we select the right half head $q_i^{(hr)}$, the right tail head $q_i^{(hr)}$, the left half head $q_i^{(hl)}$ and the left tail head $q_i^{(ht)}$. The sign system structure sketch is defined such that the relation between this signs and $q_i$ are preserved if a model $M$ satisfies:
 \begin{enumerate}
   \item $M(q_i)=M(q_i^{(r)})\otimes_I M(q_i^{(m)})\otimes_I M(q_i^{(l)})$;
   \item $M(q_i^{(r)})=M(q_i^{(hr)})\otimes_I M(q_i^{(tr)})$;
   \item $M(q_i^{(l)})=M(q_i^{(hl)})\otimes_I M(q_i^{(tl)})$;
 \end{enumerate}
this interpretation for signs reflect the relations between I-projections (see example  \ref{def:indexproduct}) expressed in the following diagram:

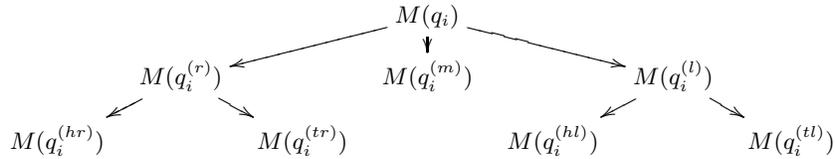
\begin{figure}[h]
\[
\small
\xymatrix @=7pt {
&&&M(q_i)\ar[lld]\ar[rrd]\ar[d]&&&\\
&M(q_i^{(r)})\ar[ld]\ar[rd]&&M(q_i^{(m)})&&M(q_i^{(l)})\ar[ld]\ar[rd]&\\
M(q_i^{(hr)})&&M(q_i^{(tr)})&&M(q_i^{(hl)})&&M(q_i^{(tl)})
}
\]
\caption{Sign interpretation structure.}\label{struct}
\end{figure}
And models of each instruction must satisfy the following conditions:
\begin{enumerate}
  \item For print instruction $q_0[f] q_1$ we should have;
\[
_{M(q_i [f] t(q_j))^\circ\otimes M(q_i)\otimes M(q_i [f] t(q_j))=M(t(q_j))
 \Rightarrow
\left\{
  \begin{array}{l}
    _{M(t(q_j)^{(r)})=M(q_i^{(r)})} \\
    _{M(t(q_j)^{(m)})=M(f)^\circ\otimes M(q_i^{(m)})\otimes M(f)}\\
    _{M(t(q_j)^{(l)})=M(q_i^{(l)})}\\
  \end{array}
\right.}
\]
  \item For instructions of type "Move one square to the left"  $q_0[d_0:L]_\lambda q_1$ we must have;
\[
_{\left\{
  \begin{array}{l}
    _{M(q_i^m)\otimes M(d_0)\geq\lambda}\\
    _{M(q_i [d_0:L]_\lambda t(q_j))^\circ\otimes M(q_i)\otimes M(q_i [d_0:L]_\lambda t(q_j))=M(t(q_j))} \\
  \end{array}
\right.
\Rightarrow
\left\{
  \begin{array}{l}
    _{M(t(q_j)^{(r)})=M(q_j^{(r)})\otimes_I M(q_i^{(m)})} \\
    _{M(t(q_j)^{(m)})=M(q_i^{(hl)})} \\
    _{M(t(q_j)^{(l)})=M(q_i^{(tl)})}\\
  \end{array}
\right.}
\]
  \item For instructions of type "Move one square to the right"  $q_0[d_0:R]_\lambda q_1$ we must have;
\[
_{\left\{
  \begin{array}{l}
    _{M(q_i^m)\otimes M(d_0)\geq\lambda}\\
    _{M(q_i [d_0:R]_\lambda t(q_j))^\circ\otimes M(q_i)\otimes M(q_i [d_0:R]_\lambda t(q_j))=M(t(q_j))} \\
  \end{array}
\right.
\Rightarrow
\left\{
  \begin{array}{l}
    _{M(t(q_j)^{(r)})=M(q_i^{(hr)})} \\
    _{M(t(q_j)^{(m)})=M(q_i^{(tr)})} \\
    _{M(t(q_j)^{(l)})=M(q_i^{(m)}) \otimes_I M(q_i^{(t)})}\\
  \end{array}
\right.}
\]
\end{enumerate}
So a model $M$ assigning to each stat a \emph{fuzzy tape} with signs in $F$, which can be seen as a infinite chain of indexed products ( see example \ref{def:indexproduct}):
\[
t=\underbrace{\overbrace{\cdots\otimes_I d_{5}}^{t^{(hr)}}\otimes_I \overbrace{d_{3}}^{t^{(tr)}}}_{t^{(r)}}\otimes_I \underbrace{d_{1}}_{t^{(m)}} \otimes_I \underbrace{\overbrace{d_{2}}^{t^{(hl)}}\otimes_I \overbrace{d_{4}\otimes_I\cdots}^{t^{(tl)}}}_{t^{(l)}}
\]
where we fixed a componente $t^{(m)}=d_{1}$ and such that each $d_i$ is a concept description $d_i:F\times I\rightarrow \Omega$. And, the model $M$ associates to each possible instruction (componente) a relation between fuzzy tapes $t_0$ and $t_1$, satisfying the described proprieties.

A fuzzy Turing machine begins its execution in a initial stat and it is a parallel device, at a given moment it can assume more than a stat. It finish its execution when it is stall in a stat or set of stats.
\end{exam}

\section{Consequence relation}\label{consequence relation}
In a semiotic $(S,M)$, we define for every relation $D$ in $Lang(S)$
the set of its $\lambda$-answers as:
\[ ans_\lambda(D)=\{ g \in  M(S): g\models_\lambda \forall_{D'} D\}\]
and it can be seen as the set of concepts $\lambda$-consistent with $D$ on the domain defined by $D'\in Lang(S)$.
\begin{exam}
The examples presented in this section are defined using a grid semiotic, having expressive power to codify structures in a grid, using a three truth-values logic $\Omega=\{\bot,\frac{1}{2},\top\}$.

Let $D$ be the diagram defining a relation between pais of entities in a grid, presented in fig. \ref{grid1}, where white points $\bar{x}$ mean $M(D)(\bar{x})=\bot$, gray points mean $M(D)(\bar{x})=\frac{1}{2}$ and darker points $\bar{x}$ mean $M(D)(\bar{x})=\top$.

\begin{figure}[h]
\begin{center}
\includegraphics[width=130pt]{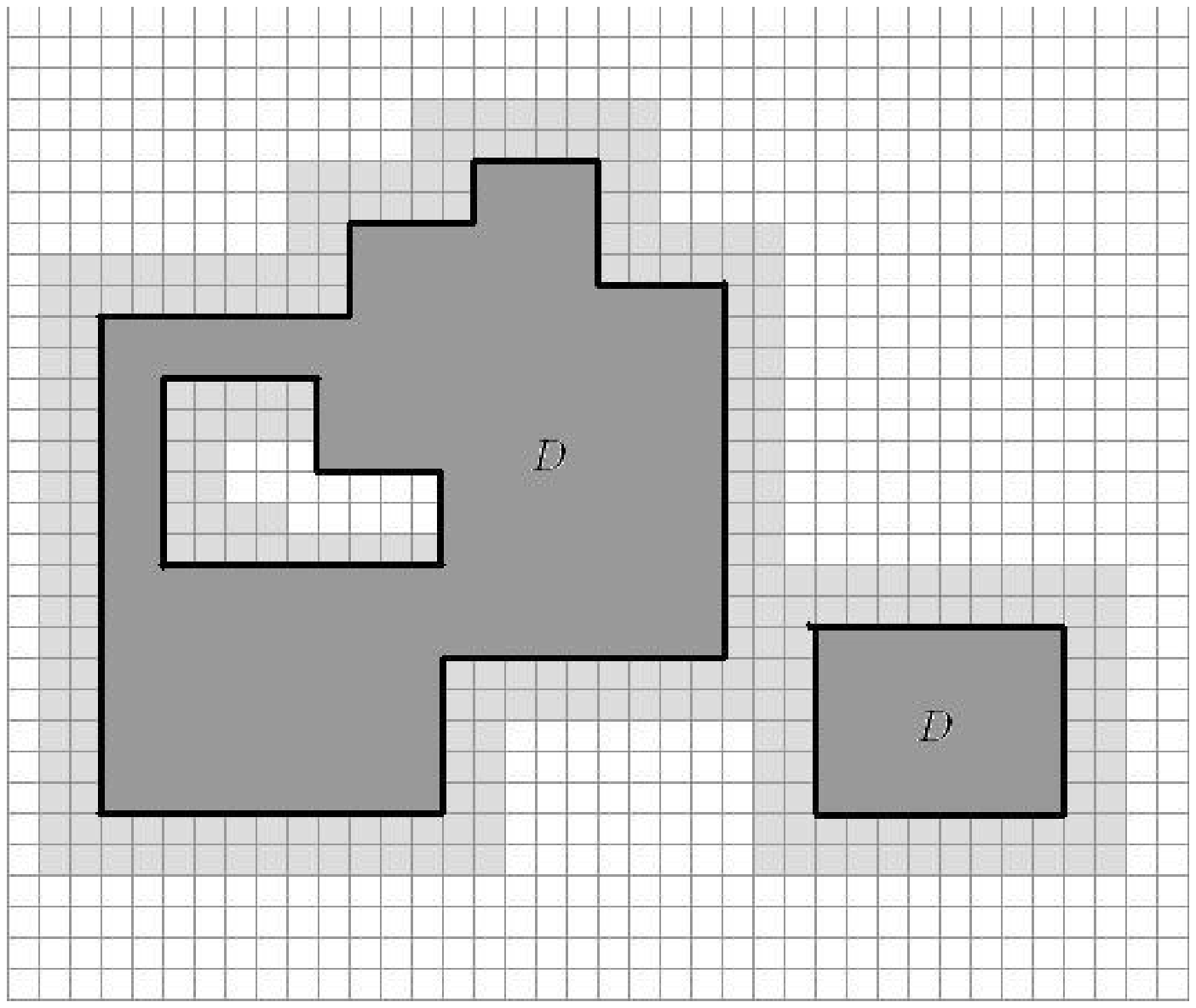} \hspace{1cm}
\includegraphics[width=130pt]{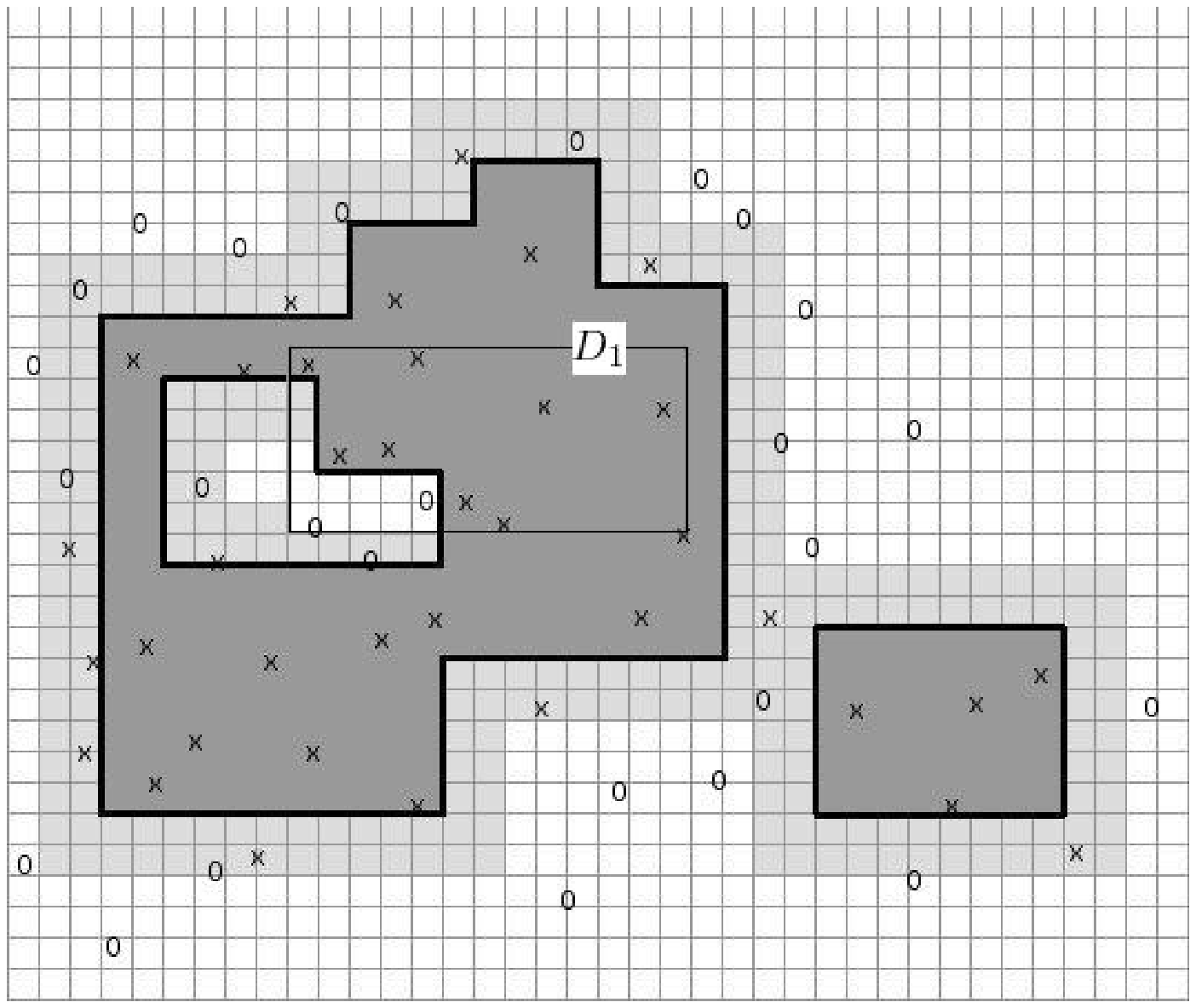}
\end{center}
\caption{Relation defined interpreting $D$ and finite relations $g_1$ where a point $\bar{x}$ marked with a X means $g_1(\bar{x})=\top$ and a point marked with 0 means $g_1(\bar{x})=\bot$.}\label{grid1}
\end{figure}

  The interior of the box, presented on the figure and labeled with $D_1$, can be seen as the set of points described by this diagrams. The relation $g_1$ presented we can be seen as an example satisfying
\[
g_1\models_{\frac{1}{2}} \forall D, g_1\models \forall_{D_1}D,
\]
which can be expressed writing
\[
g_1\in ans_{\frac{1}{2}}(D).
\]
\end{exam}

Note what, given $D$ the set $ans_\lambda(D)$ have at least an element, $M(D)\in ans_\lambda(D)$. Naturally:
\begin{thm}
If $D$ is a relation in $Lang(S)$ and $\lambda_0\leq\lambda_1$ then
\[ans_{\lambda_1}(D)\subseteq ans_{\lambda_0}(D).\]
\end{thm}
If $g\in ans_\lambda(D)$, with $g\models_\lambda \forall_{D'} D$, we express this relation by writing $g_{D'}\in ans_\lambda(D)$.

Let $D$ be a relation defined is a semiotic, by
  \[f\leq_Dg,\]
we mean that
  \[\text{if } M(D)(\bar{x})=\top \text{ then } f(\bar{x})\leq g(\bar{x}).\]
We use this relation and the operator $ans_\lambda$ to define two
modal operators, $\diamond_\lambda g$ and $\Box_\lambda g$, as the
weak and the strong images, respectively, for description $g\in
M(S)$ along the relation $\models_\lambda$:
\[\diamond_\lambda g=\{D\in Lang_R(S): (\exists f_{D'}\in ans_\lambda(D)) (g\leq_{D'} f) \} \]
\[\Box_\lambda g=\{D\in Lang_R(S): (\forall f_{D'}\in ans_\lambda(D))(f\leq_{D'} g) \} \]
Where $\diamond_\lambda g$ and $\Box_\lambda g$ can be seen,
respectively, as the set of models $\lambda$-consistent with parts
of $g$ and the set of models $\lambda$-consistent with $g$ in the
language $Lang(S)$.
\begin{exam}
For grid semiotic with three truth-values we presente in fig. \ref{grid2} two possible diagrams $D_1\in \diamond_{\frac{1}{2}} g_1$ and $D_2\in\Box_{\frac{1}{2}} g_2$.

\begin{figure}[h]
\begin{center}
\includegraphics[width=80pt]{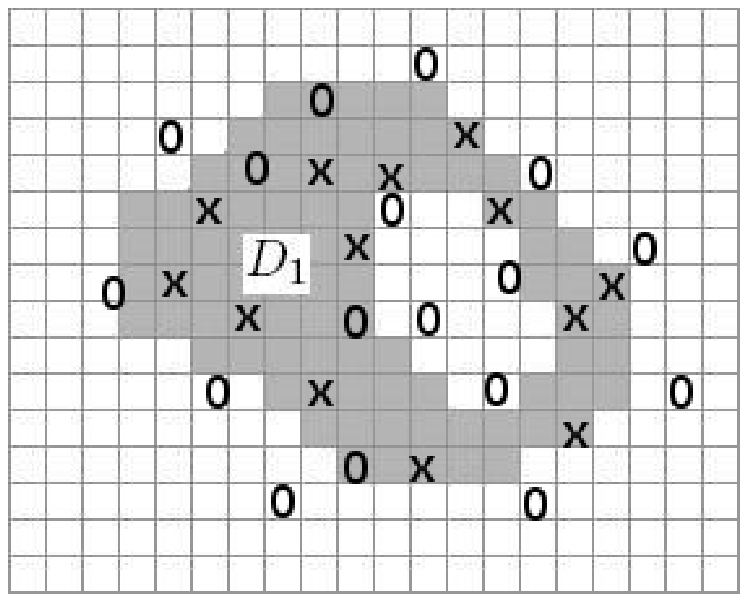} \hspace{1cm}
\includegraphics[width=80pt]{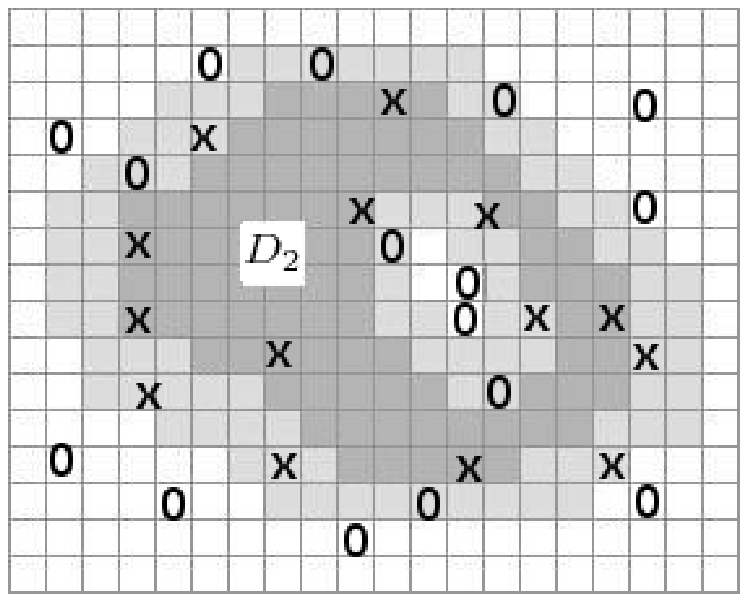}
\end{center}
\caption{$D_1\in \diamond_{\frac{1}{2}} g_1$ and $D_2\in\Box_{\frac{1}{2}} g_2$.}\label{grid2}
\end{figure}
\end{exam}

We have:

\begin{thm}
Given relations $D_0$ and $D_1$ in $Lang_R(S)$ and a description
$g$:
\begin{enumerate}
  \item if $\lambda_0\leq\lambda_1$, then $\Box_{\lambda_1} g\subseteq
  \Box_{\lambda_0} g$,
  \item if $\lambda_0\leq\lambda_1$, then $\diamond_{\lambda_1} g\subseteq
  \diamond_{\lambda_0} g$,
  \item if $D_0,D_1\in \Box_\lambda g$ then $D_0\vee D_1\in \Box_\lambda
  g$ and
  \item if $D_0,D_1\in \diamond_\lambda g$ then $D_0\wedge D_1\in \diamond_\lambda
  g$.
\end{enumerate}
\end{thm}

In the other direction we can extend $ans_\lambda$ to a set of
relations $U$ in $Lang(S)$:
\[ ans_\lambda(U)=\bigvee\{g \in  M(S): (\exists D \in \Box_\lambda g)(D\in U)\}\]
the greatest description $\lambda$-consistent with a model from $U$,
and let
\[ mod_\lambda(U)=\bigwedge\{g \in  M(S): (\forall D \in \Box_\lambda g)(D\in U)\}\]
be a description $\lambda$-consistent with every model existent in
$U$.

\begin{thm}
Let $U$ and $V$ be sets of relations in $S$. Then
\begin{enumerate}
  \item if $\lambda_0\leq\lambda_1$, $ans_{\lambda_0}(U)\geq
  ans_{\lambda_1}(U)$,
  \item $ans_\lambda(U\cup V)=ans_\lambda(U)\vee ans_\lambda(V)$,
  and
  \item $mod_\lambda(U\cup V)=mod_\lambda(U)\wedge mod_\lambda(V)$.
\end{enumerate}
And if $U\subseteq V$
\begin{enumerate}
  \item $ans_\lambda(U)\leq ans_\lambda(V)$ and
  \item $mod_\lambda(V)\leq mod_\lambda(U)$.
\end{enumerate}
\end{thm}

The $\lambda$-interior of concept $g$ in the semiotic system $(S,M)$
is defined as the greatest part of $g$ $\lambda$-consistent with a
model defined in the associated language and is given by:
\[ int_\lambda(g)=\bigvee\{h\in  M(S): (\exists D\in \Box_\lambda h)(\forall f_{D'}\in ans_\lambda(D))
(f\leq_{D'} g)\},\] and can be seen as the greatest fragment of
$g$ having a model $\lambda$-consistent in the language of the
semiotic. It is an \emph{interior operator} since;
\begin{enumerate}
  \item $int_\lambda(g)\leq g$,
  \item if $g\leq f$ then $int_\lambda(g)\leq int_\lambda(f)$ and
  \item $int_\lambda(g)= int_\lambda(int_\lambda(g))$.
\end{enumerate}
And given $\lambda_0\leq\lambda_1$, $int_{\lambda_1}(g)\leq
int_{\lambda_0}(g)$. A concept description $g$ is called
$\lambda$-\emph{open} in $(S,M)$ if
\[int_\lambda(g)= g.\]
For every set of relations $U$, $ans_\lambda(U)$ and
$mod_\lambda(U)$ are examples of $\lambda$-opens since:

\begin{prop}
In a semiotic for every set of relations $U$ and $\lambda\in\Omega$:
\begin{enumerate}
  \item $int_\lambda(ans_\lambda(U))=ans_\lambda(U)$, and
  \item $int_\lambda(mod_\lambda(U))=mod_\lambda(U)$.
\end{enumerate}
\end{prop}

More precisely:

\begin{thm}
In a semiotic for every set of relations $U$,
\[ans_\top(U)\models\bigvee U \text{ and }mod_\top(U)\models\bigwedge U.\]
\end{thm}

The closure of concept $g$ in the semiotic system $(S,M)$ is defined
as the shorter cover of $g$ codified in the language $L(S)$  and
is given by:
\[ cl_\lambda(g)=\bigwedge\{h\in M(S): (\forall D\in \Box_\lambda h)(\exists f_{D'}\in ans_\lambda(D))(g \leq_{D'} f)\}\]
and can be seen as the shortest cover containing $g$ and codified in
the language associated to the semiotic. It is a \emph{closure
operator} since;
\begin{enumerate}
  \item $g\leq cl_\lambda(g)$,
  \item if $g\leq f$ then $cl_\lambda(g)\leq cl_\lambda(f)$, and
  \item $cl_\lambda(cl_\lambda(g))= cl_\lambda(g)$.
\end{enumerate}
And given $\lambda_0\leq\lambda_1$, $cl_{\lambda_1}(g)\leq
cl_{\lambda_0}(g)$. Trivially we have:

\begin{prop}
Given a semiotic $(S,M)$, for every $g\in M(S)$,
\[int_\lambda(g)\leq g \leq cl_\lambda(g).\]
\end{prop}

A concept description $g$ is called $\lambda$-\emph{close} in
$(S,M)$ if
$cl_\lambda(g)= g$.
Descriptions $ans_\lambda(U)$ and $mod_\lambda(U)$ are also
$\lambda$-closed concepts. This can be extended to every
$\lambda$-open description:

\begin{prop}
Given a semiotic $(S,M)$, for every $g\in M(S)$, $g$ is
$\lambda$-closed iff it is $\lambda$-open.
\end{prop}

In this sense when a description is $\lambda$\emph{-open} or
$\lambda$\emph{-close} we called it a description $\lambda$-representable
on the semiotic. By this we mean that:

\begin{prop}
Let $g$ be a description in the semiotic $(S,M)$. Exists a relation
$D$, such that $g\models_\lambda D$, iff $g$ is $\lambda$-open or
$\lambda$-close.
\end{prop}

Because of the symmetry between the left and the right side of
$d\models D$, from the above definitions we have
\[int_\lambda=ans_\lambda\Box_\lambda \text { and }
cl_\lambda=mod_\lambda\diamond_\lambda\] and they also have
symmetric definitions, obtained by replacing each operator with its
symmetric:
\[\A_\lambda=\Box_\lambda ans_\lambda \text{ and } \C_\lambda=\diamond_\lambda mod_\lambda.\]
 By symmetry it is immediate that $\C_\lambda$ is an interior operator and $\A_\lambda$ is a
closure operator.

Spelling out the definition of $\A_\lambda$, for every set of
relations $U$,
\[\A_\lambda (U)=\{D\in Lang(S): (\forall f_{D'}\in ans_\lambda(D))(f\leq_{D'}
ans_\lambda (U)) \} ,\] i.e. all $\lambda$-answers for $D$ are
$\lambda$-codified using relation in $U$. And we have:

\begin{thm}\label{soundness}
For every pair $U$ and $V$ of relations in the semiotic $(S,M)$,
\begin{enumerate}
  \item $\A_\lambda(U\cup V)\supseteq \A_\lambda(U)\cup \A_\lambda(V)$,
  \item if $U\subseteq V$, $\A_\lambda(U)\subseteq \A_\lambda(V)$,
  \item if $\lambda_0\leq\lambda_1$, $\A_{\lambda_1}(U)\subseteq
  \A_{\lambda_0}(U)$,
  \item if $D\in \A_\lambda(U)$ then $\bigvee ans_\lambda(D)\leq ans_\lambda(U)$, and
  \item if $D\in \A_\lambda(U)$ then $ans_\lambda(U)\models_\lambda D$.
\end{enumerate}
\end{thm}

Spelling out operator $\C$ we have
\[\C_\lambda (U)=\{D\in Lang(S): (\exists f_{D'}\in ans_\lambda(D)) (mod_\lambda(U)\leq_{D'} f) \} ,\]  by
$D\in \C_\lambda (U)$ we mean that $D$ have an $\lambda$-answer and
every $\lambda$-codification for it are in $U$. In this case we may
proof:

\begin{thm}
For every pair $U$ and $V$ of relations in the semiotic $(S,M)$,
\[\C_\lambda(U\cup V)\subseteq \C_\lambda(U)\cap \C_\lambda(V)\]
when $U\subseteq V$, $\C_\lambda(U)\subseteq \C_\lambda(V)$.
\end{thm}

Lets write
\[ U\vdash_\lambda D \text{ iff } D\in \A_\lambda(U), \]
and since $\A_\lambda$ is a closure operator:
\begin{thm} In a semiotic $(S,M)$, for every $\lambda$, we have:
    \begin{enumerate}
        \item if $D\in U$ then $U\vdash_\lambda D$(Inclusion),
        \item if $U\vdash_\lambda D$ then $U\cup V\vdash_\lambda
        D$ (Monotony), and
        \item if $V\vdash_\lambda D$ and $U\cup \{D\}\vdash_\lambda
        D'$, then $U\cup V\vdash_\lambda D'$ (Cut).
    \end{enumerate}
\end{thm}

By this we mean that $(Lang_R(S),\vdash_\lambda)$ is a
\emph{inference system} \cite{abramsky92}  for every $\lambda$.

\begin{exam}
Given interpretations, presented in fig. \ref{grid4}, for three diagrams $D_0$, $D_1$ and $D_2$ in the grid semiotic with three valued logic:

\begin{figure}[h]
\begin{center}
\includegraphics[width=150pt]{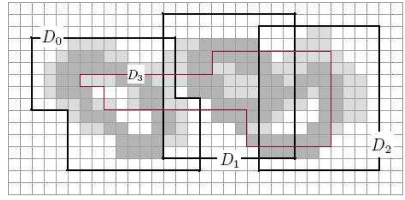}
\end{center}
\caption{Interpretations for diagrams $D_0$, $D_1$, $D_2$ and $D_3$.}\label{grid4}
\end{figure}

The diagram $D_3$, with the represented interpretation, can be see as the result of applying inference to the set of diagrams $\{D_0,D_1,D_2\}$, symbolically we write
\[D_0,D_1,D_2\vdash D_3.\]
\end{exam}

Since $\A_{\lambda_0}\leq\A_{\lambda_1}$, if
$\lambda_0\geq\lambda_1$, we have:

\[
\infer{U\vdash_{\lambda_1} D_0}{U\vdash_{\lambda_0} D_0 &
\lambda_0\geq\lambda_1}
\]

And we have using definition of $\lambda$-consistence:

\begin{thm} Let $g\models_{\lambda_0}\forall_{D'_0} D_0$ and $f\models_{\lambda_1}\forall_{D'_1} D_1$:
\begin{enumerate}
  \item $g\wedge f\models_{\lambda_0\wedge\lambda_1}\forall_{D'_0\otimes D'_1} D_0\wedge D_1$,
  \item $g\vee f\models_{\lambda_0\vee\lambda_1}\forall_{D'_0\otimes D'_1} D_0\vee D_1$,
  \item $g\otimes f\models_{\lambda_0\otimes\lambda_1}\forall_{D'_0\otimes D'_1} D_0\otimes D_1$, and
  \item $g\Rightarrow f\models_{\lambda_0\Rightarrow\lambda_1}\forall_{D'_0\otimes D'_1} D_0\Rightarrow D_1$.
\end{enumerate}
\end{thm}

Using this properties on the definition of $\lambda$-answer we have:

\begin{thm} On a semiotic we have:
\begin{enumerate}
  \item For every diagram $D\in Lang_R(S)$:
   \begin{enumerate}
     \item $ans_{\lambda_0\vee\lambda_1}(D)=ans_{\lambda_0}(D)\vee
  ans_{\lambda_1}(D)$,
     \item $ans_{\lambda_0\wedge\lambda_1}(D)=ans_{\lambda_0}(D)\wedge
  ans_{\lambda_1}(D)$ and
     \item $ans_{\lambda_0\otimes\lambda_1}(D)=ans_{\lambda_0}(D)\otimes
  ans_{\lambda_1}(D)$;
   \end{enumerate}

  \item For every concept description $g\in M(S)$:
  \begin{enumerate}
    \item $\Box_{\lambda_0\vee\lambda_1}(g)=\Box_{\lambda_0}(g)\vee
  \Box_{\lambda_1}(g)$,
    \item $\Box_{\lambda_0\wedge\lambda_1}(g)=\Box_{\lambda_0}(g)\wedge
  \Box_{\lambda_1}(g)$ and
    \item $\Box_{\lambda_0\otimes\lambda_1}(g)=\Box_{\lambda_0}(g)\otimes
  \Box_{\lambda_1}(g)$;
  \end{enumerate}

  \item For every set of diagrams $U\subset Lang_R(S)$:
  \begin{enumerate}
    \item $ans_{\lambda_0\vee\lambda_1}(U)=ans_{\lambda_0}(U)\vee
  ans_{\lambda_1}(U)$,
    \item $ans_{\lambda_0\wedge\lambda_1}(U)=ans_{\lambda_0}(U)\wedge
  ans_{\lambda_1}(U)$ and
    \item $ans_{\lambda_0\otimes\lambda_1}(U)=ans_{\lambda_0}(U)\otimes
  ans_{\lambda_1}(U)$.
  \end{enumerate}

\end{enumerate}
\end{thm}

Which gives support to the definition of the introduction rules:
\[
\infer{U\vdash_{\lambda_0\vee\lambda_1} D_0\vee D_1,
U\vdash_{\lambda_0\wedge\lambda_1} D_0\wedge D_1,
U\vdash_{\lambda_0\otimes\lambda_1} D_0\otimes
D_1}{U\vdash_{\lambda_0}D_0 &  U\vdash_{\lambda_1}D_1}
\]

The fact of, if $D_0\wedge D_1\in \A_\lambda(U)$ then $D_0\in
\A_\lambda(U)$ and $D_1\in \A_\lambda(U)$, can be expressed by the elimination rule:
\[
\infer{U\vdash_\lambda D_0,\,U\vdash_\lambda D_1}{U\vdash_\lambda D_0\wedge D_1}
\]

Naturally, in a divisible logic, we have
\[
\infer{U\vdash_{\lambda_0\otimes\lambda_1} D_0\wedge D_1}{U\vdash_{\lambda_0}
D_0 &  U\vdash_{\lambda_1}D_0\Rightarrow D_1}
\]
since if $g\models_{\lambda_0} \exists D_0$ and $g\models_{\lambda_1} \exists (D_0\Rightarrow D_1)$ then $g\models_{\lambda_0\otimes\lambda_1} \exists D_1$. Because, if $D_0\in\A_{\lambda_0}(U)$ and $D_0\Rightarrow
D_1\in\A_{\lambda_1}(U)$ then
$D_0\wedge D_1\in\A_{\lambda_0\otimes\lambda_1}(U)$. By this we mean what for
every $f\in ans_{\lambda_0}(D_0)$, $f\leq ans_{\lambda_0}(U)$, and
for every $h\in ans_{\lambda_0}(D_1)$, $f\Rightarrow h\in
ans_{\lambda_1}(D_0\Rightarrow D_1)$, and $f\otimes(f\Rightarrow
h)\in ans_{\lambda_0\otimes\lambda_1}(U)$. Note that, in a divisible ML-algebra, $f\otimes(f\Rightarrow h)\leq f\wedge h$. Then $f\wedge h\in ans_{\lambda_0\otimes\lambda_1}(U)$.

A diagram $D$ codifies all the information existent in a concept $d$, using the syntax associated to semiotic $(S,M)$, if for every diagram $D_1$ such that $d\models_\lambda \forall D_1$, $M(D\vee D_1)=M(D)$. This diagrams are called \emph{total} an can be defined by \[D=\bigvee_{d\models_\lambda \forall D_i} D_i.\]

In the category $Lang(S)$ having by objects diagrams codifying relations and where a diagram $D$ is a morphism from relation $D_0$ to relation $D_1$ if $D_0\otimes D =D_1$, we consider composition as the operations of diagram gluing. The consequence operator $\vdash_\lambda$ can be seen as a functor:
\[
\vdash_\lambda:Lang(S)\rightarrow Lang(S).
\]
A diagram $D$ is called a \emph{theory} in the $\lambda$-semiotic $(S,M)$ if it is a fixed-point for consequence operator \[\vdash_\lambda(D)=D.\]

The semiotic model $M$, can be interpreted as a functor \[M:Lang(S)\rightarrow\lambda\text{-}Hy_{(S,M)},\]
in the category of concepts $\lambda$-consistentes and computable multi-morphisms. A functor in the opposite direction can be defined using  the operator of consistence
\[\models_\lambda:\lambda\text{-}Hy_{(S,M)}\rightarrow Lang(S),\]
assigning to each $\lambda$-consistente description a total diagram with its codification on the semiotic.

Since $M(D)\models_\lambda\forall D$ we have
\[M\circ \models_\lambda = id,\]
and by definition of consequence relation
\[\models_\lambda\circ M = \vdash_\lambda.\]

If $D$ is a $\lambda$-theory in the semiotic, $M(D)$ is the model $\lambda$-consistent with this theory.

\section{Integration}\label{integration}
Our aim is to construct an integration semiotic base from several separated semiotics. This need can arise, for example, when knowledge bases are acquired independently from interactions with several domain experts. A similar problem can also arise whenever separated knowledge bases are generated by learning algorithms.  The objective of integration is then construct one system that exploits all the knowledge that is available and has a good performance, i.e. a good degree of consistence with the data.

We must differentiate between two types of integration: semiotics integration and integration of models in a semiotic. The semiotic integration goal is the definition of a semiotic integrating the sintaxe and semantic of a given family of semiotics. By the integration of models in a semiotic we mean the possibility of improve the description of concepts integrating models for it using diferente data or diferente views of the same data. The integration of models is defined by an integration schema describing the relations between different models in the same semiotic. In the semiotic integrating we integrate different logics in the same semiotic associated to different languages used by domain experts or associated to structure specification language. In both senses Knowledge integration, in conjunction with inference, can play an important rule in the process of knowledge acquisition.

We impose an important restriction to the semiotic integration: Given a family of semiotics $(S_i,M_i)_I$ its integration is defined, if and only if, equal signs and components with the same label in different semiotics have the some interpretation, with only a possible exception, the interpretations of sign $\Omega$ associated to the semiotics logic and its operators may be different.

The integration of semiotics $(S_i,M_i)_I$ is denoted by $(\bigcup_IS_i,\bigcup_IM_i)$ and it is given by the sign system
\[\bigcup_IS_i=(\bigcup_IL_i,\bigcup_I\E_i,\bigcup_I\U_i,\bigcup_Ico\U_i),\]
 if, for each $i\in I$, the semiotic $S_i$ is defined by the structure $(L_i,\E_i,\U_i,co\U_i)$. Where $\bigcup_IL_i$ is the library defined by the union of libraries $(L_i:|L_i|\rightarrow (Chains\downarrow \Sigma_i^+))_I$ associated to each semiotic. This library is given by
\[\bigcup_IL_i:\cup_I|L_i|\rightarrow (Chains\downarrow \cup_I\Sigma_i^+),
\] having by signs the union of ontology $\bigcup_I\Sigma_i^+$ defined by the signs existent in each library, and having by component labels the union of labels existent in both libraries. Note that the integration of libraries must preserve component functionalities. In this sense, the union of libraries is only defined if the component existente in different libraries, with equal label, have the some functionalities. The graphic language associate to $\bigcup_IL_i$ is denoted by $Lang(\bigcup_IL_i)$, and we have $\bigcup_I\E_i\subset Lang(\bigcup_IL_i)$.

From the description for the language associated to $\bigcup_IL_i$ recall what: Given two graphs $G_0$ and $G_1$ we define $G_0\bigcup G_1$ as the graph defined having by vertices the vertices of $G_0$ and $G_1$ and having by arrows the arrows of $G_0$ and $G_1$. If each library $L_i$ have associated multi-graphs $\G(L_i)$, we have \[\G(\bigcup_IL_i)=\bigcup_I\G(L_i).\]
Then, if we have models  $(M_i:\G(L_i^\ast))\rightarrow Set(\Omega))_I$ for different libraries, the homomorphism  $\bigcup_IM_i$ is a model for the sign system $\bigcup_IS_i$, \[\bigcup_IM_i:\G((\bigcup_IL_i)^\ast)\rightarrow Set(\Omega)\] constructed using the union of models $(M_i:\G(L_i^\ast))\rightarrow Set(\Omega))_I$, making for nodes $\bigcup_IM_i(v)=M_i(v)$ if $v\in \G(L_i^\ast)$ and $v\neq\Omega$, for some $i\in I$, and $\bigcup_IM_i(v)=M_i(v)$ for multi-arrows $f:w\rightarrow w'\in \G(L_i^\ast)$, and $w'\neq \Omega$ for some $i\in I$. By this we mean that not logic signs and multi-arrows which not represent relations are interpreted as it where in its libraries. This definition only makes sense when equal signs and equal labels have equal interpretations in different libraries.

For the logic family of logic signs $(\Omega_i)_I$ associated to the family of logic semiotics $(S_i)_I$ we define the sign $\prod_I\Omega_i$ interpreted by $\bigcup_IM_i$ as the ML-algebra product $\prod_IM_i(\Omega_i)$. The interpretation of sign $\prod_I\Omega_i$ is a ML-algebra and for every relation $r:w\rightarrow \Omega_i$ existent in each semiotic $S_i$ its interpretation by $\bigcup_IM_i$ is the relation
\[\bigcup_IM_i(r:w\rightarrow \Omega_i)=M_i(r)\otimes \pi^\top_{\Omega_i}\]
where $\pi^\top_{\Omega_i}:\Omega_i\rightarrow \prod_I\Omega_i$ is the map such that \[\pi^\top_{\Omega_i}(\alpha)=(\top,\ldots,\top,\alpha,\top,\ldots,\top),\] having different of $\top$ only the component of order $i$. Formally,

\begin{prop}
If \[S_0=(L_0,\E_0,\U_0,co\U_0),S_1=(L_1,\E_1,\U_1,co\U_1),\ldots,S_n=(L_n,\E_n,\U_n,co\U_n)\] are sign systems with models $M_0,M_1,\ldots,M_n$ then
\[\bigcup_IM_i:\G((\bigcup_IL_i)^\ast)\rightarrow Set(\Omega)\] defined as above is a model for the sign system  \[\bigcup_IS_i=(\bigcup_IL_i,\bigcup_I\E_i,\bigcup_I\U_i,\bigcup_Ico\U_i).\]
\end{prop}

Since, for models $M_0,M_1,\ldots,M_n$ of sign system $S_0,S_1,\ldots,S_n$, we have, by definition \ref{Modelspec}, for every $j=1,\ldots,n$:
\begin{enumerate}
  \item if $D\in\E_j$, then $\bigcup_IM_i(D)=M_j(D)$ is a total multi-morphism;
  \item if $(s,D,i(D),o(D))\in\U_j$, then $\bigcup_IM_i(s)= M_j(s)$ is the $\Omega$-set defined by $Lim\;MD$;
  \item if $(s,D,i(D),o(D))\in co\U_j$, then $\bigcup_IM_i(s)= M_j(s)$ is the $\Omega$-set defined by $coLim\;MD$.
\end{enumerate}

Naturally, the resulting semiotic of an integration process have the syntax and the semantic generated by the syntax and semantic associated to the semiotics. The some principle can be seen for some syntactic operators. The integration of a family of semiotics, where at last one is a  differential semiotics, is a diferencial semiotic and the same happens for temporal semiotics. If $(S_i)_I$ is a family where $(S_j)_J$ is a subfamily of temporal semiotics, given by syntactic operator $t_i:\Sigma^+_i\rightarrow\Sigma^+_i$. Then the integration $\bigcup_IS_i$ is a temporal semiotic where the syntactic operator $t:\bigcup_I\Sigma^+_i\rightarrow\bigcup_I\Sigma^+_i$ is defined by making:
\begin{enumerate}
  \item $t(s)=t_j(s)$ if $s\in \Sigma^+_j$ where $j\in J$, and
  \item $t(s)=s$ if $s\not\in \bigcup_J\Sigma^+_j$.
\end{enumerate}

An integration schema is a diagram \[\J:\G\rightarrow Set(\Omega),\] defined on the category of interpretations and computable multi-morphisms, such that \[\J(i)=MD_i\] for every vertices $i$ in $\J$. Let $(D_i)$ be a family of diagrams used on integration schema $\J$ definition. The concept description  $\Omega(\J)$ defined by $\J$ is given by the colimit of $\J$
\[\Omega(\J)=colim_i\;MD_i=colim_i\;\;Lim\;M(D_i),\]
where $colim$ is computed as defined in \ref{colim}, i.e.
\[
_{(coLim\;\J)(\ldots,\bar{x}_i,\ldots,\bar{x}_j,\ldots)=\ldots\otimes Lim\;M(D_i)(\bar{x}_i)\otimes\ldots\otimes Lim\;M(D_j)(\bar{x}_j)\otimes\ldots\otimes\bigvee_{f:MD_i\rightarrow MD_j\in \J} f(\bar{x}_i,\bar{x}_j).}
\]

\section{Reasoning about models of concepts}\label{reasoning}

The language $\lambda$-RL$(S)$ of $\lambda$-representable logic is a
formalism to speak of structures $\lambda$-representable on a
semiotic $(S,M)$. It is basically a classic string-based modal logic
defined by a generative grammar where propositional variables are
interpreted as diagrams belonging to the language associated to the
sign system $S$.

$\lambda$-RL$(S)$ is constructed from relations in $Lang(S)$, modal
operators limit, closure, interior and the lifting of the monoidal logic
connectives  $\otimes$, $\Rightarrow$, $\wedge$ and $\vee$ to relations.

Every semiotic $(S,M)$ defines a sematic for $\lambda$-RL(S) by the
truth-relation \[g\models_\lambda \varphi,\] given, for every formula
$\varphi\in\lambda$-RL(S) and every concept description $g$ in
$(S,M)$, as follows:
\begin{enumerate}
  \item $g\models_\lambda \varphi$ iff $\varphi$ is the diagram $D$ and $\Gamma(g,MD)\geq\lambda$,
 \item $g\models_\lambda [I]\varphi$ iff $int(g)\models_\lambda \varphi$,
  \item $g\models_\lambda [C]\varphi$ iff $cl(g)\models_\lambda
  \varphi$.
\end{enumerate}
And given formulas
$\varphi_0$ and $\varphi_1$ in $\lambda$-RL(S), if
\[g\models_{\lambda_0} \varphi_0\text{ and }g\models_{\lambda_1} \varphi_1\]
we have:
\begin{enumerate}
  \item $g\models_{\lambda_0\otimes\lambda_1} (\varphi_0\otimes \varphi_1)$,
  \item $g\models_{\lambda_0\Rightarrow\lambda_1} (\varphi_0\Rightarrow \varphi_1)$,
  \item $g\models_{\lambda_0\wedge\lambda_1} (\varphi_0\wedge \varphi_1)$ and
  \item $g\models_{\lambda_0\vee\lambda_1} (\varphi_0\vee \varphi_1)$.
\end{enumerate}
Using the structural compatibility between multi-morphism composition
and diagram gluing we have:

\begin{prop}
Given multi-morphism $g$ and $h$ such that
\[g\models_{\lambda_0} \varphi_0\text{ and }h\models_{\lambda_1}
\varphi_1\] we have: \[g\otimes
h\models_{\lambda_0\otimes\lambda_1} \varphi_0 \otimes \varphi_1\]
\end{prop}

By the lifting of the ML-algebra structure to the set of concept
descriptions we have:
\begin{prop}
Given concept descriptions $g$ and $h$ in $\oplus_iA_i$ such that
$g\models_{\lambda_0} \varphi$ and $h\models_{\lambda_1} \varphi$ we
have:
\begin{enumerate}
  \item $(g\otimes h)\models_{\lambda_0\otimes\lambda_1} \varphi$,
  \item $(g\Rightarrow h)\models_{\lambda_0\Rightarrow\lambda_1} \varphi$,
  \item $(g\wedge h)\models_{\lambda_0\wedge\lambda_1} \varphi$
  and
  \item $(g\vee h)\models_{\lambda_0\vee\lambda_1} \varphi$.
\end{enumerate}
\end{prop}

Given a set of relations $U$ from $Lan_R(S)$ and $\varphi$ a
relation in $\lambda$-RL$(S)$ we define:
\[U\models_\lambda \varphi \text{ iff } ans_\lambda(U)\models_\lambda \varphi\]
Using theorem \ref{soundness} we have:
\begin{thm}[Soundness]
Given a set of relations $U$ from $Lang_R(S)$ and $\varphi$ a
relation in $\lambda$-RL$(S)$,
\[\text{if }  U\vdash_\lambda
\varphi\text{ then } U\models_\lambda \varphi \] for every
$\lambda$.
\end{thm}
Naturally, the completeness isn't valid, if $U\models_\lambda
\varphi$, we may not prove using deduction $U\vdash_\lambda
\varphi$.

\section{Conclusions and future work}\label{conclusions}

The use of semiotics seems to be the appropriate formalism for defining syntax and the meaning of graphic language. Particularly when this languages are based on a library of functional components interpreted as relations evaluated in a multi-valued logic. This approach makes simplifies the integration of knowledge expressed using different languages and allowing the ingerence of new knowledge.


\bibliographystyle{amsplain}
\bibliography{multibib}
\end{document}